\newif\iflong
\newif\ifshort

\longtrue

\iflong
\else
\shorttrue
\fi
\documentclass[11pt,fleqn]{article}
\usepackage[utf8]{inputenc}

\usepackage{amsmath,amssymb,graphicx,amsthm}

\usepackage{fullpage}

\usepackage{booktabs}

\usepackage[numbers]{natbib}

\usepackage[pagebackref]{hyperref}

\usepackage{xifthen}

\usepackage{tikz}
\usetikzlibrary{decorations,arrows,petri,topaths,backgrounds,shapes,positioning,fit,calc,decorations.pathreplacing,patterns,intersections,decorations.pathmorphing,matrix}
\def \xsc {4}
\def \dist {.9}

\usepackage{float}
\usepackage{paracol}
\usepackage{mathtools}
\usepackage{dsfont}

\usepackage{multirow}

   \renewcommand{\citet}[1]{\citeauthor{#1}~\cite{#1}}
  
\makeatletter
\newcommand{\gettikzxy}[3]{%
  \tikz@scan@one@point\pgfutil@firstofone#1\relax
  \edef#2{\the\pgf@x}%
  \edef#3{\the\pgf@y}%
}
\makeatother

\tikzstyle{blueline} = [thick, blue, dotted]
\tikzstyle{redline} = [thick, red, dashed]
\tikzstyle{greenline} = [ultra thick, darkgreen, decorate,decoration=snake]
\tikzstyle{blackline} = [thick, black]
\tikzstyle{orangeline} = [thick, orange,  dash pattern={on 7pt off 2pt on 1pt off 3pt}]
\tikzstyle{concept} = [minimum height=4ex, inner sep=1pt, text centered, align=center]

\tikzstyle{agent} = [draw, circle, fill=black, minimum size=1.4ex, inner sep=1pt, text centered, align=center]
\pgfdeclarelayer{background}
\pgfdeclarelayer{foreground}

\DeclareMathOperator{\argmin}{argmin}
\newcommand{\rank}{\mathsf{rk}}

\usepackage[linecolor=green!70!black, backgroundcolor=green!10, bordercolor=black, textsize=small]{todonotes}

\usepackage{xcolor}
\definecolor{dargray}{rgb}{0.18, 0.18, 0.18}
\definecolor{darkgreen}{rgb}{0.01,0.6,0.1}
\definecolor{lightrose}{rgb}{0.996,0.75,0.793}
\definecolor{rose}{cmyk}{0.75, 0.75, 0,0}
\definecolor{winered}{rgb}{0.6,0.1,0.1}
\definecolor{lightyellow}{rgb}{1, 1, 0.6}
\definecolor{transparent}{rgb}{1,1,1}
\definecolor{lightlightgray}{rgb}{0.88, 0.88, 0.88}
\definecolor{lightgray}{rgb}{0.8, 0.8, 0.8}
\definecolor{lightblue}{rgb}{0.527,0.805,0.977}
\definecolor{lightgreen}{rgb}{.74,1,0}
\definecolor{darkblue}{rgb}{0,0,0.4}

\usepackage{paralist}

\newcommand{\pGNSPMl}{\probname{Globally Nearly Stable Perfect Matching}}
\newcommand{\pGNSPM}{\probname{Global-Near+Perf}}
\newcommand{\pLNSPMl}{\probname{Locally Nearly Stable Perfect Matching}}
\newcommand{\pLNSPM}{\probname{Local-Near+Perf}}

\newcommand{\pGLNSPMl}{\probname{Globally ({\normalfont or} Locally) Nearly Stable Perfect Matching}}

\newcommand{\pLNSEMl}{\probname{Locally Nearly Stable Egalitarian Matching}}
\newcommand{\pLNSEM}{\probname{Local-Near+Egal}}

\newcommand{\pGNSEMl}{\probname{Globally Nearly Stable Egalitarian Matching}}
\newcommand{\pGNSEM}{\probname{Global-Near+Egal}}

\newcommand{\pGLNSEMl}{\probname{Globally ({\normalfont or} Locally) Nearly Stable Egalitarian Matching}}

\newcommand{\nspp}[2]{#1 #2-nearly stable\xspace}
\newcommand{\nsp}[1]{#1 nearly stable\xspace}
\newcommand{\gns}[1]{%
  \nspp{globally}{#1}%
}

\newcommand{\lns}[1]{%
  \nspp{locally}{#1}%
}

\newcommand{\gnsnopa}[1]{\nsp{globally}}
\newcommand{\lnsnopa}[1]{\nsp{locally}}

\newcommand{\GNely}{Globally Nearly\xspace}
\newcommand{\LNely}{Locally Nearly\xspace}

\newcommand{\nstabilityp}[2]{%
  #1 #2-near stability\xspace}

\newcommand{\gnstability}[1]{%
 \nstabilityp{global}{#1}}

\newcommand{\lnstability}[1]{%
  \nstabilityp{local}{#1}%
}

\newcommand{\Nstability}{Near stability\xspace}
\newcommand{\Nstable}{Nearly stable\xspace}
\newcommand{\NStable}{Nearly Stable\xspace}

\newcommand{\nstability}{near stability\xspace}
\newcommand{\nstable}{nearly stable\xspace}

\newcommand{\gnstabilitynopa}{%
 global \nstability}

\newcommand{\lnstabilitynopa}{%
 local \nstability}

\newcommand{\pRMl}{\probname{Robust Matching}}

\usepackage{colortbl}

 \hypersetup{menucolor=orange!40!black,filecolor=winered,urlcolor=green!40!black, pdfencoding=auto, psdextra,
colorlinks=true,citecolor=green!60!black,linkcolor=winered,
pagebackref=true}

\newcommand{\myemph}[1]{{\color{darkgreen!70!black}\emph{#1}}}
\newcommand{\myparagraph}[1]{\smallskip
  
\noindent  \textbf{#1}}

\newtheorem{theorem}{Theorem}[section]
\newtheorem{corollary}[theorem]{Corollary}
\newtheorem{definition}[theorem]{Definition}
\newtheorem{lemma}[theorem]{Lemma}
\newtheorem{observation}[theorem]{Observation}
\newtheorem{proposition}[theorem]{Proposition}
\newtheorem{claim}{Claim}

\newtheoremstyle{defn}%
  {}%
  {}%
  {}%
  {}%
  {\bfseries}%
  {.}%
  { }%
  {}%
  \theoremstyle{defn}
  \newtheorem{example}[theorem]{Example}

\DeclareMathOperator{\sm}{\mathcal{S\!M}}
\DeclareMathOperator{\bp}{\mathcal{B\!P}}
\newcommand{\shifts}{\mathcal{S\!H}}
\newcommand{\shiftss}{\mathsf{shl}}

\newcommand{\polyone}{\ensuremath{\mathsf{poly}_1}}
\newcommand{\polytwo}{\ensuremath{\mathsf{poly}_2}}

\newcommand{\poly}{\ensuremath{\mathsf{poly}}}

\newcommand{\opt}{\ensuremath{\mathsf{opt}}}

\newcommand{\locald}{{\mathsf{d_L}}}
\newcommand{\globald}{{\mathsf{d_G}}}

\newcommand{\globaldapprox}{{\mathsf{d^{*}_G}}}

\newcommand{\globaldapproxbound}{{\polyone(\globald)\!\cdot\!\globald}}

\newcommand{\localdapproxbound}{{\polyone(\locald)\!\cdot\!\locald}}

\newcommand{\sucw}{\mathsf{succ}}

\newcommand{\egalcostn}{egalitarian cost}

\newcommand{\egalcost}{\ensuremath{\eta}}
\newcommand{\egalcostapprox}{\ensuremath{\egalcost^*}}
\newcommand{\egalcostapproxform}{\ensuremath{\polytwo(\egalcost)\cdot \egalcost+2}}
\newcommand{\egalcostapproxbound}{\ensuremath{\polytwo(\egalcost)\!\cdot\!\egalcost}}
\newcommand{\unmatched}{\ensuremath{\mathsf{n_{u}}}}
\newcommand{\matched}{\ensuremath{\mathsf{n_{m}}}}

\usepackage{needspace}

\newcommand{\prob}[6]{%
  \needspace{3\baselineskip}
    \begin{description} %
      \setlength\topsep{-.15ex} \setlength\itemsep{-.2ex}
    \item[#1]
    \item[\emph{#2}]#3
    \item[\emph{#4}]#5
    \end{description}%
}

\newcommand{\probname}[1]{{\normalfont\textsc{#1}}}

\newcommand{\probdef}[3]{\prob{\probname{#1}}{Input:}{#2}{Question:}{#3}{as}}

\newcommand{\proofparagraph}[1]{\par\smallskip\noindent\fbox{\textit{#1}}}

\usepackage[noend,ruled,linesnumbered]{algorithm2e}
\usepackage[ruled]{algorithm2e} %

\SetAlFnt{\small}
\SetAlCapFnt{\small}
\SetAlCapNameFnt{\small}
\SetAlCapHSkip{0pt}
\IncMargin{-\parindent}

\usepackage[sort&compress,nameinlink,capitalize]{cleveref}
\crefname{algorithm}{Algorithm}{Algorithms}
\crefname{proposition}{Proposition}{Propositions}
\crefname{observation}{Observation}{Observations}
\crefname{corollary}{Corollary}{Corollaries}
\crefname{theorem}{Theorem}{Theorem}
\crefname{lemma}{Lemma}{Lemmas}
\crefname{claim}{Claim}{Claims}
\crefname{section}{Section}{Section}
\crefalias{AlgoLine}{line}
\crefname{example}{Example}{Examples}
\crefname{definition}{Definition}{Definitions}
\crefname{figure}{Figure}{Figures}
\crefname{appendix}{Appendix}{Appendix}

\usepackage{etoolbox} %

\newcommand{\appsymb}{$\star$}
\newcommand{\appref}[1]{{\hyperref[#1]{\appsymb}}}

\newcommand{\toappendix}[1]{%
     {#1}
 }

\newcommand{\toappendixalter}[3]{%
     #2 {#3}
  }
 
\newcommand{\appendixproofwithstatement}[3]{%
    #3
}

\newcommand{\appendixcorrectnessproofwithstatement}[4]{%
    #4
}

\newcommand{\appendixsection}[1]{%
}

\newcommand{\rotationexample}{
  \begin{tikzpicture}
     \def \xsc {4}
     \def \dist {.65}
     \foreach \j/\i/\p/\o/\s/\a/\b/\c/\d  in
     {u/1/0/left/w/2/3/1/4,
    u/2/0/left/w/3/4/2/1, u/3/0/left/w/4/1/3/2, u/4/0/left/w/1/2/4/3,%
    w/1/1/right/u/1/2/3/4,
    w/2/1/right/u/2/3/4/1,
    w/3/1/right/u/3/4/1/2,
    w/4/1/right/u/4/1/2/3} {
    \node[minimum size=3ex, inner sep=5pt] at (\p*\xsc, -\i*\dist) (n\j\i) {$\j_\i\colon  \s_\a \s_\b\s_\c \s_\d$};
  }

  \foreach \s/\t in {u1/w2,u2/w3,u3/w4,u4/w1} {
    \draw[redline] (n\s.east) -- (n\t.west);
  }

  \foreach \s/\t in {u1/w1,u2/w2,u3/w3,u4/w4} {
    \draw[blackline] (n\s.east) -- (n\t.west);
  }

\end{tikzpicture}
}

\newcommand{\rotatione}{(u_3,w_2,u_4, w_1)}

\newcommand{\appendixexample}[4]{%
    {#4}
}

\newcommand{\appendixfigure}[4]{%
      #4
}

\title{Matchings under Preferences: Strength of Stability and~Trade-Offs}

\author{Jiehua Chen \and Piotr Skowron \and Manuel Sorge\\[2ex] {\small University of Warsaw, Warsaw, Poland}\\ {\footnotesize \texttt{jiehua.chen2@gmail.com} \quad \texttt{p.skowron@mimuw.edu.pl}\quad \texttt{manuel.sorge@gmail.com}}} 

\date{}

\begin{document}

\maketitle

\begin{abstract}
  We propose two solution concepts for matchings under preferences: \emph{robustness} and \emph{near stability}.
  The former strengthens while the latter relaxes the classic definition of stability by Gale and Shapley (1962). Informally speaking, robustness requires that a matching must be stable in the classic sense, even if the agents slightly change their preferences. Near stability, on the other hand, imposes that a matching must become stable (again, in the classic sense) provided the agents are willing to adjust their preferences a bit.
Both of our concepts are quantitative; together they provide means for a fine-grained analysis of the stability of matchings.
Moreover, our concepts allow the exploration of trade-offs between stability and other criteria of social optimality, such as the egalitarian cost and the number of unmatched agents. We investigate the computational complexity of finding matchings that implement certain predefined trade-offs.
We provide a polynomial-time algorithm that, given agent preferences, returns a socially optimal robust matching, and we prove that finding a socially optimal and nearly stable matching is computationally hard.
\end{abstract}

\section{Introduction}\label{sec:intro}

In the \textsc{Stable Marriage} problem~\cite{GaleShapley1962} we are given two disjoint sets of agents, $U$ and $W$. %
Each agent from one set has a strict preference list that ranks a subset of the agents from the other set.
The sets of agents and their preference lists are collectively called \myemph{preference profile}.
The goal is to find a matching---i.e., a bijection between $U$ and $W$---that does not contain a \myemph{blocking pair}, i.e., a pair of agents who prefer each other over their matched partners.  
A matching with no blocking pairs is called a \myemph{stable matching}.

The classic definition of stability is qualitative: %
A matching can be either stable or not, and there are no other states in between or beyond. In this paper, by contrast, we take a quantitative approach.
We propose and study two solution concepts:~\myemph{robustness} and \myemph{\nstability}, where the former strengthens and the latter relaxes the notion of stability.
Intuitively, a robust matching is more than stable; it remains stable even if agents change their preferences slightly. In contrast, a \nstable matching needs not be stable for the original profile,
but it becomes so after some minor changes in the preferences. 
Below we give more precise definitions of robust and \nstable matchings and motivate their study through a number of observations.

\looseness=-1
\paragraph{Robust matchings.}
Our first main observation is that the preference lists provided by the agents do not always reflect their true preferences. This can happen, for instance, because the agents do not have full information about their potential partners, or because formulating accurate preferences is a hard task that requires substantial cognitive effort~\cite{Mil56}.
It is also typical that the agents change their preferences over time, for instance, in response to changes in their operating environment%
. Thus, a matching that is stable in the classic sense (with respect to the preferences expressed by the agents at the beginning) can in fact contain two or more agents who already have or will likely have incentives to drop their assigned partners and be matched with each other. In other words, there are situations where the classic definition of stability can turn out to be too weak. %
In a different setting, a third party may want to destabilize a
matching by bribing certain agents to change their preferences. In that case,
we are interested in stable matchings which defy such
attacks.

For the above reasons, we introduce and
study %
\myemph{$d$-robustness}, a strengthened notion of stability. A
matching is \myemph{$d$-robust} for a given preference profile if
this matching is stable and remains stable after performing an arbitrary sequence of $d$ swaps.
Here, a \myemph{swap} is the reversal of two consecutive agents in a
preference list. Intuitively, if a matching is $d$-robust for some
reasonably large~$d$, then it will not become unstable even if the
agents specified slightly inaccurate preferences, nor will it become
unstable even if the agents change their preferences by a little.
\Cref{ex:intro_example} below illustrates the concept of robustness.%

\begin{example}\label[example]{ex:intro_example}
 Consider the profile~$P$ below with $4$ men and $4$ women, where the preference lists are to the right of the corresponding agents; preferences are represented as horizontal lists where more preferred agents are put to the left of the less preferred ones.

 \begin{center}
   \rotationexample
 \end{center}
\noindent This profile admits five stable matchings:

\begin{compactenum}[(1)]
\item The $U$-opt.\ stable
  matching~$M_1\!=\!\{\{u_1, w_2\}, \{u_2, w_3\}, \{u_3, w_4\},
  \{u_4,w_1\}\}$~(red~dashed~lines),

\item the $W$-opt.\ stable matching~$M_2\!=\!\{\{u_1, w_1\}, \{u_2, w_2\}, \{u_3, w_3\}, \{u_4,w_4\}\}$~(black~solid~lines), 

\item $M_3=\{\{u_1, w_3\}, \{u_2, w_2\}, \{u_3, w_1\},
  \{u_4,w_4\}\}$, %

\item $M_4\!=\!\{\{u_1, w_1\}, \{u_2, w_4\}, \{u_3, w_3\}, \{u_4,w_2\}\}$, and %
\item $M_5\!=\!\{\{u_1, w_3\}, \{u_2, w_4\}, \{u_3, w_1\}, \{u_4,w_2\}\}$.
\end{compactenum}
\noindent
In terms of robustness, $M_2$ is superior to~$M_1$ since $M_2$ is $1$-robust but $M_1$ is not. To see that $M_2$ is $1$-robust we observe that, to make $M_2$ unstable, we need to perform one swap in the preference list of an agent in~$W$. However, no such single swap will make $M_2$ unstable.
Stable matching~$M_1$ is not $1$-robust since one can swap in the preference list of any agent~$u$ from $U$ the two agents $M_1(u)$ and $w$ in the first and the second positions 
to obtain a profile where~$\{u,w\}$ is a blocking pair for~$M_1$.
\hfill $\diamond$
\end{example}

\paragraph{\Nstable matchings.}
Our second main observation is that there
exist other factors, apart from the preferences, that can discourage
the agents to break their relations with their matched partners. Such
factors may include social pressure and additional costs incurred by
changing the partner, for example. Thus, in some situations even weaker forms
of stability may guarantee a sufficient level of
resilience to agents changing their minds.
We express this as the \myemph{\lnstability{$d$}} of a matching, which stipulates that there is a
sequence of swaps such that the matching becomes stable, and in each agent's preference list, at
most $d$~swaps are made.

\looseness=-1
This concept has an intuitive interpretation similar to the
$\epsilon$-Nash-equilibrium~\cite[Section 2.6.6]{Papadimitriou07} in game theory: 
In a \lns{$d$} matching no agent can improve
its satisfaction by more than $d$ through rematching (see also the
equivalent definitions in \Cref{prop:r-maximal-d-bp-local-d}). This is analogous to
$\epsilon$-Nash-equilibria, where no agent can improve their outcome by
more than $\epsilon$. In this sense, \lnstabilitynopa also measures
the strength of the incentive for two agents in a blocking pair to change their partners.\footnote{There are some differences between the two concepts since we deal
with ordinal preferences. Yet, our concepts generalize to cardinal utilities, where the similarities are more transparent.}

Our third main observation is that, when there are constraints on
other factors of the matching like social welfare (see below), it may
not be possible to find a stable matching satisfying these
constraints. Thus, it may be necessary
to balance between the social welfare and the costs incurred by agents
that want to switch partners. This cost is captured by the \myemph{\gnstability{$d$}} of a matching~$M$, stating that there is a
sequence of at most~$d$ swaps in total such that~$M$ becomes stable.
In order to achieve the desired social welfare, we may thus provide
proportionate compensation to the agents affected by the swaps. 

Taking \nstable matchings into consideration may indeed allow
us to find a matching that is significantly better from the
perspective of the society as a whole, than if we restricted ourselves
to stable matchings only. This is illustrated in the following
example:

\begin{example}\label[example]{ex:egalitarian_cost}
Let $U = \{a_0, \ldots, a_{n-1}, x_1, \ldots, x_{n}\}$ and $W = \{b_0, \ldots, b_{n-1}, y_1, \ldots, y_{n}\}$, and consider the following preference profile $P$ of the agents; the index ``$i+1$'' is taken modulo $n$.

\medskip
\begin{tabular}{lllll}
  & $a_0\colon b_0 \; b_1$ %
                  \qquad     && $b_0\colon a_0 \; a_{n-1}$  \\ 
& $a_i\colon b_i \; b_{i + 1}$  %
                                                                                                 \qquad     && $b_i\colon a_{i - 1} \; x_1 \; \ldots  \; x_n \; a_i$  \qquad &(\text{for all $i \in \{1, \ldots, n - 1\}$})\text{.} \\ 
  &$x_i  \colon y_i  \; b_1 \;\ldots\; b_{n-1}$                                               \qquad && $y_i\colon x_i$                        \qquad &(\text{for all $i \in \{1, \ldots, n\}$})\text{.}\\
\end{tabular}

\medskip
                                                                                                      
\noindent\looseness=-1 In every stable matching of $P$ agent~$x_i$ must be matched with $y_i$ for all $i \in \{1, \ldots, n\}$, and $a_0$ with~$b_0$.
Consequently, $a_{1}$ needs to be matched with with $b_{1}$ and, by an inductive argument, we can infer that, for each $i \in \{1, \ldots, n-1\}$, $a_i$ must be matched with $b_i$.
Thus, if we look at the agents from $B= \{b_1, \ldots, b_n\}$, we observe that, except for $b_0$, each of them is matched with a partner ranked at the $(n+2)$th position.
Yet, if we consider the profile obtained from $P$ by swapping $a_0$ and $a_{n-1}$ in the preference list of $b_0$, then $M = \{\{a_0, b_1\}, \{a_1, b_2\}, \ldots, \{a_{n-1}, b_{0}\}\}\cup \{ \{x_i, y_i\} \mid 1\le i \le n\}$ would be a stable matching. In this matching everyone is matched to one of her two most preferred agents. 
\hfill $\diamond$
\end{example}

Intuitively, \Cref{ex:egalitarian_cost} shows that with a relatively small loss of stability, one can significantly improve the social cost of a matching $M$---in this example, this cost is defined as the sum of ranks that an agent $a$ has in the preference list of its matched partner~$M(a)$. In the literature this measure is often referred to as the egalitarian cost of a matching~\cite{IrLeGu1987}. We also consider another metric that counts the number of agents that are assigned a partner in a matching. Recall that we assume that the preference rankings of the agents are incomplete: the agents \emph{do not} rank those from the opposite set that they would not agree to be matched to. In such a case a stable matching does not need to be \emph{perfect}, i.e., it is possible that some agents will not be matched at all. The effect of stability on the number of matched agents is illustrated in \Cref{ex:perfectness}.

\begin{example}\label[example]{ex:perfectness}
  Consider a profile with $2$ men, $a_1$ and $a_2$, and $2$ women, $b_1$ and $b_2$, with preference lists:
  $a_1\colon b_1;\quad b_1\colon a_2 \; a_1; \quad a_2\colon b_1 \; b_2;\quad b_2\colon a_2\text{.}$
For this profile, the only stable matching is $\{\{a_2, b_1\}\}$. However, if we swapped $b_1$ and $b_2$ in the preference list of $a_2$, then $\{\{a_1, b_1\}, \{a_2, b_2\}\}$ would be a stable matching, i.e., we would obtain a stable matching where more agents have partners.
\hfill $\diamond$
\end{example}
  
\Cref{ex:egalitarian_cost,ex:perfectness} suggest that there is a
(possibly non-linear) trade-off between stability and other criteria
of social optimality. Our definition of \nstability provides a
formalism necessary to describe the trade-offs; yet, in order to take
advantage of them, one needs to be able to identify situations where a
large improvement of social welfare is possible with a relatively
small sacrifice of stability. We formalize this question as a
computational problem (see \cref{subsec:def:robust+nearstable} for formal definitions) and study its complexity. 

\paragraph{Our contributions.}

We introduce the concepts of robustness and \nstability, and explore the trade-off between stability and the egalitarian cost and between stability and the number of matched agents.
We provide a polynomial-time algorithm that, given a preference profile and a number~$d$, finds a matching which is $d$-robust if it exists (\cref{thm:d-robust-poly}).
We achieve this by providing a polynomial-size characterization of the profiles (\cref{sec:profile-decomposition}) which are close to the input profile
and by heavily exploiting the structural properties of so-called rotations adherent to a preference profile~\cite{GusfieldIrving1989}. Moreover, we provide a polynomial-time algorithm that finds a $d$-robust matching with minimum \egalcostn\ if one exists (\cref{cor:d-robust-perfect-egal-poly}). However, when ties are present, we show that finding a robust matching is NP-hard (\cref{thm:robust-ties-np-hard-d-unbounded}). 

In contrast to the polynomial-time algorithms for robust matchings, we
show that the problem of finding a matching that implements a certain
trade-off between the near stability and the egalitarian cost, or
between the near stability and the perfectness of the matching is
NP-hard, and it is NP-hard to approximate
(\cref{thm:nearly-stable-inapproximable}). Motivated by this general
hardness result, we study the parameterized complexity, mainly with
respect to the parameter number of allowed swaps (for details on parameterized complexity we refer to the books of \citet{CyFoKoLoMaPiPiSa2015}, \citet{DF13},  \citet{FG06}, and \citet{Nie06}). See
\cref{tab:summary} for a summary. Unfortunately, we mostly obtain
further hardness results. While for \lnstabilitynopa\ even only one
allowed swap leaves the problem NP-hard
(\cref{thm:nearly-stable-inapproximable}), for \gnstabilitynopa\ there is
a polynomial-time algorithm for each constant number~$\globald$ of
allowed swaps~(\cref{prop:xp-swaps}). The exponent in the running time
depends on $\globald$, however, and this dependency cannot be removed
unless the unlikely complexity-theoretic collapse FPT${}={}$W[1]
happens~(\cref{thm:global-w1hard-swaps}). We also study the complexity
in the cases where there are small numbers of unmatched or matched
agents in a classically stable matching of the input profile.

\begin{table}[t!]
\caption{Summary of our results, where $d$ denotes the number of swaps for robust matchings, $\locald$ (resp.\ $\globald$) denotes the number of swaps for \gnstabilitynopa (resp.\ \lnstabilitynopa), $\egalcost$ denotes the egalitarian cost of the desired matching, and %
  $\unmatched$ %
  denotes the number of %
  unmatched
  agents in any stable matching of the initial profile without ties.}\label{tab:summary}
\resizebox{\textwidth}{!}{
\begin{tabular}{|p{1.8cm}|c|c|c|c|}
  \toprule
 Social& {Robust} & {Robust}&{\GNely Stable} & {\LNely Stable} \\
  criteria & (without ties)& (with ties)& (without ties)& (without ties)\\\midrule
  No further & P &  {NP-h ($d=1$)} & \multicolumn{2}{c|}{Always exists even for $\globald=\locald=0$}\\
  restrictions & [Thm~\ref{thm:d-robust-poly}] &  {[Thm~\ref{thm:robust-ties-np-hard-d-unbounded}]} & \multicolumn{2}{c|}{and can be found in $O(n^2)$ time~\cite{GaleShapley1962,GusfieldIrving1989}}\\\hline
  &&&&\\[-2ex]
  Perfect & P& NP-h ($d=0$) & XP for $\globald$, W[1]-h for $\unmatched$~[Cor~\ref{cor:w1h-unmatched}]  & NP-h ($\locald=1$)~[Thm~\ref{thm:nearly-stable-inapproximable}], W[1]-h for $\unmatched$~[Cor~\ref{cor:w1h-unmatched}] \\
  matching  & [Thm~\ref{cor:d-robust-perfect-egal-poly}] & \cite{MaIrIwMiMo2002} & \emph{No} poly-approximation%
                   ~[Thm~\ref{thm:nearly-stable-inapproximable}] & \emph{No} poly-approximation %
                                                                   ~[Thm~\ref{thm:nearly-stable-inapproximable}] \\\hline
  &&&&\\[-2ex]
  Egalitarian &  P &  NP-c ($d=0$)& XP for $\globald$, W[1]-h for $\unmatched$ & NP-h ($\locald=1$)~[Thm~\ref{thm:nearly-stable-inapproximable}], W[1]-h for $\unmatched$\\
  cost~$\egalcost$ & [Thm~\ref{cor:d-robust-perfect-egal-poly}] & \cite{MaIrIwMiMo2002} & \emph{No} poly-approximation %
                          ~[Thm~\ref{thm:nearly-stable-inapproximable}] & \emph{No} poly-approximation %
                                                                          ~[Thm~\ref{thm:nearly-stable-inapproximable}] \\
\bottomrule
\end{tabular}
}
\end{table}
\paragraph{Related work.}
For an overview on the \textsc{Stable Marriage} and related problems, we refer to the books of \citet{Knuth1976}, \citet{GusfieldIrving1989}, and \citet{Manlove2013}.

First, we review work related to our concept of robustness.
As we mentioned in the beginning of this section, one of the observations that motivates our study of robust matchings is that the preferences of the agents may be uncertain.
In this regard, \citet{AzBiGaHaMaRa2016}, \citet{MiOk2017}, and \citet{CheNieSkoECmstable2018} study a variant of \textsc{Stable Marriage} where  
there is a collection of ``possible'' preference profiles given as input, and
they look for a matching that is stable in each of the given profiles 
(the corresponding computational problem is NP-hard even for constant number of input profiles).
Our work starts with the assumption that the preferences provided by the agents are a good approximation of their true preferences.
Thus, our robustness concept respects every profile that is close to the preferences provided by the agents.
This makes a crucial difference---finding a robust matching if one exists, according to our definition, is solvable in polynomial time.

\looseness=-1
Our robustness concept is related to the works of~\citeauthor{MaiVaz2018}~\cite{MaiVaz2018,MaiVaz2018-birkhoff-arxiv}. 
They introduced a probabilistic model, where there are \emph{polynomially} many preference profiles given in the input,
each differing from the original one by a \emph{single} agent's preference list.
While they do not assume this difference to be small,
they assume there is a probabilistic distribution over these polynomially many preference profiles, and 
the goal is to find a stable matching that stays stable with the highest probability.
In contrast, in our definition of robustness,
we require that the sought matching must be stable in \emph{every} profile which is close to the original one,
but which can differ from the original profile in preference lists by more than one agent.
Furthermore, we do not assume that the distribution of the profiles is given,
but rather infer the ``relevant'' (close) profiles directly from the original profile. Our approach, based on the concept of distances, induces
a quantitative measure of the strength of stability; we further extend it in the converse direction by considering matchings that are \nstable,
getting a full set of tools that allow to reason about the strength of stability for any matching.

\citeauthor{MaiVaz2018}~\cite{MaiVaz2018,MaiVaz2018-birkhoff-arxiv} proved that a matching
that stays stable with the largest probability for a given probability
distribution can be found in polynomial time as long as only a single agent changes her preferences in each profile. Their
techniques crucially rely on the fact that, for each of the preference profiles~$P'$ that has nonzero probability, the set of matchings
that are stable for both $P'$ and the input profile~$P$ has a certain type of sublattice structure of the lattice of stable matchings of~$P$.
In our model, this is not the case; the set of stable matchings for the preference profiles obtained by $d$ arbitrary
swaps may not be a sublattice anymore. 
Thus, this approach is not directly applicable in our scenario.
Moreover, in the general case, where we allow arbitrary changes between profiles, it is impossible to obtain a polynomial-time algorithm unless P${}={}$NP: \citet{MiOk2017} and \citet{CheNieSkoECmstable2018} showed that finding a matching which is stable for even only two profiles is NP-hard.
Nevertheless, we provide a compact characterization of all close profiles,
based on which, and partly inspired by the techniques of \citet{MaiVaz2018}, we provide a polynomial-time algorithm for our robustness model.
In this regard, our algorithmic techniques can be considered as a
generalization of the ones by
\citet{MaiVaz2018}.

Finally, let us mention a relation between robustness and strategy-proofness. We say that a matching algorithm is \myemph{strategy-proof}~(see \cite[Chapter 4]{GusfieldIrving1989}, \cite[Chapter~1.7]{RothSotomayor1992}, \cite[Chapter~2.9]{Manlove2013}) if no agent can obtain a better partner by misreporting her preferences; it is known that there exists no strategy-proof matching mechanism. Robustness implies a very weak form of strategy-proofness, where the set of agents' strategies is limited---the agents are willing to report only those rankings that are not significantly different from their true preferences. 
Even more closely related, robustness implies resilience to certain forms of bribery---the problem of bribery, originally defined for single-winner elections~\citep{FalRot16}, can be naturally adapted to matchings.         

\looseness=-1
Now we turn to work that is more related to our \nstability{} concept.
Another interpretation of a \lns{$d$} matching is that in each blocking pair there is an agent whose rank improvement by switching partners would be at most~$d$. \citet{DruBou13} use this rank improvement approach to study \textsc{Stable Marriage} problem under \emph{partially ordered} preferences. They introduced the notion of an \myemph{$r$-maximally stable} matching, i.e.,
a matching such that for each linear completion of the input profile
and for each unmatched pair at least one agent in the pair ranks the other higher than its matched partner by at most $r$ positions.
When restricting the input to linear preferences, as is our focus here, $r$-maximal stability is equivalent to \lnstability{r} for each $r \ge 0$.
We prove this formally in \cref{prop:r-maximal-d-bp-local-d}. 
Here, in contrast, we do not deal with partial preferences, but instead we want to achieve a given social welfare in addition to $r$-maximal stability.

\citet{pini_stability_2013} and \citet{anshelevich_anarchy_2013} studied a concept called (additive)
$\alpha$-stability that measures the degree of instability for
utility-based preferences.
For ordinal preferences, their concept is equivalent to our \lnstability{$\alpha$}.
\citet{anshelevich_anarchy_2013} studied the trade-off between the total utility of a matching and its $\alpha$-stability for restricted structures of utility scores (which
cannot model ordinal preferences). \citet{pini_stability_2013} showed
that a certain kind of lexicographically optimal $\alpha$-stable
matching can be found in polynomial time and
they considered manipulation issues.

\looseness=-1
Finally, we review further related work, not necessarily directly
related to our notions of robustness or \nstability. Recently, \citet{MenLar18} proposed a different robustness concept to deal with uncertain preferences--the authors assume that each agent has preferences with ties on the agents of the opposite set and look for a perfect matching so as to minimize the maximum number of blocking pairs among all linear extensions of the input preferences.
In contrast to our approach, however, these blocking pairs may represent an arbitrarily large rank improvement, i.e. an arbitrarily large number of swaps needed to make the matching stable.
Finding a solution as above is equivalent to finding a \emph{perfect} matching with minimum number of so-called \emph{super-blocking pairs}, a concept introduced by \citet{Irving1994} to cope with preferences with ties, i.e.\ weak orders  (also see \cite{GusfieldIrving1989}).
\citet{MenLar18} mainly obtained inapproximability results.

\citeauthor{GenSiaSimSul17}~\cite{GenSiaSimSul17,GenSiaSimSul17b} provide yet another view on robustness in the context of stable matchings. They define an $(x, y)$-supermatch as a stable matching that satisfies the following property: If any $x$~agents break up, it is possible to rematch these $x$ agents so that the new matching is again stable; further, this re-matching must be done by breaking at most $y$~other pairs. Hence, an $(x,y)$-supermatch
may not be robust in our sense, but it needs to be easy to repair.

\looseness=-1
In the second part of this paper we study trade-offs between the stability (of various strength) and other criteria of optimality such as the egalitarian cost and the number of unmatched agents. This is related to the studies on the price of stability in matching markets~\citep{BouKno13}.
Concepts similar to our robustness have been also studied in other contexts, for instance for single-winner~\citep{ShiYuElk13} and multi-winner elections~\citep{BreFalKacNieSkoTal18}. 
\section{Basic Definitions, Notations, and Our Stability Concepts}
\label{sec:defi}
\looseness=-1
For each natural number $t$ by $[t]$ we denote the set~$\{1, 2, \ldots, t\}$.
Let $U = \{u_1, \ldots, u_n\}$ and $W = \{w_1, \ldots, w_n\}$ be two $n$-element disjoint sets of agents.
A \myemph{preference profile}~$P=((\succ_u^P)_{u\in U}, (\succ_w^P)_{w\in W})$ is a collection of the \myemph{preference lists} of the agents from $U$ and $W$.
Here, for each agent~$u \in U$, the notation~$\succ_u^P$ denotes a linear order on a subset~$W'$ of~$W$ that represents the ranking of agent $u$ over all agents from $W'$ in profile~$P$. The agents in~$W'$ are also called \myemph{acceptable} to~$u$. The candidates \myemph{not ranked} by $u$ are those in $W \setminus W'$, that is, those that $u$ does not agree to be matched to; we also call them \myemph{unacceptable}. If $w \succ_u^P w'$ then we say that $w$ is \myemph{preferred} to~$w'$ by $u$ in~$P$.
Analogously, for each agent~$w \in W$, $\succ_w$ represents a linear order on (a subset of)~$U$ that represents the ranking of~$w$ in profile~$P$ and we likewise use the notions of preference list, preferred, (un-)acceptable, and (not) ranked.
We say that $P$ has \myemph{complete} preferences if each agent finds all agents from the opposite set acceptable.

Given an agent~$x$ with her preference list~$\succ_x$ and given an agent~$y$ from the opposite set,
the \myemph{rank~$\rank_{x}(y, \succ_x)$} of~$y$ in the preference list of~$x$ is equal to the number of agents %
that are preferred to~$y$ by~$x$. If $y$ is not acceptable to~$x$ then
we let $\rank_x(y, \succ_x)$ be equal to the number of agents acceptable to $x$.
We usually omit the symbol~$\succ_x$ in $\rank_x(y,\succ_x)$ and write only $\rank_x(y)$ whenever the preference list of~$x$ is clear from the context.
For instance, the rank of $y_3$ in the preference list~$\succ_x\colon y_1 \succ y_3 \succ y_2$ is one.
We say that \myemph{$x$ ranks $y$ higher than $z$}, if $\rank_x(y) < \rank_x(z)$.

Throughout, except in \cref{sec:Robust+Ties}, by ``$x\succeq y$'' for two agents~$x$ and $y$ we mean ``$x=y$ or $x\succ y$''.

\looseness=-1
\paragraph{Blocking pairs and stable matchings.}
Given two disjoint sets of agents,~$U$ and $W$, a \myemph{matching}~$M$ is a set of pairwise disjoint pairs, each pair containing one agent from $U$ and one agent from $W$, {i.e.}\ 
$M\subseteq \{\{u,w\}\mid u \in U \wedge w \in W\}$ and for each two pairs~$p, p'\in M$ it holds that $p \cap p' = \emptyset$.
Given a pair~$\{u,w\}$ with $u\in U$ and $w\in W$,  
if it holds that $\{u,w\}\in M$, then we use \myemph{$M(u)$} to refer to $w$ and $M(w)$ to refer to $u$,
and we say that $u$ and $w$ are their respective \myemph{partners} under $M$;
otherwise we say that $\{u,w\}$ is an \myemph{unmatched pair}.
We say that $\{u,w\}$ is \myemph{blocking} (or \myemph{a blocking pair of}) \myemph{$M$} if both $u$ and $w$ would prefer to be matched together than to stay with their current partners. Formally, $\{u,w\}$ is a blocking pair if the following holds:
\begin{inparaenum}[(1)]
  \item $u$ and $w$ find each other acceptable but are not matched together,
  \item $u$ is either unmatched by $M$ or $\rank_u(w) < \rank_u(M(u))$, and
  \item $w$ is either unmatched by $M$ or $\rank_w(u) < \rank_w(M(w))$.
\end{inparaenum}
We say that a matching~$M$ is \myemph{stable} if no unmatched pair forms a blocking pair for $M$.
\cref{ex:intro_example} in the introduction illustrates stable matchings.

We use \myemph{$\sm(P)$} to denote the set of all stable matchings for a preference profile~$P$.
Given a matching~$M$, we use \myemph{$\bp(P,M)$} to denote the set of all unmatched pairs that are blocking $M$ in profile~$P$.
Obviously, for each stable matching~$M\in \sm(P)$, it holds that $\bp(P,M)=\emptyset$.

\subsection{Our Spectrum of Stability Notions and Problems}\label{subsec:def:robust+nearstable}
Let us now define our concepts of robustness and \nstability, informally introduced in \cref{sec:intro}. %

First of all, we need the notion of \myemph{swaps}, which describes the operation of taking two consecutive agents~$x$ and $y$ in a preference list of a third agent~$z$ and switching their relative order in order to obtain a new preference list.
We also use $(z,\{x, y\})$ to denote such a swap.
Given two preference lists~$\succ$ and $\succ'$, the \myemph{swap distance} (also known as the Kendall~$\tau$ distance~\cite{Kendall1938}) between~$\succ$
and $\succ'$ 
is defined as the number of differently ordered pairs in the two lists; if the two lists are defined on two different acceptable sets, then the swap distance is infinity.
Intuitively, the swap distance is equal to the minimum number of swaps that are required to turn $\succ$ into $\succ'$.
Accordingly, the \myemph{swap distance} between two preference profiles $P_1$ and $P_2$, denoted as $\tau(P_1, P_2)$, is defined as
the sum of swap distances between the two preference lists of each agent in profiles~$P_1$ and $P_2$.

\begin{definition}[Robustness]\label[definition]{def:robustness}
For a given preference profile $P$, we say that a matching $M$ is \myemph{$d$-robust} if for each profile $P'$ with $\tau(P, P') \leq d$ it holds that $M$ is stable in $P'$.    
\end{definition}

Note that our robustness concept is monotone--each $d$-robust matching is also $d'$-robust for $0\le d' \le d$.
We are interested in the computational question of finding the maximal integer~$d$ such that there is a $d$-robust matching. This can be phrased as a decision problem as follows:

\probdef{\pRMl}
{A preference profile~$P$ with agent sets~$U$ and $W$ of size $n$ each, %
and an integer~$d\in \mathds{N}$.}
{Is there a $d$-robust matching for $P$?}

Now, we define \nstability. Here, we provide two definitions---\gnstabilitynopa and \lnstabilitynopa---that differ in the scope of admissible changes to the original preference profile.

\begin{definition}[\Nstability]\label[definition]{def:near_stability}
For a given preference profile $P$, we say that a matching $M$ is \myemph{\gns{$d$}} if there exists a profile $P'$ with $\tau(P, P') \leq d$ such that $M$ is stable in $P'$. We say that $M$ is \myemph{\lns{$d$}} if there exists a profile $P'$ with $\tau(\succ_x^{P}, \succ_x^{P'}) \leq d$ for each agent $x \in U\cup W$ such that $M$ is stable in $P'$. %
\end{definition}

\looseness=-1
Since \nstability is a more permissive concept than stability as defined by~\citet{GaleShapley1962}, it is straight-forward that a \gns{$d$} (or \lns{$d$}) matching always exists for $d\ge 0$. Here, our main focus is to explore the trade-offs between the strength of stability and other criteria of social optimality. We say that a matching $M$ is \myemph{perfect} if each agent has a partner in $M$. The \myemph{\egalcostn} of $M$ in a profile~$P=(\succ_x)_{x\in U\cup W}$ is
$\egalcost(M) \coloneqq \sum_{x \in U\cup W} \rank_x(M(x), \succ_x)$.
\iflong
This leads to the following computational problems, abbreviated as \pGNSEM, and \pLNSEM.
\probdef{\pGLNSPMl}
{A preference profile~$P$ with agent sets~$U$ and $W$ of size $n$ each,
and an integer~$d\in \mathds{N}$.}
{Is there a \gns{$d$} (or \lns{$d$}) stable matching for~$P$ which is perfect?}

\fi

\probdef{\pGLNSEMl}
{A preference profile~$P$ with agent sets~$U$ and $W$ of size $n$ each, 
and two integers~$d, \egalcost \in \mathds{N}$.}
{Is there a \gns{$d$} (or \lns{$d$}) stable matching for~$P$ which has \egalcostn\ at most $\egalcost$ in~$P$?}

\ifshort
\pGNSPMl~(\pLNSPM) and \pLNSPMl~(\pLNSPM) are defined similarly,
with the only differences that the input will not contain the \egalcostn~$\egalcost$
and we ask for a nearly stable perfect matching.
\fi

 For preferences without ties (i.e.\ every agent has a strict preference list), we will use the following fundamental result from the literature.

\begin{proposition}[{\cite[Theorem~1.4.2]{GusfieldIrving1989}}]\label[proposition]{prop:SMI-matched-agents-the-same}
\looseness=-1  For incomplete preferences without ties, the agent set can be partitioned into two disjoint subsets~$R$ and $S$ such that every stable matching matches every agent from $R$ and none of the agents from $S$.
  For agent set of size~$2n$, this partition can be computed in $O(n^2)$~time.
\end{proposition}

\subsection{Structural Properties of Robust and \NStable Matchings}\label[section]{sub:structural-properties} %
\appendixsection{sub:structural-properties}
Before we proceed further, we provide some structural results concerning \Cref{def:robustness,def:near_stability}. First we give two observations regarding robustness. These are not necessary for the considerations about algorithms later on, but serve to strengthen the intuition about profiles that allow for robust matchings and, we feel, are interesting in their own right. Further below, we consider the trade-off between near stability and perfectness of matchings and give alternative characterizations of locally nearly stable matchings.

\newcommand{\sized}{%
  If $d\ge n$ and there exists one agent who finds at least two other agents acceptable, 
  then no stable matching is $d$-robust.
}
\begin{proposition}%
  \label[proposition]{prop:d-n}
  \sized%
\end{proposition}
\appendixproofwithstatement{prop:d-n}{\sized}{
\begin{proof}
  Let $M$ be an arbitrary stable matching.
  To show that no stable matching is $d$-robust, it suffices to show that performing at most $n$~swaps can make an unmatched pair a blocking pair of~$M$.
  To this end, let $x$ be an agent who finds at least two other agents acceptable.
  Further, let $y$ be an acceptable agent of $x$, satisfying the following.
  If $x$ is unmatched under $M$ or if $\rank_{x}(M(x)) \ge 1$, then $y$ is the most preferred agent of $x$ (that is, $\rank_x(y) = 0$); otherwise,~$y$ is the second-most preferred agent of $x$ (that is, $\rank_x(y) = 1$).
  Now, use at most one swap to make agent~$y$ the most preferred agent of $x$, and at most $n - 1$ swaps to make agent~$x$ the most preferred agent of $y$.
  This results in $\{x, y\}$ being a blocking pair of $M$.
  Hence,~$M$ is not $d$-robust, as $d\ge n$.
\end{proof}
}

A matching is \myemph{top-choice} if each agent is matched to her most
preferred partner. A profile is \myemph{position-wise distinct} if
there are no two agents that have the same agent in the same position
in their preference lists.
\newcommand{\topchoicelemma}{%
Every $(n-1)$-robust matching is top-choice and every profile
allowing for an $(n - 1)$-robust matching is position-wise distinct.%
}
\begin{proposition}%
  \label[proposition]{prop:top-choice}
  \topchoicelemma
\end{proposition}

\appendixproofwithstatement{prop:top-choice}{\topchoicelemma}{
\begin{proof}
  Let $P$ be a preference profile on two $n$-element sets~$U$ and $W$,
  and let $M$ be an $(n - 1)$-robust matching.
  We first show that $M$ is top-choice. %
  This is clear if each agent finds only one other agent acceptable.
  Otherwise, there is at least one unmatched pair of agents.
  Observe that for each unmatched pair~$\{x,y\}$ of agents %
  it must hold that
  \begin{align}
    \rank_x(y)+\rank_y(x)\ge n \label{eq:rank-sum}
  \end{align} as otherwise we can perform at most $n-1$ swaps,
  $\rank_x(y)$ swaps in $x$'s preference list and $\rank_y(x)$ swaps in $y$'s preference list,
  to make $x$ and $y$ be each other's most preferred agent.
  This results in $\{x,y\}$ being a blocking pair of~$M$---a contradiction to $M$ being $(n-1)$-robust.

  To show that $M$ is top-choice, towards a contradiction, suppose
  that $M$ is not top-choice. This means that there exists an
  unmatched pair~$\{x,y\}$ such that $y$ is the most preferred agent of
  $x$, i.e.\ $\rank_x(y)=0$. However, by~\eqref{eq:rank-sum}, it
  implies that $\rank_y(x)\ge n$---a contradiction to the fact that
  the rank of an agent is at most $n-1$.

  It remains to prove that $P$ is position-wise distinct. %
  We first consider the case when $n = 2$ and then the case when
  $n \ge 3$; the case with $n=1$ is trivial. Assume that $n=2$. Since
  $M$ is top-choice, we infer that the most preferred agents of two
  different agents are different from each other. Now suppose, for the
  sake of contradiction, that there are three distinct
  agents~$x_1,x_2,y$ such that $\rank_{x_1}(y)= 1$ and
  $\rank_{x_2}(y) = 1$. By \eqref{eq:rank-sum}, we must have that
  $\rank_y(x_1)\ge n-1=1$ and $\rank_y(x_2)\ge n-1=1$.
  Hence,~$\rank_y(x_1)=\rank_y(x_2)=1$ as each rank is at most
  $n-1=1$, this is a contradiction to the fact that no two agents have
  the same rank by the same agent. This finishes the proof for the
  case when $n=2$.
  
  In the remainder of the proof, we assume that $n\ge 3$. Before
  proving that $P$ is position-wise distinct in this case as well, we
  observe that each agent finds all other agents acceptable.

  \looseness=-1 Let $u$ and $w$ be two unmatched agents such that $\rank_u(w) = 1$. By \eqref{eq:rank-sum}, we have $\rank_w(u) \geq n - 1$ and indeed $\rank_w(u) = n - 1$ since $n - 1$ is the largest-possible rank.
  This implies that $w$ appears in the preference list of every agent from $U$.
  By the top-choice property of $M$, we infer that, except for the partner~$M(w)$ of $w$,
  every agent~$x$ from $U \setminus \{M(w)\}$ finds at least two agents acceptable: her partner~$M(x)$ and agent~$w$.
  We claim that indeed $M(w)$ also finds at least two other agents acceptable.
  Since $n\ge 3$, there is a third agent $u'\in U \setminus \{M(w), u\}$ who finds $w$ acceptable.
  Let $w'\in W$ be an agent with $\rank_{u'}(w')=1$.
  Again by \eqref{eq:rank-sum}, this implies that $\rank_{w'}(u')=n-1$.
  Hence, $w'$ also has complete preferences and finds $M(w)$ acceptable, implying that $M(w)$ finds $w'$ acceptable.
  Since $u\neq u'$ and $\rank_{w}(u)=n-1$, we infer that
  $w'\neq w$ because $\rank_{w'}(u')=n-1$.
  This implies that $M(w)$ also finds at least two agents acceptable, namely $w$ and $w'$.
  Since no two agents can have rank~$n-1$ in the same preference list
  it must hold that the second-most preferred agents of all agents from $U$ are different from each other.
  Thus, using \eqref{eq:rank-sum}, each agent~$w\in W$ must have complete preferences.
  By the symmetry of acceptability, each agent~$u\in U$ must also have complete preferences.

  We are now ready to prove that $P$ is position-wise distinct when
  $n \geq 3$. This is clear if each agent finds only one other agent
  acceptable. Otherwise, there exists at least one unmatched pair of
  agents. We will show the stronger statement that, for each unmatched
  pair~$\{u, w\}$, with $u\in U$ and $w\in W$ and for each
  $z \in [n - 1]$ it holds that
  \begin{align}\label{eq:rank-sum2}
    \rank_u(w) = z \text{ if and only if } \rank_w(u)= n - z.
  \end{align}
  (Note that we can replace ``if and only if'' by ``only if'' to obtain an equivalent statement, but the former is more convenient when we prove it by induction below.)
  To see that \eqref{eq:rank-sum2} implies that $P$ is top-choice, suppose, for the
  sake of contradiction, that there are three distinct
  agents~$x_1,x_2,y$ and an integer~$z$ such that
  $\rank_{x_1}(y) = \rank_{x_2}(y) = z$. Since $M$ is top-choice,
  $z > 0$. Clearly, $z \leq n - 1$. Thus, by \eqref{eq:rank-sum2}
  $\rank_y(x_1) = \rank_y(x_2) = n - z$, a contradiction.

  We show \eqref{eq:rank-sum2} via induction on the rank index~$z\coloneqq \rank_u(w)$, starting with the base case~$z=1$.
  To this end, let $\{u,w\}$ be an unmatched pair.
  To show the ``only if'' part of \eqref{eq:rank-sum2}, assume that $z=\rank_u(w)=1$.
  By \eqref{eq:rank-sum}, it follows that 
  $\rank_{w}(u)\ge n-1$.
  Since the rank of each agent is at most $n-1$,
  it follows that $\rank_{w}(u)=n-1=n-z$.

  For the ``if'' part of the base case
  suppose, towards a contradiction, that there is an unmatched pair~$\{u,w\}$ with $\rank_w(u)=n-1$ but $\rank_{u}(w)\neq 1$.
  By \eqref{eq:rank-sum}, it follows that $\rank_{u}(w) > 1$.
  Since each agent has complete preferences, by the above reasoning,
  there exists an other agent~$u'\in U\setminus \{u\}$ with $\rank_{u'}(w)=1$.
  However, then by the ``if'' part of the base case, it follows that $\rank_{w}(u')=n-1$---a contradiction. Thus, \eqref{eq:rank-sum2} holds when $z = 1$.
  
  For the induction assumption, assume that \eqref{eq:rank-sum2} holds for every index~$z'\le z-1$.
  For the ``only if'' part, consider an unmatched pair~$\{u,w\}$ with $\rank_{u}(w)=z$.
  By \eqref{eq:rank-sum}, it follows that $\rank_{w}(u)\ge n-z$.
  Suppose, for the sake of contradiction, that
  $\rank_w(u) = n-z'$ with $z'< z$.
  By the ``if'' part of the induction assumption, we infer that  $\rank_{u}(w) = z' < z$---a contradiction.

  The ``if'' part of the induction step follows analogously.
\end{proof}
}

\looseness=-1
Now we discuss the trade-offs formalized in the problems regarding \nstability and social optimality.
As mentioned in \cref{ex:egalitarian_cost} even a single swap in the preference profile can improve the egalitarian cost of the stable matching by $\Omega(n^2)$. However, this is not the case when the social optimality is measured by the number of agents who will have a partner in the matching.

\newcommand{\fewmatchedforswaps}{%
 Let $P_1$ and $P_2$ be two preference profiles with $\tau(P_1, P_2) = 1$.
  Let $S_1$ and $S_2$ denote the set of agents that are unmatched by any stable matching of $P_1$ and of $P_2$ respectively. 
  Then, $|(S_1 \setminus S_2) \cup (S_2 \setminus S_1)| \le 2$.%
}
\begin{theorem}%
  \label[theorem]{thm:few-matched-for-swap}
  \fewmatchedforswaps
 \end{theorem}
 
\appendixproofwithstatement{thm:few-matched-for-swap}{\fewmatchedforswaps}{
\begin{proof}

   Without loss of generality, assume that profile $P_2$ is obtained from $P_1$ by swapping agents $w_1$ and $w_2$ in the preference list of agent $u_1$ so that $u_1$ prefers $w_1$ to $w_2$ in $P_1$ and $w_2$ to $w_1$ in~$P_2$. %
  
  By \cref{prop:SMI-matched-agents-the-same}, in order to show the statement, it suffices
  to show that $P_1$ and $P_2$ admit stable matchings~$M_1$ and $M_2$, respectively,
  such that the following is satisfied.
  Let $S_1$ and $_S$ denote the set of agents that is unmatched under $M_1$ and $M_2$, respectively.
  Then,  $|(S_1 \setminus S_2) \cup (S_2\setminus S_1)|\ge 2$.  

  To achieve this, we start with an arbitrary but fixed stable matching, $M_1$, of $P_1$.
  And we will show how to modify~$M_1$ to obtain a stable matching of $P_2$ such that the set of unmatched agents differ by at most two agents.

  Observe that if $M_1$ would not be stable for profile~$P_2$ anymore, then $\{u_1, w_2\}$ is the only possible blocking pair because $P_1$ and $P_2$ differ only by one single swap of the preference list of agent~$u_1$.
  If $M_1$ is not stable for $P_2$, then we modify $M_1$ in the following way.
  We break the pairs $\{u_1, M_1(u_1)\}$ and $\{M_1(w_2), w_2\}$ (if they exist) and we replace them with $\{u_1, w_2\}$. Now, there are two new unmatched agents: $M_1(w_2) \in U$ and $M_1(u_1) \in W$. Further, if we remove these agents from the consideration, the matching would be stable.

  We now proceed as follows. We will perform a sequence of changes to $M_1$. After each change we will keep in the penalty box at most two unmatched agents, one from $U$ and one from $W$, starting with $M_1(w_2)$ and $M_1(u_1)$. Further, each change will keep satisfying the following invariant:
  if we remove the agents contained in the penalty box, then the resulting matching would be stable.
  Let us now describe the way in which we perform the changes. Let $M$ be the matching at the current iteration. We take out an agent $u \in U$ from the penalty box (if such an agent does not exist, we stop). Agent $u$ might be involved in a number of blocking pairs---if it does not, we stop. We take $u$'s most preferred agent $w \in W$ such that $\{u, w\}$ is a blocking pair of the current matching~$M$---we remove $\{M(w), w\}$ from the matching~$M$ and replace it with $\{u, w\}$. Finally, we add $M(w)$ to the penalty box. Clearly, $u$ cannot be involved in any blocking pair, thus, any blocking pair must involve an agent from the penalty box; hence the invariant is indeed satisfied. 

Each such a change replaces a $U$-agent from the penalty box with another $U$-agent. Further each change improves one of the $W$-agents by giving her a more preferred partner. Thus, our procedure must stop at some point. When this is the case we remove the $U$-agent from the penalty box. Then, we perform an analogous procedure but each time replacing a $W$-agent in the penalty box with another $W$-agent. By an analogous arguments, such changes keep the invariant satisfied and the procedure finally stops. 

Clearly, when the procedure stops, there are no blocking pairs. Further, the resulting matching has at most two more agents without partners than $M_1$. %
\end{proof}
}
Repeated application of \cref{thm:few-matched-for-swap} yields that, in order to increase the number of matched agents by $\ell \in \mathds{N}$ in a given stable matching of some profile we have to allow for at least $\ell/2$ swaps. In other words, if a stable matching leaves $s$ agents unmatched, then there is a \gns{$d$} \emph{perfect} matching only if $d \geq s/2$.

\newcommand{\bps}[1]{#1-{nearly} {bp} stable}
\newcommand{\bpstability}[1]{#1-{nearly} {bp} stability}
Let us recall the notion of \myemph{$r$-maximal stability} for the case with linear orders~\cite{DruBou13}: A matching~$M$ is \myemph{$r$-maximally stable}
if for each unmatched pair~$\{u, v\} \notin M$, it holds that
$\min \{\rank_u(M(u)) - \rank_u(v), \rank_v(M(v)) - \rank_v(u)\} \leq
r$. At the first glance, this notion looks quite different from local
$d$-near stability; we show below that in fact they are equivalent.
Moreover, \lnstability{$\locald$} is equivalent to the following
measure of the weight of a blocking pair.
We say that a matching~$M$ is \myemph{\bps{$\locald$}}
if for each blocking pair $b \in \bp(P, M)$, there exists
a profile~$P'_b$ such that $b \notin \bp(P'_b, M)$ and
$\tau(P, P'_b) \leq \locald$.
\newcommand{\thethreeconcepts}{%
Let $P$ be a preference profile without ties, $M$ be a matching for~$P$, and $\locald$
  a nonnegative integer. The following are equivalent.
 \begin{inparaenum}[(i)]
  \item $M$ is $\locald$-maximally stable.
  \item $M$ is locally $\locald$-nearly stable.
  \item $M$ is \bps{$\locald$}.
 \end{inparaenum}
}
\begin{proposition}%
  \label[proposition]{prop:r-maximal-d-bp-local-d}
  \thethreeconcepts
\end{proposition}

\appendixproofwithstatement{prop:r-maximal-d-bp-local-d}{\thethreeconcepts}{
\begin{proof}
  (i) $\Rightarrow$ (ii): Construct a directed graph~$G$ on the
  set~$V$ of agents as follows. For each blocking pair
  $\{u, v\} \in \bp(P, M)$ find the agent, say $u$, such that
  $\rank_u(M(u)) - \rank_u(v) \leq \rank_v(M(v)) - \rank_v(u)$ and add
  the arc $(v, u)$ to~$G$ (that is, add an arc directed
  towards~$u$).

  To obtain a modified profile $P'$ in which $M$ is stable, define,
  for each agent~$u$, a set of swaps in $u$'s preference list as
  follows. Let $B_u$ be the set of agents~$v$, such that $(v, u)$ is
  an arc in~$G$. Pick $w = \argmin_{v \in B_u}\rank_u(v)$.
  Observe that $\rank_u(M(u)) - \rank_u(w) \leq d$ since $(w, u)$ is an arc in~$G$
  and since $M$ is $\locald$-maximally stable. Swap $M(u)$ in~$u$'s
  preference list with the agent directly preceding~$M(u)$ until
  $M(u) \succ^{P'}_u w$ in the resulting profile~$P'$. In this way,
  for each agent we have made at most~$\locald$ swaps to obtain~$P'$.

  Note that throughout the swapping process no new blocking pairs are
  introduced, that is, $\bp(P', M) \subseteq \bp(P, M)$, because in
  each step only a matched agent improves her rank in its matched
  partner's preference list. Moreover, for each blocking pair
  $\{u, v\} \in \bp(P, M)$, we have either $M(u) \prec^{P'}_u v$ or
  $M(v) \prec^{P'}_v u$ by construction. Thus, $M$ is stable with
  respect to~$P'$, showing that $M$ is \lns{$\locald$}.

  (ii) $\Rightarrow$ (iii): Let $P'$ be a profile as promised by \lnstability{$\locald$}.
  For each blocking pair~$\{u, v\} \in \bp(P, M)$
  either $M(v)$ has been swapped before $u$ in $v$'s preference list
  or $M(u)$ has been swapped before $v$ in $u$'s preference list using
  in either case at most~$d$ swaps. Restricting $P'$ to only these
  swaps yields a profile $P'_{\{u, v\}}$ as required by \bps{$\locald$}.

  (iii) $\Rightarrow$ (i): Let $\{u, v\}$ be an unmatched pair under $M$.
  If $\{u, v\} \notin \bp(P, M)$, then $\min \{ \rank_u(M(u)) - \rank_u(v), \rank_v(M(v))-\rank_v(u) \} \leq 0 \leq \locald$.
  Otherwise, $\{u, v\} \in \bp(P, M)$. By \bpstability{$\locald$}, there are $\locald$ swaps such that in
  the resulting profile~$P'$ we have $M(u) \succ^{P'}_u v$ or
  $M(v) \succ^{P'}_v u$ because $\{u,v\} \notin \bp(P',M)$.
  Since $\{u,v\}\in \bp(P,M)$, meaning that both $v \succ^{P}_u M(u)$ and
  $u \succ^{P}_v M(v)$ hold,
  and since $\tau(P,P')\le \locald$, we have $\rank_u(M(u)) - \rank_u(v)  \leq \locald$ or
  $ \rank_v(M(v)) - \rank_v(u)  \leq \locald$. In other words, $M$ is
  $\locald$-maximally stable.
\end{proof}
}

\section{A Polynomial-Time Algorithm for Finding Robust Matchings}

In this section we present a polynomial-time algorithm for the \pRMl{} problem. First, in \Cref{sec:preliminaries} we provide a brief overview of tools and results from the literature that we will use in our algorithm.
We remark that all these results are originally stated for \emph{complete} preferences.
Nevertheless, since all stable matchings match the same set of agents~(\cref{prop:SMI-matched-agents-the-same}), 
we can verify that they also hold when the preferences may be incomplete.
The results described in the subsequent sections, starting from \cref{sec:profile-decomposition}, are our original contributions. 

\subsection{Preliminaries}\label{sec:preliminaries}
\appendixsection{sec:preliminaries}
\toappendix{

Recall that a pair~$\{u,w\}$ with $u\in U$ and $w\in W$ is a \emph{stable pair} of a preference profile~$P$ if it is contained in at least one stable matching of~$P$.

\begin{proposition}[{\citep[Theorem~3.4.3]{GusfieldIrving1989}}]\label[proposition]{prop:stable-pairs}
  For each preference profile,
  after $O(n^4)$ preprocessing time,
  determining whether a given set~$Q$ of $t$~pairs is a stable set of~$P$ can be done in $O(t^2)$ time.
\end{proposition}

}
As already observed in the literature, the set of all stable matchings for a given
preference profile forms a lattice--a specific partially ordered set--that is useful in designing
algorithms for finding special kinds of stable matchings. The maximum
and minimum elements are so-called optimal stable matchings: 
Consider a preference profile~$P$ with two sets, $U$ and $W$, of
agents and consider two matchings~$M$ and $M'$. We say that an
agent~$a\in U\cup W$ \myemph{prefers $M$ to $M'$}, denoted as
\myemph{$M \succ_a M'$}, if $\rank_a(M(a)) < \rank_a(M'(a))$.
Similarly, agent~$a$ \myemph{weakly prefers $M$ to $M'$}, denoted as
\myemph{$M \succeq_a M'$}, if $M(a)=M'(a)$ or $\rank_a(M(a)) < \rank_a(M'(a))$.
Accordingly, we say that $M$ is a \myemph{$U$-optimal} (resp.\
\myemph{$W$-optimal}) stable matching if it is stable and there is
\emph{no} other stable matching~$M'$ different from $M$ such that each
agent from $U$ prefers $M'$ to $M$.

It is well-known that $U$-optimal and $W$-optimal stable matchings are unique. %
The concepts of $U$-optimal and $W$-optimal stable matchings are already illustrated in \cref{ex:intro_example}. %

\toappendix{
\Cref{prop:diff-partners-diff-prefs} below shows that, when comparing two stable matchings, an improvement of an agent $u \in U$ always comes at the cost of some other agent from $W$.

\begin{proposition}[{\cite[Theorem~1.3.1, Chapter~1.4.2]{GusfieldIrving1989}}]\label[proposition]{prop:diff-partners-diff-prefs}
  Let $M_1$ and $M_2$ be two stable matchings of the same preference profile with (possibly) incomplete preferences, and let $u$ and $w$ be two agents such that
  $M_1(u)=w$ but $M_2(u)\neq w$.
  Then, $M_1 \succ_u M_2$  if and only if $M_2 \succ_w M_1$.
\end{proposition}

Finally, we recall that the famous Gale/Shapley algorithm always finds the $U$-optimal (or, depending on the variant of the algorithm used, the $W$-optimal) stable matching.

\begin{proposition}[\cite{GaleShapley1962},{\cite[Chapter~1.4.2]{GusfieldIrving1989}}]
  The $U$-optimal and the $W$-optimal stable matchings of a preference profile always exist and can be found in $O(n^2)$ time.
\end{proposition}
}

We now review a fundamental object, \myemph{rotations}, and some well-known structural properties of stable matchings. These concepts will play an instrumental role in our analysis in the subsequent sections. For more details, we refer to the exposition by~\citet{GusfieldIrving1989}.

\begin{definition}[Successor agent, rotations, and rotation elimination]\label[definition]{def:rotations}
  Let $P$ be a preference profile with two disjoint sets of agents,~$U$ and $W$, and with (possibly) incomplete preferences.
  Given a stable matching $M\in \sm(P)$, for each agent~$u\in U$, we define its \myemph{successor~$\sucw_M(u)$} as the \emph{first} (after $M(u)$) agent~$w$ on the preference list of $u$ such that $w$ is matched under $M$ and prefers $u$ to its partner~$M(w)$.
  \iflong
  We illustrate the concept of the successor below:
  \begin{align*}
     u\colon \ldots M(u) \ldots \sucw_M(u) \ldots \qquad \sucw_M(u)\colon \ldots u \ldots M(\sucw_M(u)) \ldots
  \end{align*}  
  \fi
  \looseness=-1 A sequence~$\rho=((u_{0},w_{0}), (u_1,w_1), \ldots, (u_{r-1}, w_{r-1}))$ of pairs is called a \myemph{rotation} if there exists a stable matching $M\in \sm(P)$ such that 
  for each $i\in \{0,1,\ldots,r-1\}$ we have $(u_i,w_i)\in U\times W$, $M(u_i)=w_i$, and $\sucw_M(u_i)=w_{i+1}$ (index $i+1$ taken modulo $r$).
  We say rotation~$\rho$ is \myemph{exposed} in~$M$.

  We use the notation~$M/\rho$ to refer to the matching resulting from $M$ by replacing each pair~$\{u_i,w_i\}$ with $\{u_i,w_{i+1}\}$.
  Formally,
  \iflong  \begin{align*}
    M/\rho = M \setminus \{\{u_i,w_i\} \mid 0\le i \le r-1\} \cup \{\{u_i,w_{i+1}\} \mid 0\le i \le r-1\}.
  \end{align*}
  \else
    $M/\rho = M \setminus \{\{u_i,w_i\} \mid 0\le i \le r-1\} \cup \{\{u_i,w_{i+1}\} \mid 0\le i \le r-1\}$.
    \fi           
  The transformation of $M$ to $M/\rho$ is called the \myemph{elimination of $\rho$ from $M$}.
\end{definition}

Eliminating a rotation from a stable matching results in another stable matching~\cite{GusfieldIrving1989}.
The concepts from \Cref{def:rotations} are illustrated in the example below.

\begin{example}\label[example]{ex:rotations}
  Consider the profile in \cref{ex:intro_example}. %
  Relative to $M_1$, agent~$w_2$ is the first agent among all agents in the preference list of~$u_1$ that prefer $u_1$ to their respective partners. %
  Thus, $\sucw_{M_1}(u_1)=w_2$.
  Sequence $((u_1,w_2), (u_2,w_3), (u_3, w_4), (u_4, w_1))$ is the only rotation exposed in~$M_1$.~\hfill$\diamond$
\end{example}

Interestingly, while a given profile with $O(n)$ agents may admit exponentially ($O(n!)$) many different stable matchings, the number of rotations is polynomial ($O(n^2)$)~\cite[Corollary~3.2.1]{GusfieldIrving1989}. %
Indeed, the set of all rotations gives a compact representation of the set of all possible stable matchings for a given preference profile.
\iflong To determine robustness efficiently, we will use this representation intensely.
\fi

\toappendix{
The next structural result concerns the properties of a stable matching after eliminating a rotation $\rho$.

\begin{proposition}[{\citep[Theorem~2.5.6, Lemma~3.2.1, Lemma~3.2.2]{GusfieldIrving1989}}]\label[proposition]{prop:stable-pairs+rotations}
  Consider a preference profile $P$ with two disjoint sets of agents,~$U$ and $W$,
  and with (possibly) incomplete preferences. 
  For each two agents~$u\in U$ and $w\in W$, the following holds; recall that $x \succeq y$ means that
  $x=y$ or $x\succ y$.

  \begin{compactenum}[(i)]

    \item\label{prop:stable-pair-char} $\{u,w\}$ is in a stable matching if and only if either it is in the $W$-optimal stable matching or
    $(u,w)$ belongs to some rotation.

    \item\label{prop:at-most-1-from-w} There is at most one rotation~$\rho$ with $\rho=((u_0,w_0),\ldots, (u_{r-1}, w_{r-1}))$ such that for some~$i\in \{0,\ldots,r-1\}$
    it holds that $u=u_i$ and $w_{i} \succeq_{u} w\succ_{u} w_{i+1}$.
    \item\label{prop:at-most-1-up} There is at most one rotation~$\rho$ with $\rho=((u_0,w_0),\ldots, (u_{r-1}, w_{r-1}))$ such that for some~$i\in \{0,\ldots,r-1\}$
    it holds that $u=u_i$ and $w = w_{i+1}$.
    \item\label{prop:at-most-1-from-m} There is at most one rotation~$\rho$ with $\rho=((u_0,w_0),\ldots, (u_{r-1}, w_{r-1}))$ such that for some~$i\in \{0,\ldots,r-1\}$
    it holds that $w=w_i$ and $u_{i-1} \succ_{w} u\succeq_{w} u_{i}$.
  \end{compactenum}
\end{proposition}

}
Now we are ready to introduce the notion of the rotation poset of a given preference profile~$P$. As we will see later on, each stable matching can be obtained by performing a number of eliminations of rotations on the $U$-optimal stable matching.
When starting from $U$ some rotations can be exposed only after some other have been already eliminated. This induces a partial order on rotations and defines the rotation poset.

\newcommand{\pred}{\ensuremath{\rhd}}
\newcommand{\predr}{\ensuremath{\unrhd}}
\begin{definition}[Predecessors of rotations, the rotation poset, and the rotation digraph]\label[definition]{def:digraph}
  Let $\pi$ and $\rho$ be two rotations for a preference profile~$P$.
  We say that $\pi$ is a \myemph{predecessor} of $\rho$,
  written as \myemph{$\pi \pred^P \rho$},
  if no stable matching in which $\rho$ is exposed can be obtained from the $U$-optimal stable matching by a sequence of eliminations of rotations without eliminating~$\pi$ first.
  The reflexive closure of the relation~$\pred^P$, denoted as $\predr^P$, defines a partial order on the set of all rotations and is called the \myemph{rotation poset} for~$P$.
  We abbreviate the name of a subset of the poset that is closed under predecessors as a \myemph{closed subset}.
  
  An alternative representation of the rotation poset~$\predr(P)$ is through an acyclic directed graph, called \myemph{rotation digraph of $P$} and written as \myemph{$G(P)$}, whose vertex set is the set of rotations of $P$, and there is a direct arc from rotation $\pi$ to rotation $\rho$ if and only if~$\pi$ precedes~$\rho$ and there is no other rotation~$\sigma$ such that $\pi \pred^P \sigma \pred^P \rho$.
\end{definition}

\toappendixalter{
}
  {The following example illustrates the rotation poset of profile given in \cref{ex:intro_example}.}{
\begin{example}\label[example]{ex:rotation-poset}
Let us consider the profile~$P$ given in \cref{ex:intro_example} again.
As we mentioned in \cref{ex:rotations}, rotation~$\pi_1=((u_1, w_2), (u_2, w_3), (u_3, w_4), (u_4, w_1))$ is the only rotation
exposed in the $U$-optimal stable matching~$M_1$.
After eliminating~$\pi_1$ from~$M_1$, we obtain the stable matching~$M_3=M_1/\pi_1$.
One can also verify that the sequence~$\pi_2=((u_1,w_3), (u_3,w_1))$ and $\pi_3=((u_2,w_4),(u_4,w_2))$
are the only two rotations exposed in stable matching~$M_3$.
After eliminating~$\pi_2$ from~$M_3$, we obtain the stable matching~$M_4=M_3/\pi_2$.
After eliminating~$\pi_3$ from $M_3$, we obtain the stable matching~$M_5=M_3/\pi_3$.
After eliminating rotation~$\pi_3$ from~$M_4$ or eliminating the rotation~$\pi_2$ from~$M_5$, we obtain the $W$-optimal stable matching~$M_2$.
\columnratio{.42}
  \begin{paracol}{2}

  Since $\pi_1$ is only exposed in $M_1$ and since $\pi_2$ and $\pi_3$ are only exposed after the elimination of $\pi_1$ 
  we have that $\pi_2$ and $\pi_3$ are two (direct) successors of $\pi_1$.

  The Hasse diagram on the right-hand side depicts how the stable matchings for $P$ are related to each other, in terms of dominance with respect to the satisfaction of the agents from $U$.
  Herein, the matchings are depicted as matrices such that each pair in a matching is represented by a column in the corresponding matrix.

  \switchcolumn

  {\centering
    \begin{tikzpicture}
      [every text node part/.style={align=right}, match/.style={minimum size=3ex,inner sep=7pt}, label/.style={inner sep=2pt}, scale = 1.2]
    \node[match] at (0,0) (M1){$M_1\colon \begin{pmatrix} u_1 & u_2 & u_3 & u_4\\ w_2 & w_3& w_4&  w_1\end{pmatrix}$};
    \node[match] at (0,-1.5) (M3){$M_3\colon \begin{pmatrix} u_1 & u_2 & u_3 & u_4\\ w_3 & w_4& w_1&  w_2\end{pmatrix}$};
    \node[match] at (-2,-3) (M4){$M_4\colon \begin{pmatrix} u_1 & u_2 & u_3 & u_4\\ w_1 & w_4& w_3&  w_2\end{pmatrix}$};
    \node[match] at (2,-3) (M5){$M_5\colon \begin{pmatrix} u_1 & u_2 & u_3 & u_4\\ w_3 & w_2& w_1&  w_4\end{pmatrix}$};
    \node[match] at (0,-4.6) (M2) {$M_2\colon \begin{pmatrix} u_1 & u_2 & u_3 & u_4\\ w_1 & w_2& w_3&  w_4\end{pmatrix}$};

    \draw (M1) edge[shorten <=-0.2cm, shorten >=-0.2cm] node[label, left] {$\pi_1$} (M3);
    \draw (M3) edge[shorten <=-0.2cm, shorten >=-0.2cm] node[label,above left] {$\pi_2$} (M4);
    \draw (M3) edge[shorten <=-0.2cm, shorten >=-0.2cm] node[label,above right] {$\pi_3$} (M5);
    \draw (M4) edge[shorten <=-0.2cm, shorten >=-0.2cm] node[label,below left] {$\pi_3$} (M2);
    \draw (M5) edge[shorten <=-0.2cm, shorten >=-0.2cm] node[label,below right] {$\pi_2$} (M2);
  \end{tikzpicture} 
  \par
}
\hfill $\diamond$
\end{paracol}
\end{example}
}

Finally, let us describe a central result from the literature that relates rotations and stable pairs.

\begin{proposition}[{\cite[Theorem~2.5.7, Lemma~3.3.2]{GusfieldIrving1989}}]\label[proposition]{prop:rotations+stable-matchings}
  Let $R$ denote the set of all rotations of a preference profile~$P$,
  and let $G(P)$ denote the rotation digraph of $P$.
  \begin{compactenum}[(i)]
    \item\label{rot:closedsubet-stable} A matching~$M$ is a stable matching of $P$ if and only if there is a closed subset of rotations~$R'\subseteq R$ with respect to the precedence relation $\pred^P$ such that $M$ can be generated by taking the $U$\nobreakdash-optimal stable matching and by eliminating the rotations in $R'$ in an order consistent with~$\pred^P$.
    \item\label{rot:runtime} The rotation set~$R$ and the rotation digraph~$G(P)$ can be computed in $O(n^2)$~time.
  \end{compactenum}
\end{proposition}

\subsection{Profile Characterization}\label{sec:profile-decomposition}
\appendixsection{sec:profile-decomposition}

\looseness=-1
For a given profile $P$ with $O(n)$~agents and a given swap distance bound~$d=O(n)$
there are exponentially many profiles which are within swap distance~$d$ to~$P$.
In this section, we show that we do not need to consider all of them in order to find a
$d$-robust matching.
Instead, we characterize them based on pairs of shifts.
Briefly put, a \myemph{shift} is a set of swaps which all involve swapping the same agent forward in a single preference list.
We describe a polynomial-size family of ``relevant'' profiles and we characterize each of them
through a pair of shifts---such a pair of shifts will be represented by a quadruple of agents.
Intuitively, if there exists a profile $P'$ witnessing that a certain matching~$M$
is \emph{not} $d$-robust, and if $P'$ contains more than two shifts with respect to the original profile $P$ , then $P'$ can be represented by a number of profiles which satisfy the following.
Each of these profiles contains only two shifts and one of them witnesses that $M$ is not $d$-robust.
Later on, we will show that the quadruples which characterize the profiles relevant
for checking $d$-robustness are closely related to certain
rotations---this will give us the tools essential for constructing a polynomial-time algorithm.

\begin{definition}[Stable quadruples and swap sets]\label[definition]{def:swap-sets}
  Let $P=((\succ_u^P)_{u\in U}, (\succ_w^P)_{w\in W})$ be a preference
  profile for the two agent sets~$U$ and $W$. A \myemph{stable
    quadruple} (with respect to~$P$) is a quadruple $(u^*,w^*,u,w)$ of four distinct agents
  with $u^*,u\in U$ and $w^*,w\in W$ such that there exists a stable matching for~$P$ that contains both
  $\{u^*,w\}$ and $\{u,w^*\}$.

  For each stable quadruple~$q=(u^*,w^*,u,w)$ of $P$,
  we define the \myemph{swap set} associated with $P$ and~$q$, denoted as \myemph{$\shifts(P,q)$},
  as the smallest set of swaps
  which involve the following two kinds of shifts in the preference lists of $u^*$ and $w^*$.
  \begin{inparaenum}
    \item The first kind of shifts puts agent~$w^*$ forward until she is right in front of $w$ in the preference list of $u^*$, and
    \item the second kind of shifts puts agent~$u^*$ forward until she is right in front of $u$ in the preference list of $w^*$.
  \end{inparaenum}
  If $w^*$ (resp.\ $u^*$) is already in front of $w$ (resp.\ $u$), then no swap in the corresponding preference list is needed.
  Formally,
  \iflong
\begin{align*}
  \shifts(P,q) \coloneqq \bigcup_{y\in W\colon w \succeq^P_{u^*} y \succ^P_{u^*} {w^*}}\{({u^*},\{y,{w^*}\})\} \cup \bigcup_{x\in U\colon u \succeq^P_{w^*} x \succ^P_{w^*} {u^*}}\{({w^*},\{{u^*},x\})\}.
\end{align*}
\else

{\centering
  $  \shifts(P,q) \coloneqq \bigcup_{y\in W\colon w \succeq^P_{u^*} y \succ^P_{u^*} {w^*}}\{({u^*},\{y,{w^*}\})\} \cup \bigcup_{x\in U\colon u \succeq^P_{w^*} x \succ^P_{w^*} {u^*}}\{({w^*},\{{u^*},x\})\}.$
  \par}

\fi
Herein, the notation~$x \succeq y$ means either $x=y$ or $x\succ y$.
Further, let \myemph{$\shiftss(\succ^P_{u^*},q)$} denote the preference list resulting from starting with~$\succ^P_{u^*}$ and performing the swaps from $\shifts(P,q)$ that involve the preference list of $u^*$.
Analogously, let $\shiftss(\succ^P_{w^*},q)$ denote the preference list resulting from starting with $\succ^P_{w^*}$ and performing the swaps from $\shifts(P,q)$ that involve the preference list~$\succ^P_{w^*}$.
Now, let \myemph{$P[\shifts(P,q)]$} denote the preference profile resulting from $P$ by replacing the preference lists of $u^*$ and $w^*$ with $\shiftss(\succ^P_{u^*},q)$ and $\shiftss(\succ^P_{w^*}, q)$, respectively.
Formally,
\iflong\begin{align*}
  P[\shifts(P,q)] \coloneqq ((\succ_{x}^P)_{x\in U \setminus \{u^*\}}+\shiftss(\succ^P_{u^*},q), (\succ_{y}^P)_{y\in W \setminus \{w^*\}}+\shiftss(\succ_{w^*},q)).
       \end{align*}
       \else

       {
         \centering
  $P[\shifts(P,q)] \coloneqq ((\succ_{x}^P)_{x\in U \setminus \{u^*\}}+\shiftss(\succ^P_{u^*},q), (\succ_{y}^P)_{y\in W \setminus \{w^*\}}+\shiftss(\succ_{w^*},q)).$
     \par  }
       
       \fi
\end{definition}

\begin{example}\label[example]{ex:swap-sets}
  For an illustration, let us consider the profile given in \cref{ex:intro_example},
  denoted as $P=((\succ^P_{u_i})_{u_i\in U}, (\succ^{P}_{w_i})_{w_i \in W})$, and the following stable quadruple~$q=\rotatione$; note that $\{\{u_3,w_1\}, \{u_4,w_2\}\}$ is a stable set~(see $M_3$). 
  The swap set~$\shifts(P,q)$ consists of two swaps; both involve changing $u_3$'s preference list: %
  $\shifts(P,q) = \{(u_3, \{w_2,w_3\}), (u_3,\{w_2,w_1\})\}$.

  By performing the swaps given in $\shifts(P,q)$ on~$\succ^P_{u_3}$ and on the preference profile, we obtain that
     $\shiftss(\succ^P_{u_3},q) = \{u_3\colon w_4 \succ w_2 \succ w_1 \succ w_3\}$, and  %
    $P[\shifts(P,q)]  = ((\succ^P_{u_1}, \succ^P_{u_2},  \succ^P_{u_3}, \shiftss(\succ^{P}_{u_4},q)), (\succ^p_{w_1}, \succ^P_{w_2}, \succ^P_{w_3},  \succ^P_{w_4}))$.
  Finally, we observe that in $P[\shifts(P,q)]$, $u_4$ prefers $w_1$ to $w_2$ and $w_1$ prefers $u_4$ to $u_3$.
  \mbox{}\hfill$\diamond$
\end{example}

\noindent A stable quadruple~$q$ and the corresponding profile~$P[\shifts(P,q)]$ satisfy the following properties.

\begin{observation}\label[observation]{obs:sh-P-props}
  Let $P$ be a preference profile over the two agent sets $U$ and $W$, let $q$ be a stable quadruple with $q=(u^*,w^*,u,w)$ and 
  let $Q=P[\shifts(P,q)]$ denote the preference profile after performing the swaps in the set~$\shifts(P,q)$.
  \begin{compactenum}[(i)]
    \item\label{obs:shift-other-agents}
    Each agent~$x\in U\cup W \setminus \{u^*,w^*\}$ other than $u^*$ and $w^*$
    has $\succ^Q_{x}=\,\succ^P_{x}$.
    \item\label{obs:shift-u}
    If $w^* \succ^P_{u^*} w$, then $\succ^{Q}_{u^*} = \succ^{P}_{u^*}$;
    otherwise, for each two distinct agents~$y,z \in W\setminus \{w^*\}$ the following holds.
    \begin{inparaenum}[(a)]
      \item $y \succ_{u^*}^Q z$ iff.\ $y \succ_{u^*}^P z$,
      \item $y \succ_{u^*}^Q w^*$ iff.\ $y \succ_{u^*}^P w$,
      \item $w^* \succ_{u^*}^Q y$ iff.\ $w \succeq^{P}_{u^*} y$.
    \end{inparaenum}
    \item\label{obs:shift-w}
    If $u^* \succ^P_{w^*} u$, then $\succ^{Q}_{w^*} = \succ^{P}_{w^*}$;
    otherwise, for each two distinct agents~$y,z \in U\setminus \{u^*\}$ the following holds.
    \begin{inparaenum}[(a)]
      \item $y \succ_{w^*}^Q z$ iff.\ $y \succ_{w^*}^P z$,
      \item $y \succ_{w^*}^Q u^*$ iff.\ $y \succ_{w^*}^P u$,
      \item $u^* \succ_{w^*}^Q y$ iff.\ $u \succeq^{P}_{w^*} y$.
    \end{inparaenum}
    \item\label{obs:blocking} In $P[\shifts(P,q)]$, agent~$u^*$ prefers $w^*$ to $w$, and agent~$w^*$ prefers $u^*$ to $u$.
  \end{compactenum}
\end{observation}

\looseness=-1
Informally, we will argue that,
to find a robust matching, it suffices to focus on profiles obtained by performing swaps induced by certain quadruples. Further, we will show that for each quadruple $q=(u^*,w^*,u,w)$ in profile $P[\shifts(P,q)]$ we only need to ensure that $\{u^*,w^*\}$ is not a blocking~pair. %

\toappendix{
  \ifshort
  Recall that in the main part, we aim to only focus on quadruples~$q=(u^*,w^*,u,w)$ and their corresponding profiles~$P[\shifts(P,q)]$.
  \fi
  \Cref{lem:profile-only-possible-blocking-pair,lem:profile-decompose} below formalize our intuition that $\{u^*,w^*\}$ is the only possible blocking pair in $P[\shifts(P,q)]$.

\newcommand{\onlypossiblebp}{%
  Consider a preference profile $P$ and a stable matching $M$ for~$P$. For a stable quadruple $q=(u^*,w^*,u,w)$, pair~$\{u^*,w^*\}$ is the only possible blocking pair of $M$ in $P[\shifts(P,q)]$.
}
\begin{lemma}\label[lemma]{lem:profile-only-possible-blocking-pair}
  \onlypossiblebp
\end{lemma}
\begin{proof}
  Let $Q=P[\shifts(P,q)]$. Suppose, towards a contradiction, that $M$ admits a blocking pair~$\{x,y\}$ with $x\in U$ and $y\in W$ in profile~$Q$ and that $\{x,y\}\neq \{u^*,w^*\}$.
  Since $Q$ differs from $P$ only in the preference lists of $u^*$ and $w^*$ and since $M$ is stable in $P$, it follows that either~$x=u^*$ or $y=w^*$.
  If $x=u^*$, implying that $\{u^*,y\}$ is blocking $M$ in $Q$, 
  then it holds that $y \succ^{Q}_{u^*} M(u^*)$ and
  $u^* \succ^Q_{y} M(y)$.
  However, since $\{x,y\}\neq \{u^*,w^*\}$ it follows that $y\neq w^*$, 
  and that  $y \succ^{{P}}_{u^*} M(u^*)$  (see \cref{obs:sh-P-props}\eqref{obs:shift-u}) and
  $u^* \succ^{P}_{y} M(y)$ (see \cref{obs:sh-P-props}\eqref{obs:shift-other-agents}).
  This implies that $\{u^*,y\}$ is also blocking~$M$ in~${P}$, a contradiction to $M$ being stable in ${P}$.  
  Analogously, we can derive a contradiction for the case of $x\neq u^*$ and $y = w^*$.
\end{proof}

\newcommand{\profiledecompose}{%
  Let $P_1$ and $P_2$ be two preference profiles for the same two disjoints sets~$U$ and $W$, and let $M\in \sm(P_1)$ be a stable matching of $P_1$.
  Let~$\{u^*,w^*\}\in \bp(P_2,M)$ be a blocking pair for $P_2$ with $u^*\in U$ and $w^*\in W$. Define $q = (u^*,w^*,M(w^*),M(u^*))$. %
  The following holds.
  \begin{inparaenum}[(i)]
    \item\label{cond:only-bp} $\bp(P[\shifts(P_1, q)], M) = \{\{u^*,w^*\}\}$.
    \item\label{cond:nmore-swaps} $|\shifts(P_1,q)| \le \tau(P_1,P_2)$.
  \end{inparaenum}
}
\begin{lemma}\label[lemma]{lem:profile-decompose}
  \profiledecompose
\end{lemma}

  \begin{proof}
  To show the first statement, assume that $M$ is not stable in $P_2$ and let $\{u^*,w^*\}$ be a blocking pair of $M$ for $P_2$.
  Set $Q=P[\shifts(P_1,q)]$.

  By \cref{obs:sh-P-props}(\ref{obs:blocking}), we immediately get that 
  $\{u^*,w^*\}$ is blocking $M$ in profile~$Q$.
  The fact that $\{u^*,w^*\}$ is the only blocking pair follows from \cref{lem:profile-only-possible-blocking-pair}.
  
  Now let us consider the second statement.
  By the definition of swap sets on $q$, we have:
  \[|\shifts({P_1},q)|=\max(\rank_{u^*}(w^*, \succ_{u^*}^{P_1}) - \rank_{u^*}(M(u^*), \succ_{u^*}^{P_1}),0) + \max(\rank_{w^*}(u^*,\succ^{P_1}_{w^*})-\rank_{w^*}(M(w^*), \succ_{w^*}^{P_1}), 0).\]
  Since $\{u^*,w^*\}$ is blocking $M$ in $P_2$ but $M$ is stable for $P_1$, it holds that
  \begin{align*}
    w^* \succ^{P_2}_{u^*} M(u^*) \text{ and } u^* \succ^{P_2}_{w^*} M(w^*), \text{ while }
    M(u^*) \succ^{P_1}_{u^*} w^* \text{ or } M(w^*) \succ^{P_1}_{w^*} u^*.
  \end{align*}
  Thus, 
  \begin{align*}
    \tau(P_1,P_2)  & \ge \max(\rank_{u^*}(w^*, \succ_{u^*}^{P_1}) - \rank_{u^*}(M(u^*), \succ_{u^*}^{P_1}),0) + {} \\ & \quad\hfill + \max(\rank_{w^*}(u^*,\succ^{P_1}_{w^*})-\rank_{w^*}(M(w^*), \succ_{w^*}^{P_1}), 0)\\
    & = |\shifts({P_1},q)|.
  \end{align*}
  proving the statement.
\end{proof}
}

Finally, the following lemma summarizes the informal intuition that we provided so far---it shows that when searching for a $d$-robust matching, we only need to focus on some relevant profiles which are close to the initial one.

\newcommand{\dswaps}{%
  Let $P_0$ be a preference profile for two disjoint sets of agents, $U$ and $W$,
  and let $d\in \mathds{N}$ be a non-negative integer.
  A matching~$M$ is $d$-robust for profile~$P_0$ if and only if for each stable quadruple~$q=(u^*,w^*, u,w)$ of $P_0$ 
  such that $|\shifts(P_0, q)| \le d$,
  matching~$M$ is also stable in $P[\shifts(P_0,q)]$.
}

\begin{lemma}%
  \label[lemma]{lem:d-swaps-profiles}
  \dswaps
\end{lemma}

\appendixproofwithstatement{lem:d-swaps-profiles}{\dswaps}{
  \begin{proof}
  The ``only if'' direction is straight-forward because $M$ is stable in each profile~$P$ with $\tau(P_0,P)\le d$ and for each stable quadruple~$q=(u^*,w^*, u,w)$ %
  such that $|\shifts(P_0, q)| \le d$, it holds that $\tau(P_0, P[\shifts(P_0,q)]) = |\shifts(P_0, q)| \le d$.

  For the ``if'' direction, assume that there is a matching, called $M$,
  such that for each stable quadruple~$q=(u^*,w^*, u,w)$ 
  with $|\shifts(P_0, q)| \le d$,
  matching~$M$ is stable in $P[\shifts(P_0,q)]$.
  Suppose, for the sake of contradiction, that there is a preference profile~$P$ with $\tau(P_0, P) \le d$ such that $M$ is \emph{not} stable in $P$.
  Let $\{x,y\}$ be a blocking pair of $M$ in $P$ with $x\in U$ and $y \in W$.
  Now let us consider the quadruple~$q'=(x,y,M(y),M(x))$.
  Note that $q'$ is a stable quadruple with respect to~$P_0$ since $M$ is stable for~$P_0$. Since $\{x,y\} \in \bp(P, M)$, by \cref{lem:profile-decompose}\eqref{cond:only-bp}, it follows that
  $\bp(P[\shifts(P_0,q)], M) = \{\{x,y\}\}$ and, by \cref{lem:profile-decompose}\eqref{cond:nmore-swaps}, $|\shifts(P_0,q)| \le \tau(P_0,P) \le d$---a contradiction to our assumption.
\end{proof}
}

\subsection{Relation Between Stable Quadruples and Rotations}\label{sub:relation-squadruple-rotation}
\appendixsection{sub:relation-squadruple-rotation}
\looseness=-1
Before we state our central results, we need one more element: In this subsection we define two specific rotations corresponding to a stable quadruple and we investigate their properties pertaining to robustness. The results stated in this subsection might look quite technical, yet we deliberately chose these particular formulations as we believe they make the analysis of our algorithm %
transparent.

\begin{definition}[$\pi(q)$ and $\rho(q)$ for a stable quadruple~$q$]\label[definition]{def:two-specific-rots}
  Let $P_0=(\succ_{x}^{P_0})_{x\in U \cup W}$ be a preference profile with two sets of agents,~$U$ and $W$,
  and consider a stable quadruple~$q=(u^*,w^*,u,w)$.
  
  We use the notation~\myemph{$\pi(q)$} to refer to a rotation~$\pi\coloneqq ((u'_0,w'_0), \ldots, (u'_{r-1},w'_{r-1}))$
  with $u^*=u'_i$ (for some $i\in \{0,\ldots,r-1\}$) that fulfills the following conditions.
  \begin{inparaenum}[(i)]
    \item If $w^* \succ^{P_0}_{u^*} w$, then
    $w^*=w'_{i}$ or $w'_i \succ^{P_0}_{u^*} w^* \succ^{P_0}_{u^*} w'_{i+1}$; %
    \item Otherwise, meaning that $w \succ^{P_0}_{u^*}w^*$, then $w = w'_{i+1}$.
  \end{inparaenum}

  We use the notation~\myemph{$\rho(q)$} to refer to a rotation~$\rho\coloneqq ((u'_0,w'_0), \ldots, (u'_{r-1},w'_{r-1}))$
  with $w^*=w'_i$ (for some $i\in \{0,\ldots,r-1\}$) that fulfills the following conditions.
  \begin{inparaenum}[(i)]
    \item If $u^* \succ^{P_0}_{w^*} u$, then
    $u^*=u'_{i-1}$ or $u'_{i-1} \succ^{P_0}_{w^*} u^* \succ^{P_0}_{w^*} u'_{i}$; %
    \item Otherwise, meaning that $u \succ^{P_0}_{w^*}u^*$, then $u=u_i'$.
  \end{inparaenum}

\end{definition}

The below figure illustrates the two specific rotations; recall that for two agents~$x$ and $y$, the expression~``$x\succeq y$'' means that $x=y$ or $x \succ y$.

\appendixfigure{def:two-specific-rots}{%
\noindent  \cref{appfigure:def:two-specific-rots} illustrates the two specific rotations.}{Illustration of \cref{def:two-specific-rots}}{
   \tikzstyle{matrixstyle} = [matrix of math nodes,,ampersand replacement=\&,column sep=-4pt]
  \begin{tikzpicture}[>=stealth]
   \def \xtt {0.6}
   \def \xst {0.2}
   \node[] (ucase1) {\textbf{Case~(i): $\boldsymbol{u^*}$ prefers $\boldsymbol{w^*}$ to $\boldsymbol{w}$, } i.e.\ $\boldsymbol{w^*\succ^{P_0}_{u^*}w}$.};
     \node[right = 0pt of ucase1] (ui1) {Then, $u^*:$};
     \matrix[right = -5pt of ui1, matrixstyle] (ui1p) 
     {w'_i \& \succeq \& w^* \& \succ \& w'_{i+1}. \\};

    \draw[->,rounded corners] ($(ui1p-1-1)+(0,-\xst)$) -- ($(ui1p-1-1)+(0,-\xtt)$) -- node[midway,fill=white] {$\pi(q)$}  ($(ui1p-1-5)+(0,-\xtt)$) -| ($(ui1p-1-5)+(0,-\xst)$);
    
    \node[below = 8ex of ucase1.west, anchor=west] (ucase2) {\textbf{Case~(ii): $\boldsymbol{u^*}$ prefers $\boldsymbol{w}$ to $\boldsymbol{w^*}$, }  i.e.\ $\boldsymbol{w\succ^{P_0}_{u^*}w^*}$.};
    \node[right = -5pt of ucase2] (ui2) {Then, $u^*:$};
    \matrix[right = -5pt of ui2, matrixstyle] (ui2p) 
    {w'_i \& \succ \& w = w'_{i+1} \& \succ \& w^*.\\};

    \draw[->,rounded corners] ($(ui2p-1-1)+(0,-\xst)$) -- ($(ui2p-1-1)+(0,-\xtt)$) -- node[midway,fill=white] {$\pi(q)$}  ($(ui2p-1-3)+(0,-\xtt)$) -| ($(ui2p-1-3)+(0,-\xst)$);

    \node[below = 8ex of ucase2.west, anchor=west] (wcase1) {\textbf{Case~(i): $\boldsymbol{w^*}$ prefers $\boldsymbol{u^*}$ to $\boldsymbol{u}$, } i.e.\ $\boldsymbol{u^* \succ^{P_0}_{w^*} u}$.};
    \node[right = 0pt of wcase1] (wi1) {Then, $w^*:$};
    \matrix[right = -5pt of wi1, matrixstyle] (wi1p) 
    {u'_{i-1} \& \succeq \& u^* \& \succ \& u'_{i}. \\};

    \draw[->,rounded corners] ($(wi1p-1-5)+(0,-\xst)$) -- ($(wi1p-1-5)+(0,-\xtt)$) -- node[midway,fill=white] {$\rho(q)$}  ($(wi1p-1-1)+(0,-\xtt)$) -| ($(wi1p-1-1)+(0,-\xst)$);

    \node[below = 8ex of wcase1.west, anchor=west] (wcase2) {\textbf{Case~(ii): $\boldsymbol{w^*}$ prefers $\boldsymbol{u}$ to $\boldsymbol{u^*}$, }  i.e.\  $\boldsymbol{u \succ^{P_0}_{w^*} u^*}$.};
    \node[right = 0pt of wcase2] (wi2) {Then, $w^*:$};
    \matrix[right = -5pt of wi2, matrixstyle] (wi2p) 
    {u'_{i-1} \& \succ \& u  =  u'_{i} \& \succ \& u^*. \\};

    \draw[->,rounded corners] ($(wi2p-1-3)+(0,-\xst)$) -- ($(wi2p-1-3)+(0,-\xtt)$) -- node[midway,fill=white] {$\rho(q)$}  ($(wi2p-1-1)+(0,-\xtt)$) -| ($(wi2p-1-1)+(0,-\xst)$);
  \end{tikzpicture}
}

Rotations $\pi(q)$ and $\rho(q)$ can be informally described as follows. Consider the preference profile $Q = P[\shifts(P_0, q)]$.
Rotation~$\pi(q)$ is the first rotation (according to the precedence relation on rotations) that moves the partner of $u^*$ from $w^*$ or from an agent who is more preferred than $w^*$ to an agent that is less preferred than $w^*$, where the preference relation is according to profile~$Q$.
Similarly, rotation~$\rho(q)$ is the first rotation that moves the partner of $w^*$ from an agent who is less preferred than $u^*$ to $u^*$ or to an agent that is more preferred than $u^*$, where the preference relation is according to profile~$Q$.
However, in \cref{def:two-specific-rots} we deliberately do \emph{not} refer to profile $Q$ and define the rotations $\pi(q)$ and $\rho(q)$ solely based on $P_0$ in order to make the subsequent formal analysis and the algorithm as clear as possible.

Roughly speaking, eliminating rotation~$\pi(q)$ could make a stable matching of the original profile not stable anymore in the new profile $Q$---indeed, this is the ``first'' rotation, elimination of which causes $u^*$ to prefer $w^*$ over its matched partner.
In order to make sure that $\{u^*, w^*\}$ is not blocking the constructed matching in $Q$, we need to enforce that, whenever the matching includes $\pi(q)$,
agent~$w^*$ must obtain a partner who she prefers over $u^*$.
This is achieved by also including~$\rho(q)$ to the matching.
In other words, when selecting rotations which should form a robust matching, adding $\rho(q)$ fixes some potential issues that arise as a result of adding $\pi(q)$ to the matching.
This intuition is formalized in the subsequent lemmas and theorems. While the main idea is intuitive, the formal analysis is complex since we need to take care of a few technical nuances.   

Note that, by our definition, neither $\pi(q)$ nor $\rho(q)$ needs to exist.
However, if they exist, then they are unique. We will prove this statement in \Cref{lem:the-two-unique}, below. %
The following example that illustrates the definitions of $\pi(q)$~and~$\rho(q)$.

\begin{example}\label[example]{ex:pi-rho}
  \appendixexample{ex:pi-rho}{%
    Consider the profile given in \cref{ex:intro_example} again.
    Let $q=\rotatione$ be the stable quadruple discussed in \cref{ex:swap-sets}.
    Then, $\pi(q)=\pi_1$ and $\rho(q)=\pi_3$.
    The full version \cite{CheSkoSor2019} contains a derivation of this fact.}{
  }{ Recall that in \cref{ex:rotation-poset} we have derived the rotation
    poset of the profile given in \cref{ex:intro_example}, and
    $q=\rotatione$ is the stable quadruple discussed in
    \cref{ex:swap-sets}.
    
    Since $u_3$ prefers $w_1$ to $w_2$, to define $\pi(q)$, we are searching for a rotation, which includes $(u_3,x)$ for some agent~$x\in W$ such that after the elimination of this rotation,~$u_3$ receives agent~$w_1$ as a partner.
  Rotation~$\pi_1$ is the only rotation that fulfills this condition.
  Thus, $\pi(q)\coloneqq \pi_1$.
  Let $P'$ be the profile resulting from performing the two swaps given in $\shifts(P,q)$.
  One can verify that in $P'$ agent~$u_3$ prefers $w_2$ to $w_1$.
  Thus, in the same profile, either \begin{inparaenum}[(i)] \item $u_3$ prefers $w_2$ to the partner assigned by a stable matching whose corresponding closed subset of rotations includes $\pi(q)$, or \item the partner of $u_3$ is $w^* = w_2$. \end{inparaenum}
  Indeed, each rotation eliminated, for which rotation~$\pi(q)$ is a predecessor,
  either makes $w_2$ still the partner of $u_3$
  or changes the partner of $u_3$ to one which is less preferred than $w_1$ in the initial profile, including agent~$w^*$ which is more preferred than $w_1$ in $P'$.
  This means that such kind of stable matching may be blocked by $\{u_3, w_2\}$ in $P'$.

  As for $\rho(q)$, since~$w_2$ prefers $u_3$ to $u_4$, we are searching for a rotation
  which includes $(y,w_2)$ such that
   \begin{align}
    \label{eq:to-u*}  &  \text{either } y = u_3, \\
     & \text{or } w_2 \text{ prefers } u_3 \text{ over } y \text{ and will obtain a partner which } w_2 \text{ prefers over } y \text{.}\label{eq:away-u*}
   \end{align}
   Indeed, since $\pi_3$ satisfies \eqref{eq:away-u*}, one can verify that $\rho(q)=\pi_3$.
   As a final remark, note that, since $\pi_3$ satisfies \eqref{eq:away-u*}, by \citeauthor{GusfieldIrving1989}~\cite[Lemma~2.5.1]{GusfieldIrving1989} it follows that no stable matching matches $u_3$ to $w_2$. %
   Observe that, after the elimination of this rotation,~$w_2$ obtains $u_2$, which is her most preferred agent, i.e., an agent that is strictly better than $u_3$.
   Again, one can verify that in $P'$ agent~$w_2$ still prefers~$u_2$ to~$u_3$.
  Thus,~$w_2$ prefers its partner, assigned by a stable matching whose corresponding closed subset includes $\rho(q)$, to~$u_3$.
  However, this means that such kind of stable matching cannot be blocked by $\{u_3, w_2\}$.

  For a comparison, let us consider another stable quadruple~$q'=(u_1, w_2, u_4, w_1)$.
  Since $u_1$ prefers $w_2$ to $w_1$ we are searching for a rotation
  which includes $(u_1, x)$ such that
   \begin{align}
    \label{eq:from-w*}  &  \text{either } x = w_2, \\
     & \text{or } u_1 \text{ prefers } x \text{ over } w_2 \text{ and will obtain a partner which } u_1 \text{ prefers over } x \text{.}\label{eq:away-w*}
   \end{align}
   Since $\pi_1$ includes $(u_1, w_2)$, satisfying \eqref{eq:from-w*}, we have that $\pi(q') = \pi_1$.

   As for $\rho(q')$, since $w_2$ prefers $u_4$ to $u_1$, we need to find a rotation which includes $(u_4, w_2)$.
   Since $\pi_3$ includes $(u_4, w_2)$, we have that $\rho(q')=\pi_3$.}
  \hfill$\diamond$
\end{example}

Rotations $\pi(q)$ and $\rho(q)$ are critical concepts that will be used by our algorithm for finding robust matchings. The next two lemmas provide tools which allow us to use these concepts conveniently. We start by showing that $\pi(q)$ and $\rho(q)$ are unique.

\newcommand{\rotationsunique}{%
  Let $q=(u^*,w^*,u,w)$ be a stable quadruple of a preference profile~$P$.
  The following holds.
\ifshort  \begin{inparaenum} \else  \begin{compactenum} \fi
    \item If $\pi(q)$ exists, then it is unique.
    \item If $\rho(q)$ exists, then it is unique.
\ifshort  \end{inparaenum} \else  \end{compactenum} \fi
}
\begin{lemma}%
  \label[lemma]{lem:the-two-unique}
  \rotationsunique
\end{lemma}

\appendixproofwithstatement{lem:the-two-unique}{\rotationsunique}{
\begin{proof}
  For the first statement, assume that rotation~$\pi(q)$ exists with $\pi(q)=((u'_0,w'_0),\ldots, (u'_{r-1},w'_{r-1}))$ and $u^*=u'_i$.
  We distinguish between two cases.
  
  \myparagraph{Case~(i):} $w^* \succ_{u^*} w$.
  Thus,~
  $w^*=w'_{i}$ or $w'_i \succ^{P}_{u^*} w^* \succ^{P}_{u^*} w'_{i+1}
$ by definition of $\pi(q)$.
   In either case, \cref{prop:stable-pairs+rotations}\eqref{prop:at-most-1-from-w} guarantees that $\pi(q)$ is unique.

  \myparagraph{Case~(ii):} $w \succ_{u^*} w^*$.
  By the definition of $\pi(q)$, we have that $w=w'_{i+1}$.
  By \cref{prop:stable-pairs+rotations}~\eqref{prop:at-most-1-up}, rotation~$\pi(q)$ is unique.

  Now, we turn to the second statement.
  Assume that rotation~$\rho(q)$ exists with $\rho(q)=((u'_0,w'_0),\ldots,$ $(u'_{r-1},w'_{r-1}))$ and $w^*=w'_i$.
  Again, we distinguish between two cases:
  
  \myparagraph{Case~(i):} $u^* \succ_{w^*} u$.
  This implies that~$u^*=u'_{i-1}$ or $u'_{i-1} \succ^{P}_{w^*} u^* \succ^{P}_{w^*} u'_{i}$ by definition of $\rho(q)$.
  If $u^*=u'_{i-1}$, then the uniqueness is guaranteed by \cref{prop:stable-pairs+rotations}\eqref{prop:at-most-1-up}.
  If $u'_{i-1} \succ^{P}_{w^*} u^* \succ^{P}_{w^*} u'_{i}$, then the uniqueness follows from \cref{prop:stable-pairs+rotations}\eqref{prop:at-most-1-from-m}.

  \myparagraph{Case~(ii):} $u \succ_{w^*} u^*$.
  By the definition of $\rho(q)$, we have that $u=u'_{i}$.
  By \cref{prop:stable-pairs+rotations}~\eqref{prop:at-most-1-from-m}, rotation~$\rho(q)$ is unique.
\end{proof}}

The following result is a
centerpiece of the algorithm, specifying exactly which constraints
need to be fulfilled by a closed subset of rotations which corresponds to a robust matching.

\newcounter{myowncounter}

\newcommand{\characterizationtworotations}{%
  Let $P_0$ be a profile
  and $q=(u^*,w^*,u,w)$ be a stable quadruple of $P_0$.
  Let $Q=P[\shifts(P_0,q)]$ denote the profile after we perform the swaps in~$\shifts(P_0,q)$ on $P_0$.
  The following holds.

  \begin{compactenum}[(i)]
    \item\label{lem:pi-nex}
    Assume that $\pi(q)$ does not exist.
    If $w^*\!\succ^{P_0}_{u^*}\! w$, then each stable matching~$N\!\in\! \sm(P_0)$ has $w^* \succ^{Q}_{u^*} N(u^*)$.
    If $w\!\succ^{P_0}_{u^*}\! w^*$, then each~$N\!\in\! \sm(P_0)$ has either $N(u^*)\!=\!w^*$ or $w^* \succ^{Q}_{u^*} N(u^*)$.
    \item\label{lem:rho-nex}
    Assume that $\rho(q)$ does not exist. 
    If $u^*\succ^{P_0}_{w^*} u$, then each stable matching~$N\in \sm(P_0)$ has $u^* \succ^{Q}_{w^*} N(w^*)$.
    If $u\succ^{P_0}_{w^*} u^*$, then
    each~$N\in \sm(P_0)$ has either $N(u^*)=w^*$ or $u^* \succ^{Q}_{w^*} N(w^*)$.
      \item\label{lem:rho-nex+pi-nex}
    If neither~$\pi(q)$ nor~$\rho(q)$ exist,
    then $\sm(P_0)\cap \sm(Q)=\emptyset$.
    
    \setcounter{myowncounter}{\value{enumi}}
  \end{compactenum}
  
   \noindent Let $S$ be a closed subset of rotations for $P_0$ and let $M$ be the corresponding stable matching.
    \begin{compactenum}[(i)]
      \setcounter{enumi}{\value{myowncounter}}

    \item \label{lem:pi-nex+rho-ex-nin->nstable}
    If $\pi(q)$ does not exist 
    and $\rho(q)$ exists but $\rho(q)\notin S$,
    then $M \notin \sm(Q)$.
   
    \item\label{lem:rho-exists-in->stable} If $\rho(q)$ exists and $\rho(q)\in S$, then $M \in \sm(Q)$.
    
    \item \label{lem:pi-ex-nin->stable} If $\pi(q)$ exists and $\pi(q)\notin S$, then $M \in \sm(Q)$.

    \item\label{lem:pi-ex-in+rho-nex-or-nin->nstable} 
    If $\pi(q)$ exists and $\pi(q)\in S$
    and either $\rho(q)$ does not exist or $\rho(q)$ exists but $\rho(q)\notin S$,
    then $M \notin \sm(Q)$.
  \end{compactenum}
}

\begin{lemma}%
  \label[lemma]{lem:characterization-two-rotations-stable}
  \characterizationtworotations
\end{lemma}

\appendixproofwithstatement{lem:characterization-two-rotations-stable}{\characterizationtworotations}{
  \begin{proof}

  \proofparagraph{Statement~\eqref{lem:pi-nex}.} Assume that $\pi(q)$ does not exist. 
  Since $q$ is a stable quadruple, by definition, $\{\{u^*, w\}, \{u, w^*\}\}$ is a stable set. First, let us consider the case when $w^* \succ^{P_0}_{u^*} w$.
  Since $\pi(q)$ does not exist, by definition, $(u^*,w^*)$ is not in any rotation.
  Neither is $\{u^*,w^*\}$ in the $W$-optimal stable matching of $P_0$ because of the following.
  Since $w^* \succ_{u^*}^{P_0} w$ hold and $\{u,w\}$ is in some stable matching, say $M$, of $P_0$,
  it follows that $M(w^*) \succ^{P_0}_{w^*} u^*$, implying that $\{u^*,w^*\}$ is not in the $W$-optimal stable matching. 
  Hence, by \cref{prop:stable-pairs+rotations}\eqref{prop:stable-pair-char}, $\{u^*,w^*\}$ is not in any stable matching of $P_0$.  %
  Thus, we only need to show that \emph{no} stable matching~$N\in \sm(P_0)$ has $N(u^*) \succ^{Q}_{u^*} w^*$.
  Towards a contradiction, suppose that there exists such a stable matching~$N$ with $N(u^*) \succ_{u^*}^{Q} w^*$.
  By \cref{obs:sh-P-props}\eqref{obs:shift-u} and since $w^*\succ_{u^*}^{P_0}w$, it follows that $\succ_{u^*}^{P_0} = \succ_{u^*}^{Q}$, and so $N(u^*) \succ_{u^*}^{P_0} w^*$, implying that $N(u^*)\succ^{P_0}_{u^*} w^* \succ^{P_0}_{u^*}w$ because $w^* \succ^{P_0}_{u^*} w$.
  By the definition of stable quadruples, $\{u^*,w\}$ is in some stable matching.
  Thus, there are two stable matchings, where $u^*$ obtains a partner (namely, $N(u^*)$) who is more preferred than~$w^*$,
  and a partner  (namely, $w$) who is less preferred than~$w^*$.
  By \cref{prop:stable-pairs+rotations}\eqref{prop:stable-pair-char}, this is a contradiction to the assumption that $\pi(q)$ does not exist.

  Now, consider the case when $w\succ^{P_0}_{u^*} w^*$.
  Since $\{\{u^*, w\}\}$ is a stable set and $\pi(q)$ does not exist, we infer that $\{u^*, w\}$ is in the $U$-optimal stable matching, i.e., that
  every stable matching~$N\in \sm(P_0)$ has either $N(u^*) = w$ or $w\succeq^{P_0}_{u^*}N(u^*)$.
  By \cref{obs:sh-P-props}\eqref{obs:shift-u}(c), it follows that
  every stable matching~$N\in \sm(P_0)$ has either $N(u^*)=w^*$ or $w^*\succ^{Q}_{u^*}N(u^*)$.
   
  \proofparagraph{Statement~\eqref{lem:rho-nex}.} Assume that $\rho(q)$ does not exist. Again, we consider two cases, starting with 
  $u^* \succ^{P_0}_{w^*} u$. Since $\rho(q)$ does not exist,
  we can infer that $\{u^*,w^*\}$ does not belong to any stable matching of $P_0$. Indeed, if a stable matching containing $\{u^*,w^*\}$ existed, then there would be a rotation that changes the partner of $w^*$ from one which is less preferred than $u^*$ to $u^*$ (here, again we use the fact that $q$ is a quadruple and so $w^*$ is matched to $u$ in some stable matching), with respect to profile~$P_0$.

  Thus, we only need to show that \emph{no} stable matching~$N\in \sm(P_0)$ has $N(w^*) \succ^{Q}_{w^*} u^*$.
  Towards a contradiction, suppose that there exists such a stable matching~$N$ with $N(w^*) \succ_{w^*}^{Q} u^*$.
  By \cref{obs:sh-P-props}\eqref{obs:shift-w} and since $u^* \succ_{w^*}^{P_0} w$, it follows that $\succ_{w^*}^{P_0} = \succ_{w^*}^{Q}$ and so that $N(w^*) \succ_{w^*}^{P_0} u^*$, implying that $N(w^*)\succ^{P_0}_{w^*} u^* \succ^{P_0}_{w^*} u$ because $u^* \succ^{P_0}_{w^*} u$.
  By the definition of stable quadruples, $\{u,w^*\}$ belongs to some stable matching. Summarizing, there exist two stable matchings where $w^*$ is matched to a partner which is less preferred than $u^*$ and a partner which is more preferred than~$u^*$, respectively. 
  However, this is a contradiction to the assumption that $\pi(q)$ does not exist.

  Now, let us move to the case when $u\succ^{P_0}_{w^*} u^*$. Recall that, since $q$ is a stable quadruple, we know that $\{u, w^*\}$ belongs to some stable matching. Since $\rho(q)$ does not exist 
by \cref{prop:stable-pairs+rotations}\eqref{prop:stable-pair-char} we infer that $\{u, w^*\}$ is in the $W$-optimal stable matching. In other words, for every stable matching~$N\in \sm(P_0)$ we have 
either $u=N(w^*)$ or $u\succ^{P_0}_{w^*}N(w^*)$.
  By \cref{obs:sh-P-props}\eqref{obs:shift-w}(c), it follows that
  every stable matching~$N\in \sm(P_0)$ has either $N(w^*)=u^*$ or $u^*\succ^{Q}_{w^*}N(w^*)$.

  \proofparagraph{Statement~\eqref{lem:rho-nex+pi-nex}.} Assume that neither $\pi(q)$ nor $\rho(q)$ exists.
  We distinguish between two cases.

  \myparagraph{Case (1): $w \succ^{P_0}_{u^*}w^*$.} Since $\pi(q)$ does not exist,
  by statement~\eqref{lem:pi-nex}, it follows that every stable matching~$N\in \sm(P_0)$ has either $N(u^*)=w^*$ or $w^* \succ^{Q}_{u^*} N(u^*)$.
  Consider an arbitrary stable matching~$N \in \sm(P_0)$, and first let us analyze the case when $N(u^*)=w^*$. Since $q$ is a stable quadruple, there exists a matching, call it $M'$, such that $M'(u^*) = w$ and $M'(u) = w^*$.
  By our assumption, $u^*$ prefers $M'$ to $N$; thus, by \cref{prop:diff-partners-diff-prefs} we get that $w^*$ must prefer $N$ to $M'$, i.e., it must hold that~$u^* \succ^{P_0}_{w^*} u$.
  By statement~\eqref{lem:rho-nex}, that ``$u^* \succ^{P_0}_{w^*} u$'' and the assumption that ``rotation~$\rho(q)$ does not exist''
  imply that $u^* \succ^{Q}_{w^*} N(w^*)$.
  This contradicts our assumption that $N(u^*)=w^*$.
  Now, let us move to the second alternative, when $w^*\succ^{Q}_{u^*}N(u^*)$. We know that $\rho(q)$ does not exist. Thus, by statement~\eqref{lem:rho-nex}, we get that $u^* \succ^{Q}_{w^*} N(w^*)$ because of the following.
  \begin{compactitem}
    \item Either $u^*\succ^{P_0}_{w^*} u$, whence by statement~\eqref{lem:rho-nex} we have $u^* \succ^{Q}_{w^*} N(w^*)$,
    \item or $u\succeq^{P_0}_{w^*} u^*$ and by statement~\eqref{lem:rho-nex} we have that $u^* \succ^{Q}_{w^*} N(w^*)$ or that $N(u^*) = w^*$ (the latter case has just been handled; either way, we have that $u^* \succ^{Q}_{w^*} N(w^*)$.
  \end{compactitem}
  Yet, this implies that $\{u^*,w^*\}$ is a blocking pair of $N$.  Thus, $N\notin \sm(Q)$.

  \noindent \textbf{Case (2): $w^* \succ^{P_0}_{u^*} w$.} Since
  $\pi(q)$ does not exist, from statement~\eqref{lem:pi-nex} it
  follows that every stable matching~$N\in \sm(P_0)$ satisfies
  $w^*\succ^{Q}_{u^*}N(u^*)$. In particular, it follows that
  $\{\{u^*,w^*\}\}$ is not a stable set in $Q$. By
  statement~\eqref{lem:rho-nex}, from this and from the assumption
  that $\rho(q)$ does not exist, we infer that
  $u^* \succ^{Q}_{w^*} N(w^*)$ (either $u^*\succ^{P_0}_{w^*} u$ and we
  get $u^* \succ^{Q}_{w^*} N(w^*)$ directly from
  statement~\eqref{lem:rho-nex}, or $u\succeq^{P_0}_{w^*} u^*$ and by
  statement~\eqref{lem:rho-nex} we get that either
  $u^* \succ^{Q}_{w^*} N(w^*)$ or $N(u^*) = w^*$---and we have already
  shown that in the latter case $N\notin \sm(Q)$). Thus, the
  pair~$\{u^*,w^*\}$ is blocking $N$ in $Q$.

  Summarizing, we have shown that \emph{no} stable matching of $P_0$ is stable for $Q$.

  \proofparagraph{Statement~\eqref{lem:pi-nex+rho-ex-nin->nstable}.} Assume that $\pi(q)$ does not exist and $\rho(q)$ exists  but $\rho(q) \notin S$.
  Since $\pi(q)$ does not exist, by statement~\eqref{lem:pi-nex}, for every stable matching~$N\in \sm(P_0)$ it holds that $N(u^*)=w^*$ or $w^*\succ^{Q}_{u^*}N(u^*)$.
  This includes $M$, meaning that $M(u^*)=w^*$ or $w^* \succ^{Q}_{u^*}M(u^*)$.

  We consider these two cases separately.

  \noindent \textbf{Case (1): $M(u^*) = w^*$.} By statement~\eqref{lem:pi-nex}, it follows that $w \succ^{P_0}_{u^*} w^*$.
  Further, since $q$ is a stable quadruple, there exists a stable matching $N\in \sm(P_0)$, such that $N(u^*) = w$ and $N(u) = w^*$. In $P_0$, since $u^*$ prefers $N$ to $M$, by \cref{prop:diff-partners-diff-prefs}, 
  it must be the case that $w^*$ prefers $M$ to $N$, i.e.,~$u^* \succ^{P_0}_{w^*} u$.
  Thus, the rotation~$\rho(q)$ (which, by our assumption, is guaranteed to exist) operates as follows: it changes the partner of $w^*$ from an agent that is less preferred than $u^*$ to $u^*$ or to an agent that is more preferred than $u^*$ (regarding $P_0$).
  Since $\rho(q)\notin S$, we infer that in matching~$M$ agent $w^*$ obtains a partner that is less preferred than $u^*$, i.e., $u^* \succ^{P_0}_{w^*} M(w^*)$. This leads to a contradiction with $M(u^*)=w^*$.

  \noindent \textbf{Case (2): $w^* \succ^{Q}_{u^*}M(u^*)$.} 
  Towards a contradiction, suppose that $M$ is also stable for $Q$.
  This implies that $M(w^*)\succ^{Q}_{w^*} u^*$.
  By \cref{obs:sh-P-props}\eqref{obs:shift-w}(b), we have that $M(w^*)\succ^{P_0}_{w^*} u$.
  If $u \succ^{P_0}_{w^*} u^*$, then the rotation $\rho(q)$ changes the partner of $w^*$ from $u$ to some agent which is more preferred than $u$ (regarding the preferences in $P_0$).
  Since, according to $M$, agent~$w^*$ already has a partner that is more preferred than $u$, we infer that $\rho(q)$ is the predecessor of some rotation in $S$, meaning that itself $\rho(q) \in S$ by the closedness of~$S$---a contradiction. 
  If  $u^* \succ^{P_0}_{w^*} u$, then $\succ^{P_0}_{w^*} = \succ^{Q}_{w^*}$ by \cref{obs:sh-P-props}\eqref{obs:shift-w}, and so we get that $M(w^*)\succ^{P_0}_{w^*} u^*\succ^{P_0}_{w^*} u$ because~$M(w^*)\succ^{Q}_{w^*} u^*$.
  In this case, rotation $\rho(q)$ changes the partner of $w^*$ from an agent that is less preferred than $u^*$ to $u^*$ or an agent who is more preferred than $u^*$. However, since in $M$ agent $w^*$ already has a partner who is preferred over $u^*$, we again infer that $\rho(q)\in S$---a contradiction.

  Summarizing, we conclude that $M\notin \sm(Q)$.

  \proofparagraph{Statement~\eqref{lem:rho-exists-in->stable}.} Let us assume that $\rho(q)$ exists and $\rho(q)\in S$.
  
  By \cref{lem:profile-only-possible-blocking-pair}, except $\{u^*,w^*\}$, no other unmatched pair with respect to $M$ could be blocking $Q$.
  In the following, we claim that $\{u^*,w^*\}$ is not blocking $Q$, implying that $M$ is stable for $Q$.
  We distinguish between two cases.
  
  If $u^{*}\succ_{w^*}^{P_0} u$, then by the definition of $\rho(q)$ and since $\rho(q)\in S$,
  it follows that $M(w^*)=u^*$ or $M(w^*) \succ_{w^*}^{P_0} u^*$.
  Moreover, by \cref{obs:sh-P-props}\eqref{obs:shift-w} we have that $\succ_{w^*}^{Q}=\succ_{w^*}^{Q}$,
  implying $M(w^*)=u^*$ or $M(w^*) \succ_{w^*}^{Q} u^*$. 
  Thus, $\{u^*,w^*\}$ cannot be blocking $M$ in $Q$.

  If $u\succ_{w^*}^{P_0} u^*$, then by the definition of $\rho(q)$ and since $\rho(q)\in S$,
  it follows that $M(w^*) \succ_{w^*}^{P_0} u$.
  Moreover, by \cref{obs:sh-P-props}\eqref{obs:shift-w}(b) 
  we obtain that $M(w^*) \succ_{w^*}^{Q} u^*$. 
  Thus, $\{u^*,w^*\}$ cannot be blocking $M$ in $Q$.

  \proofparagraph{Statement~\eqref{lem:pi-ex-nin->stable}.} Let us assume that $\pi(q)$ exists and $\pi(q)\notin S$.
  By \cref{lem:profile-only-possible-blocking-pair}, except $\{u^*,w^*\}$, no other unmatched pair with respect to $M$ could be blocking $Q$.
  In the following, we claim that $\{u^*,w^*\}$ is not blocking $Q$, which implies that $M$ is stable for $Q$.
  We distinguish between two cases.
  
  If $w^{*}\succ_{u^*}^{P_0} w$, then by the definition of $\pi(q)$ and since $\pi(q)\notin S$,
  it follows that $M(u^*)=w^*$ or $M(u^*) \succ_{u^*}^{P_0} w^*$.
  Moreover, by \cref{obs:sh-P-props}\eqref{obs:shift-u} we have that $\succ_{u^*}^{Q}=\succ_{u^*}^{Q}$,
  implying $M(u^*)=w^*$ or $M(u^*) \succ_{u^*}^{Q} w^*$. 
  Thus, $\{u^*,w^*\}$ cannot be blocking $M$ in $Q$.
  
  If $w\succ_{u^*}^{P_0} w^*$, then by the definition of $\pi(q)$ and since $\pi(q)\notin S$,
  it follows that $M(u^*) \succ_{u^*}^{P_0} w$.
  Moreover, by \cref{obs:sh-P-props}\eqref{obs:shift-u}(b) 
  we obtain that $M(u^*) \succ_{u^*}^{Q} w^*$. 
  Thus, $\{u^*,w^*\}$ cannot be blocking $M$ in $Q$.

  \proofparagraph{Statement~\eqref{lem:pi-ex-in+rho-nex-or-nin->nstable}.}
  Assume that $\pi(q)$ exists and $\pi(q) \in S$
  and either $\rho(q)$ does not exist or it exists but $\rho(q)\notin S$.

  Suppose, for the sake of contradiction, that $M$ is stable for $Q$.
  We distinguish between three cases, in each case obtaining a contradiction.
  \myparagraph{Case~1: $\boldsymbol{w^* \succ^{P_0}_{u^*} w}$.}
  By the definition of $\pi(q)$ and since $\pi(q)\in S$, referencing \cref{prop:rotations+stable-matchings}\eqref{rot:closedsubet-stable},
  it follows that $w^*\succ^{P_0}_{u^*} M(u^*)$.
  Thus, $w^*\succ^{Q}_{u^*} M(u^*)$ because $\succ^{P_0}_{u^*}=\succ^{Q}_{u^*}$ (by \cref{obs:sh-P-props}\eqref{obs:shift-u}).
  In particular, this implies that $\{u^*,w^*\}$ is an unmatched pair  in $M$.
  By assumption that $M$ is stable for $Q$,
  we must have that
  \begin{align}
    M(w^*) \succ_{w^*}^{Q} u^*.\label{eq:Mwstar<ustar}
  \end{align}
  If $u^*\succ_{w^*}^{P_0} u$, then by \cref{obs:sh-P-props}\eqref{obs:shift-w}, we have $M(w^*) \succ_{w^*}^{P_0} u^* \succ^{P_0}_{w^*} u$.
  By the fact that $\{u,w^*\}$ is in some stable matching (recall that $q$ is a stable quadruple), 
  there are two stable matchings, where $w^*$ obtains a partner (namely, $u$) who is less preferred than~$u^*$,
  and a partner  (namely, $M(w^*)$) who is more preferred than~$u^*$.
  This implies that $\rho(q)$ exists and that $\rho(q)\in S$---a contradiction.

  If $u\succ_{w^*}^{P_0} u^*$, then by \cref{obs:sh-P-props}\eqref{obs:shift-w} and by \eqref{eq:Mwstar<ustar},
  we have that $M(w^*) \succ_{w^*}^{P_0} u \succ_{w^*}^{P_0} \succ u^*$.
  Again, by the fact that $\{u,w^*\}$ is in some stable matching (recall that $q$ is a stable quadruple), 
  there are two stable matchings, where $w^*$ obtains partner~$u$, 
  and a partner  (namely, $M(w^*)$) who is more preferred than~$u$.
  This implies that $\rho(q)$ exists such that $\rho(q) \in S$---a contradiction.
  
  \myparagraph{Case~2: \boldmath${w \succ^{P_0}_{u^*} w^*}$ and $M(u^*)\neq w^*$.}
  If we can show that $w^* \succ^{Q}_{u^*} M(u^*)$,
  then we can use the same reasoning as we did for the first case to show the same contradiction.
  Thus it suffices to prove that $w^* \succ^{Q}_{u^*} M(u^*)$.
    
  By the definition of $\pi(q)$ and since $\pi(q)\in S$, referencing \cref{prop:rotations+stable-matchings}\eqref{rot:closedsubet-stable},
  it follows that $M(u^*)=w$ or $w\succ^{P_0}_{u^*} M(u^*)$,
  and thus  $w^*\succ^{Q}_{u^*} M(u^*)$ because  $M(u^*)\neq w^*$ (by assumption) and $w^* \succ^{Q}_{u^*} w$.
  This finishes the proof for the second case.

  \myparagraph{Case~3: \boldmath$w\succ^{P_0}_{u^*} w^*$ and $M(u^*)=w^*$.}
  Since $q$ is a stable quadruple,~$\{u^*,w\}$ and $\{u,w^*\}$ exist in some stable matching of $P_0$, say $N$.
  Thus, $P_0$ admits two different stable matchings~$M$ and $N$,
  where $M(u^*)=w^*$,  $N(u^*)=w$, and $N(w^*)=u$.
  By the precondition that $w\succ^{P_0}_{u^*} w^*$ and by \cref{prop:diff-partners-diff-prefs},
  we must have that
  $u^*\succ_{w^*} u$ (i.e.\ $w^*$ prefers $M$ to $N$).
  In particular, this means that there must be a rotation which changes the partner of $w^*$ from one that is less preferred than $u^*$ to agent~$u^*$.
  Thus, $\rho(q)$ exists and must be in $S$---a contradiction.
\end{proof}
}

\subsection{Polynomial-Time Algorithms for Robust Matchings}
We now first present an $O(n^4)$-time algorithm for finding a robust
matching if it exists. Then we use a Linear Programming (LP)
formulation to show that perfect robust matchings and robust matchings
with minimum egalitarian cost can be found in polynomial time if they
exist. Both approaches crucially rely on
\begin{inparaenum}[(a)]
\item the one-to-one correspondence between the
stable matchings and the closed subsets of the rotation
poset~\cite[Chapter~3.7]{GusfieldIrving1989},
\item the implications between the presence of the two rotations
  $\pi(q)$ and $\rho(q)$ of stable quadruples~$q$ derived in
  \cref{lem:characterization-two-rotations-stable}, and
\item the fact that all stable quadruples can be computed in $O(n^4)$ time.
\end{inparaenum}
The proof for (c)~is roughly by iterating over all possible rotations
and building a lookup table that stores for all
pairs~$(x, y) \in U \times W$ of agents a constant number of rotations
that make the partner of~$x$ less preferred to~$y$ or more preferred
to~$y$. For given stable quadruple~$q$, rotations $\pi(q)$ and
$\rho(q)$ can then be looked up in the table. We state this
observation for reference below.
\newcommand{\stablequadrupleruntime}{%
  Determining all stable quadruples~$q$ and their respective
  rotations $\pi(q)$ and $\rho(q)$ as defined in
  \cref{def:two-specific-rots} can be done in $O(n^4)$~time. }
\begin{proposition}%
  \label[proposition]{claim:stable-quadruples-rotations-runtime}
  \stablequadrupleruntime
\end{proposition}
  \appendixproofwithstatement{claim:stable-quadruples-rotations-runtime}{\stablequadrupleruntime}
  {
  \begin{proof}  \renewcommand{\qedsymbol}{(of
       \cref{claim:stable-quadruples-rotations-runtime})~$\diamond$}
  By \cref{prop:stable-pairs} and by \cref{prop:rotations+stable-matchings}\eqref{rot:runtime},
  all $O(n^4)$ stable quadruples can be found in $O(n^4)$~time.

  For each stable quadruple~$q$, we show how to find $\pi(q)$ and $\rho(q)$, in $O(1)$ time for a given stable quadruple~$q$.
  We build in $O(n^4)$ time a size-$O(n^2)$ look-up table~$T$ to store for each ordered pair~$(x,y)\in U \times W$ the following up to six rotations:
  \begin{compactenum}
    \item Let $\sigma_1(x,y)$ denote the rotation which changes the partner of $x$ from someone who is more preferred than $y$ to $y$.
    
    Formally, $\sigma_1(x,y) \coloneqq ((u'_0,w'_0), \ldots, (u_{r-1}',w_{r-1}'))$
    such that $x=u'_i$ and $y=w'_{i+1}$ for some $i\in \{0,\ldots,r-1\}$.
    Note that the uniqueness of this rotation is guaranteed by \cref{prop:stable-pairs+rotations}\eqref{prop:at-most-1-up}.
    \item Let $\sigma_2(x,y)$ denote the rotation which changes the partner of $x$ from someone who is more preferred than $y$ to someone who is less preferred than $y$.

    Formally, $\sigma_2(x,y) \coloneqq ((u'_0,w'_0), \ldots, (u_{r-1}',w_{r-1}'))$
    such that $x=u'_i$ and $w'_{i} \succ_{x} y \succ_x w'_{i+1}$ for some $i\in \{0,\ldots,r-1\}$.
    Note that the uniqueness of this rotation is guaranteed by \cref{prop:stable-pairs+rotations}\eqref{prop:at-most-1-from-w}.

    \item Let $\sigma_3(x,y)$ denote the rotation which changes the partner of $x$ from $y$
    to someone who is less preferred than $y$.

    Formally, $\sigma_3(x,y) \coloneqq ((u'_0,w'_0), \ldots, (u_{r-1}',w_{r-1}'))$
    such that $x=u'_i$ and $y=w'_{i}$ for some $i\in \{0,\ldots,r-1\}$.
    Note that the uniqueness of this rotation is guaranteed by \cref{prop:stable-pairs+rotations}\eqref{prop:at-most-1-from-w}.
    Moreover, the existence of $\sigma_2(x,y)$ precludes the existence of $\sigma_1(x,y)$ and $\sigma_3(x,y)$
    because $\sigma_2(x,y)$ implies that $\{x,y\}$ is not in any stable matching while $\sigma_1(x,y)$ or $\sigma_3(x,y)$ implies that $\{x,y\}$ is some stable matching.
    
    \item  Let $\tau_1(y,x)$ denote the rotation which changes the partner of $y$ from someone who is less preferred than $x$ to $x$.

    Formally, $\tau_1(y,x) \coloneqq ((u'_0,w'_0), \ldots, (u_{r-1}',w_{r-1}'))$
    such that $y=w'_i$ and $x=u'_{i-1}$ for some $i\in \{0,\ldots,r-1\}$.
    Note that the uniqueness of this rotation is guaranteed by \cref{prop:stable-pairs+rotations}\eqref{prop:at-most-1-up}.

    \item Let $\tau_2(y,x)$ denote the rotation which changes the partner of $y$ from someone who is less preferred than $x$ to someone who is more preferred than $x$.

    Formally, $\tau_2(y,x) \coloneqq ((u'_0,w'_0), \ldots, (u_{r-1}',w_{r-1}'))$
    such that $y=w'_i$ and $u'_{i-1} \succ_{y} x \succ_y u'_{i}$ for some $i\in \{0,\ldots,r-1\}$.
    Note that the uniqueness of this rotation is guaranteed by \cref{prop:stable-pairs+rotations}\eqref{prop:at-most-1-from-m}.
    
    \item Let $\tau_3(x,y)$ denote the rotation which changes the partner of $y$ from $x$
    to someone who is more preferred than $x$.

    Formally, $\tau_1(y,x) \coloneqq ((u'_0,w'_0), \ldots, (u_{r-1}',w_{r-1}'))$
    such that $y=w'_i$ and $x=u'_{i}$ for some $i\in \{0,\ldots,r-1\}$.
    Note that the uniqueness of this rotation is guaranteed by \cref{prop:stable-pairs+rotations}\eqref{prop:at-most-1-from-m}.

    Moreover, the existence of $\tau_2(y,x)$ precludes the existence of $\tau_1(y,x)$ and $\tau_3(y,x)$
    because $\tau_2(y,x)$ implies that $\{x,y\}$ is not in any stable matching while $\tau_1(y,x)$ or $\tau_3(y,x)$ implies that $\{x,y\}$ is some stable matching.
  \end{compactenum}

  \cref{figure:six-rotations} illustrates the six rotations we have just defined.
  \begin{figure}
       \tikzstyle{matrixstyle} = [matrix of math nodes,,ampersand replacement=\&,column sep=1pt]
\centering
    \begin{tikzpicture}[>=stealth]
    \node[] (ucase1) {};
    \node[right = 0pt of ucase1] (ui1) {$x:$};
    \matrix[right = -3pt of ui1, matrixstyle] (ui1p) 
    {\ldots \& \ldots \& \succ \& y \& \succ \& \ldots \& \ldots \\};

    \draw[->,rounded corners] (ui1p-1-2) -- ($(ui1p-1-2)+(0,-.8)$) -- node[midway,fill=white] {\footnotesize $\sigma_2(x,y)$}  ($(ui1p-1-6)+(0,-.8)$) -| (ui1p-1-6);
  
    \draw[->,rounded corners] ($(ui1p-1-1)+(0,0.2)$) -- ($(ui1p-1-1)+(0,.8)$) -- node[fill=white,pos=0.5] {\footnotesize $\sigma_1(x,y)$}  ($(ui1p-1-4)+(-.1,.8)$) |- ($(ui1p-1-4)+(-.1,0.2)$);
  
    \draw[->,rounded corners] ($(ui1p-1-4)+(0.1,0.2)$) -- ($(ui1p-1-4)+(0.1,.8)$) -- node[fill=white,pos=0.5] {\footnotesize $\sigma_3(x,y)$}  ($(ui1p-1-7)+(-.1,.8)$) |- ($(ui1p-1-7)+(-.1,0.2)$);

    \node[right = 15pt of ui1p] (wi1) {$y:$};
    \matrix[right = -3pt of wi1, matrixstyle] (wi1p) 
    {\ldots \& \ldots \& \succ \& x \& \succ \& \ldots \& \ldots \\};

    \draw[->,rounded corners] (wi1p-1-6) -- ($(wi1p-1-6)+(0,-.8)$) -- node[midway,fill=white] {$\tau_2(y,x)$}  ($(wi1p-1-2)+(0,-.8)$) -| (wi1p-1-2);

    \draw[->,rounded corners] ($(wi1p-1-4)+(-.1,0.2)$) -- ($(wi1p-1-4)+(-.1,.8)$) -- node[fill=white,pos=0.5] {\footnotesize $\tau_3(y,x)$}  ($(wi1p-1-1)+(0,.8)$) |- ($(wi1p-1-1)+(0,0.2)$);
  
    \draw[->,rounded corners] ($(wi1p-1-7)+(0.1,0.2)$) -- ($(wi1p-1-7)+(0.1,.8)$) -- node[fill=white,pos=0.5] {\footnotesize $\tau_1(y,x)$}  ($(wi1p-1-4)+(.1,.8)$) |- ($(wi1p-1-4)+(.1,0.2)$);
  \end{tikzpicture}\caption{Illustration for the six rotations defined in the proof of \cref{claim:stable-quadruples-rotations-runtime}.}\label{figure:six-rotations}
\end{figure}
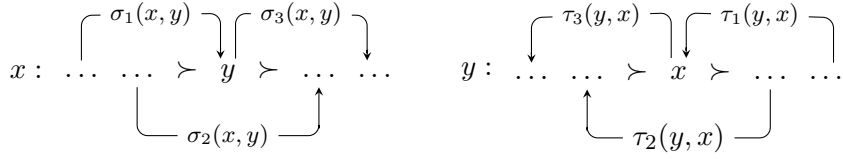

\allowdisplaybreaks[1]
Now, we continue with the determination of $\rho(q)$ and $\pi(q)$.
  Let $q=(u^*,w^*,u,w)$.
  \begin{align*}
    \pi(q) \coloneqq &
    \begin{cases}
      \sigma_2(u^*,w^*) , &\text{ if }w^* \succ_{u^*} w \text{ and }\sigma_2(u^*,w^*) \text{ exists, }\\
      \sigma_3(u^*,w^*) , &\text{ if }w^* \succ_{u^*} w \text{ and } \sigma_3(u^*,w^*) \text{ exists, }\\
      \sigma_2(u^*,w) , &\text{ if }w \succ_{u^*} w^* \text{ and }\sigma_1(u^*,w) \text{ exists, }\\
      \textsf{undefined}, &\text{otherwise.}\\
    \end{cases}\\
     \rho(q) \coloneqq &
    \begin{cases}
      \tau_1(w^*,u^*) , &\text{ if }u^* \succ_{w^*} u \text{ and }\tau_1(w^*,u^*) \text{ exists, }\\
      \tau_2(w^*,u^*) , &\text{ if }u^* \succ_{w^*} u \text{ and } \tau_2(w^*,u^*) \text{ exists, }\\
      \tau_3(w^*,u) , &\text{ if } u\succ_{u^*} u^* \text{ and }\tau_3(w^*,u) \text{ exists, }\\
      \textsf{undefined}, &\text{otherwise.}\\
    \end{cases}
  \end{align*}
  
  One can verify that the above construction corresponds to \cref{def:two-specific-rots}.
  Since there are $O(n^2)$ rotations and $n^2$ ordered pairs,
  the whole table, containing $O(n^2)$ entries, can be determined in $O(n^4)$~time. (Note that, for a given rotation~$\rho$, we can first find the agents who are affected by the rotation in $O(n)$ time, and, for each of the affected agents~$z$, find in $O(n)$~time all the pairs~$(z, w)$ such that $\rho$ needs to be added to the table entry of~$(z, w)$.)
  After computing the table, we can determine in constant time the two rotations~$\pi(q)$ and $\rho(q)$ for each~$q$ from the $O(n^4)$~stable quadruples, by looking up into the table.
  In total, the running time is~$O(n^4)$.
\end{proof}
}

We now prove our main result for the \pRMl{} problem.

\begin{theorem}\label[theorem]{thm:d-robust-poly}
  Given an instance of \pRMl{} with $2n$ agents, in $O(n^4)$~time we can either find a $d$-robust matching or correctly report that no such matching exists.
\end{theorem}
\begin{proof}
  Our approach is described in \cref{alg:robust}. 
  Let $P$ be the
  profile in the given instance.
  \columnratio{.58}
  \begin{paracol}{2}
\begin{algorithm}[h]
  \DontPrintSemicolon

  \KwIn{A preference profile~$P$ with agent sets~$U$ and $W$, and an integer~$d\in \mathds{N}$.}
  
  \KwOut{A $d$-robust matching for $P$ or $\bot$ if none exists.}
  
  \BlankLine
  
  Compute the rotation digraph $G(P)$\; \label{alg:graph-start}

  $G_1(P) \leftarrow G(P)$\;  \label{alg:G1-start} 
  \ForEach{stable quadruple $q$ with $|\shifts(P, q)| \leq d$\label{alg:sq-start}}
  {
    Compute $\pi(q)$ and $\rho(q)$ if they exist using Prop.~\ref{claim:stable-quadruples-rotations-runtime}\;\label{alg:sq-end}
    \lIf{$\not \exists\pi(q)$ and $\not \exists\rho(q)$}{\Return $\bot$}\label{alg:no-pi-no-rho-return-bot}
    \lIf{$\exists \pi(q)$ and $\exists \rho(q)$}{$G_1(P)\! \leftarrow\! G_1(P)\!+\! (\rho(q), \pi(q))$}\label{alg:add-arc} %
  }\label{alg:G1-end}

  $D \leftarrow \{\pi(q) \mid q $ is a stable quadruple with $|\shifts(P, q)| \leq d$ \label{alg:pi-no-rho-D}

  \qquad \qquad \qquad s.t.\ $\exists \pi(q)$ but $\not \exists \rho(q)\}$\;

  $G_2(P) \leftarrow G_1(P) - D - \{v \in V(G_1(P)) \mid \exists$\text{ a dipath in }$G_1(P)$\label{alg:G2-pi-no-rho-D}

 \qquad \qquad \qquad \qquad \qquad \text{ from a vertex in }$D$ to $v\}$\;
  
  \label{alg:graph-end} $A \leftarrow \{\rho(q) \mid q$  is a stable quadruple with $|\shifts(P, q)| \leq d$\; \label{alg:A} 

  \qquad  \qquad \qquad s.t.\ $\exists \rho(q)$ but $\not \exists \pi(q)$\}
  
  \lIf{$A \nsubseteq V(G_2(P)) $}{\Return $\bot$}\label{alg:required-vertex-nexists-return-bot}
  $T \leftarrow A \cup \{v \in V(G_2(P)) \mid \exists$ a dipath in $G_2(P)$ from  $v$\label{alg:closure-T}

  \qquad \qquad \qquad \qquad \qquad \quad to some vertex in  $A\}$\;
  \Return the matching corresponding to the closed set~$T$ of rotations\;\label{alg:last}
  \caption{Computing $d$-robust matchings.}
  \label{alg:robust}
\end{algorithm}
\switchcolumn
    \noindent To obtain an $O(n^4)$ algorithm we
  work with the rotation digraph~$G(P)$ (see \cref{def:digraph} and
  \cref{prop:rotations+stable-matchings}).
  Call a vertex subset~$S$ in a directed graph~$G$ \emph{closed}, if there is no arc in~$G$ pointing outwards from~$S$.
  Recall that a stable
  matching for~$P$ corresponds to a closed subset of the rotations in the rotation
  poset, i.e., a closed vertex subset~$S$ of~$G(P)$. 
  Intuitively, \cref{alg:robust}
  first adds arcs to $G(P)$ in lines~\ref{alg:G1-start} to~\ref{alg:G1-end} that model implications between rotations contained in $d$-robust matchings given in \cref{lem:characterization-two-rotations-stable}.
  Then, it removes rotations from~$G(P)$ that cannot occur in $d$-robust matchings according to \cref{lem:characterization-two-rotations-stable} in lines~\ref{alg:pi-no-rho-D} and~\ref{alg:G2-pi-no-rho-D}. Finally, it
  checks in line~\ref{alg:closure-T} whether there is a closed subset of rotations which contains the required
  rotations for $d$-robust matchings according to~\cref{lem:characterization-two-rotations-stable}.
\end{paracol}
  We now prove the correctness and then analyze the running time.
Below, when referring to $G_1(P)$, we
mean the graph~$G_1(P)$ obtained from $G(P)$ after
line~\ref{alg:G1-end} and by $G_2(P)$ we mean the graph obtained after
line~\ref{alg:graph-end}.

  \looseness=-1 We claim that, if \cref{alg:robust} returns something different
  from~$\bot$, then it is a $d$-robust matching. We first
  show that the set~$T$ computed in line~\ref{alg:closure-T} is a closed subset of
  rotations in the rotation poset: Clearly, $T$ is closed in $G_2(P)$.
  Since $G_2(P)$ is obtained from $G_1(P)$ by removing vertices
  together with all of their successors, $T$ is closed in $G_1(P)$ as
  well. Since $G_1(P)$ is obtained from $G(P)$ by adding arcs, $T$ is closed in~$G(P)$, implying the claim. By \cref{prop:rotations+stable-matchings}, there is a stable matching~$M$ for~$P$ associated with~$T$.

  Since $T$ is a closed subset,
  \cref{lem:characterization-two-rotations-stable}~\eqref{lem:pi-nex+rho-ex-nin->nstable}
  to~\eqref{lem:pi-ex-in+rho-nex-or-nin->nstable} apply. We now verify
  that, for each stable quadruple~$q$ with $|\shifts(P, q)| \leq d$,
  we have $M \in \sm(P[\shifts(P, q)])$ (recall that
  $P[\shifts(P, q)]$ is the profile obtained from $P$ by performing the swaps in
  $\shifts(P, q)$). By \cref{lem:d-swaps-profiles} it then follows
  that~$M$ is $d$-robust.

  Let $q$ be a stable quadruple of $P$ such that
  $|\shifts(P, q)| \leq d$. Suppose that $\rho(q)$ exists. If $\pi(q)$
  does not exist, then $\rho(q) \in T$ by line~\ref{alg:A} and line~\ref{alg:closure-T}, and thus
  $M \in \sm(P[\shifts(P, q)])$ by
  \cref{lem:characterization-two-rotations-stable}~\eqref{lem:rho-exists-in->stable}.
  If $\pi(q)$ exists, then, since $T$ is closed and by line~\ref{alg:add-arc}, either
  $\pi(q) \notin T$, giving $M \in \sm(P[\shifts(P, q)])$ by
  \cref{lem:characterization-two-rotations-stable}~\eqref{lem:pi-ex-nin->stable},
  or $\rho(q) \in T$, giving $M \in \sm(P[\shifts(P, q)])$ by
  \cref{lem:characterization-two-rotations-stable}~\eqref{lem:rho-exists-in->stable}. 
  Now suppose that $\rho(q)$ does not exist. Then, $\pi(q)$ exists because otherwise we would have returned $\bot$ in line~\ref{alg:no-pi-no-rho-return-bot}.
  By lines \ref{alg:pi-no-rho-D} and \ref{alg:G2-pi-no-rho-D}, $\pi(q) \notin T$, giving
  $M \in \sm(P[\shifts(P, q)])$ by
  \cref{lem:characterization-two-rotations-stable}~\eqref{lem:pi-ex-nin->stable}.
  Thus, indeed the returned matching is $d$-robust.

  It remains to show that a $d$-robust matching is returned if there
  is a $d$-robust matching~$M$ for~$P$. By the above, it suffices to
  show that $\bot$ is not returned in lines~\ref{alg:no-pi-no-rho-return-bot} and~\ref{alg:required-vertex-nexists-return-bot}.  By
  \cref{lem:d-swaps-profiles}, $M \in \sm(P[\shifts(P, q)])$ for each
  stable quadruple~$q$ with $|\shifts(P, q)| \leq d$. Thus, by
  \cref{lem:characterization-two-rotations-stable}~\eqref{lem:rho-nex+pi-nex}
  at least one of $\rho(q)$ and $\pi(q)$ exists, meaning that
  $\bot$ cannot be returned in line~\ref{alg:no-pi-no-rho-return-bot}. If $\bot$ was
  returned due to line~\ref{alg:required-vertex-nexists-return-bot}, then there is a stable quadruple~$q$ with
  $|\shifts(P, q)| \leq d$ such that $\pi(q)$ does not exist and
  $\rho(q)$ exists and, furthermore,
  $\rho(q) \in V(G_1(P)) \setminus V(G_2(P))$. Let $S$ be the closed
  subset of rotations in~$G(P)$ associated with~$M$. By
  \cref{lem:characterization-two-rotations-stable}~\eqref{lem:pi-nex+rho-ex-nin->nstable},
  $\rho(q) \in S$.
  Since $\rho(q) \notin V(G_2(P))$, by lines~\ref{alg:pi-no-rho-D} and \ref{alg:G2-pi-no-rho-D}, there is a stable quadruple~$q'$ with
  $|\shifts(P, q')| \leq d$ such that $\pi(q') \in D$ and there is a path (possibly of length zero) from
  $\pi(q')$ to $\rho(q)$ in $G_1(P)$. Since $\rho(q) \in S$,
  thus also $\pi(q') \in S$. By line~\ref{alg:G2-pi-no-rho-D}, $\rho(q')$ does not exist.
  Thus, by
  \cref{lem:characterization-two-rotations-stable}~\eqref{lem:pi-ex-in+rho-nex-or-nin->nstable}
  $M \notin \sm(P[\shifts(P, q')])$, a contradiction to~$M$ being
  $d$-robust. Thus, indeed a $d$-robust matching is returned if there
  is one.
 
  The running time of $O(n^4)$ can be obtained as follows. By
  \cref{prop:rotations+stable-matchings}, the rotation digraph in
  line~\ref{alg:graph-start} can be computed in $O(n^2)$~time. Lines~\ref{alg:sq-start} and \ref{alg:sq-end} can be
  carried out in $O(n^4)$~time by
  \cref{claim:stable-quadruples-rotations-runtime}. Thus, clearly,
  lines~\ref{alg:G1-start} to~\ref{alg:G1-end} can be carried out in $O(n^4)$~time. Lines~\ref{alg:pi-no-rho-D}--\ref{alg:G2-pi-no-rho-D} can
  be done in~$O(n^4)$  %
  because $G(P)$ contains $O(n^2)$ vertices. %
  Analogously,  lines \ref{alg:A}--\ref{alg:last} take $O(n^4)$ time.
\end{proof}

\begin{example}\label{ex:alg}
  To illustrate \cref{alg:robust}, consider the profile~$P$ given in \cref{ex:intro_example} and let $d=1$.
  $P$ admits three rotations, $\pi_1=((u_1,w_2)$, $(u_2, w_3)$, $(u_3, w_1)$, $(u_4, w_1))$, $\pi_2=((u_1, w_3), (u_3, w_1))$, and $\pi_3=((u_2, w_4), (u_4, w_2))$.
  There are ten stable quadruples for $d=1$.
  The corresponding $\pi(q)$ and $\rho(q)$ are summarized in the lower left table.

  \columnratio{.45}

  \begin{paracol}{2}

  {\centering
  \begin{tabular}{@{}c|cc@{}}
    \toprule
    Stable quadruple~$q$  & & \\
    with $|\shifts(P, q)| \le 1$ & $\pi(q)$ & $\rho(q)$ \\
    \midrule
    $(u_1, w_1, u_3, w_3)$ & $\pi_1$ & $\pi_2$\\
    $(u_1, w_2, u_4, w_1)$ & $\pi_1$ & $\pi_3$\\
    $(u_1, w_2, u_4, w_3)$ & $\pi_1$ & $\pi_3$\\
    $(u_1, w_3, u_2, w_2)$ & no & $\pi_1$\\
    $(u_1, w_4, u_2, w_1)$ & $\pi_2$ & $\pi_3$\\
    $(u_2, w_1, u_3, w_2)$ & $\pi_3$ & $\pi_2$\\
    $(u_2, w_2, u_4, w_4)$ & $\pi_1$ & $\pi_3$\\
    $(u_2, w_4, u_3, w_3)$ & no & $\pi_1$\\
    $(u_3, w_1, u_4, w_4)$ & no  & $\pi_1$\\
    $(u_3, w_2, u_4, w_3)$ & $\pi_2$ & $\pi_3$\\
    \bottomrule
  \end{tabular}\par}

\switchcolumn

  \noindent The digraphs~$G(P)$ and $G_1(P) = G_2(P)$ constructed in \cref{alg:robust} are depicted in the lower right figure.

{\centering  \begin{tikzpicture}[>=stealth]
    \node[draw, inner sep=1pt, circle] (p1) {$\pi_1$};
    \node[draw, inner sep=1pt,  circle, below  = 10pt of p1, xshift=-4ex] (p2) {$\pi_2$};
    \node[draw, inner sep=1pt,  circle, below = 10pt of p1, xshift=4ex] (p3) {$\pi_3$};
    \node[above = .5ex of p1] {$G(P):$};
    \draw[bend right=20] (p1) edge[->] (p2);
    \draw[bend left=20] (p1) edge[->] (p3);

    \begin{scope}[xshift=20ex]
      \node[draw, inner sep=1pt, circle] (p1) {$\pi_1$};
    \node[draw, inner sep=1pt,  circle, below  = 10pt of p1, xshift=-4ex] (p2) {$\pi_2$};
    \node[draw, inner sep=1pt,  circle, below = 10pt of p1, xshift=4ex] (p3) {$\pi_3$};
    \node[above = .5ex of p1] {$G_1(P)=G_2(P):$};
    \draw[bend right=20] (p1) edge[->] (p2);
    \draw[bend left=20] (p1) edge[->] (p3);

    \draw[bend right=20] (p2) edge[->, winered] (p1);
    \draw[bend right=20] (p2) edge[->, winered] (p3);
    \draw[bend right=20] (p3) edge[->, winered] (p2);
    \draw[bend left=20] (p3) edge[->, winered] (p1);
  \end{scope}
\end{tikzpicture}
\par}

\noindent One can verify that $A=\{\pi_1\}$ (see rows 4, 8, 9 in the table).
  $T=\{\pi_1, \pi_2, \pi_3\}$ is the only closed set in $G_2$ that includes $\pi_1$, which corresponds to $M_2$.
  Indeed our algorithm will return $M_2$ (see \cref{ex:intro_example}) as the only $1$-robust matching. 

\hfill$\diamond$\end{paracol}
\end{example}

\looseness=-1 Now we turn to the problem variants where we look for a perfect $d$-robust matching or one with minimum egalitarian cost. Our polynomial-time algorithm for these variants builds on a Linear Programming (LP) formulation which finds a stable matching.
This LP formulation in turn is based on the one-to-one correspondence between the stable matchings and the closed subsets of the rotation poset~\cite[Chapter~3.7]{GusfieldIrving1989}. %
A crucial property of this formulation is that its constraint matrix is totally unimodular.
Hence, each extreme point of the polytope defined by this formulation is integral.

The LP formulation is as follows.
Let $P_0$ be a preference profile with two disjoint sets, $U$ and $W$, each containing $n$~agents.
Let $R(P)$ be the set of rotations for $P_0$ and let $G(P_0)$ with arc set~$E(P_0)$ be the rotation digraph of ${P_0}$ regarding the precedence relation~$\pred^{P_0}$; by \cref{prop:rotations+stable-matchings}\eqref{rot:runtime}, both the rotation set~$R({P_0})$ and the rotation digraph~$G({P_0})$ can be computed in $O(n^2)$~time.
For each rotation~$\rho\in R({P_0})$, we introduce a variable~$x_{\rho}$ with box constraints~$0\le x_{\rho} \le 1$,
where $x_{\rho}=1$ will correspond to adding $\rho$ to the solution subset while $x_{\rho}=0$ means that $\rho$ will not be taken into the subset.
By \citeauthor{GusfieldIrving1989}~\cite[Chapter~3.7]{GusfieldIrving1989}, the constraint matrix of the constraints
\begin{alignat}{2}
  \tag{LP1} x_{\rho} - x_{\pi} &\le 0, &\qquad&  \forall \pi, \rho \in R({P_0}) \text{ with } (\pi,\rho) \in E({P_0}),\label{eq:closed-subset-cons}\\
  \tag{LP2}  \label{eq:box-cons}0\le x_{\rho} &\le 1,                             &&  \forall \rho \in R({P_0}), 
\end{alignat}
\noindent is totally unimodular and thus there is a solution in which each variable takes either value zero or one.
In this way, the set $S=\{\rho \mid x_{\rho}=1\}$, defined by including exactly those rotations whose variable values are set to one is closed under the rotation poset and thus defines a stable matching.

\toappendix{
Before we state our main result for the \pRMl{} problem,
we recall a condition which ensures that an LP formulation gives an integral solution.

\begin{proposition}[\cite{Camion1965}]\label[proposition]{prop:TUM}
  If $A\in \{-1,0,+1\}^{\hat{n}\times \hat{m}}$ and $b\in \mathds{Z}^{\hat{m}}$
  such that each row in $A$ has at most one $+1$ and at most one $-1$,
  then $A$ is totally unimodular,
  and every extreme point of the system $A x \le b$, $x\in \mathds{N}_{0}^{\hat{m}}$ is integral.
\end{proposition}
}

\begin{theorem}\label{cor:d-robust-perfect-egal-poly}
  Finding a $d$-robust perfect matching and finding a $d$-robust
  matching with minimum egalitarian cost, if they exist, can both be
  done in polynomial time.
\end{theorem}
\begin{proof}
  Following \cref{lem:characterization-two-rotations-stable},
  we will add some additional constraints to the LP given by (\ref{eq:closed-subset-cons}) and (\ref{eq:box-cons}),
  which results in an LP whose constraint matrix remains totally unimodular
  \ifshort (see~\cite{Camion1965}).
  \else
  (see~\cref{prop:TUM}).
  \fi  
  To determine whether there is a $d$-robust matching for our instance, we need to consider every possible profile that differs from the original profile by at most $d$~swaps.
  For $2n$~agents, there are, however, $(2n)^{O(d)}$ such profiles.
  To avoid this, we characterize these profiles by stable quadruples, using \cref{lem:d-swaps-profiles}.
  To achieve this, we compute for each stable quadruple~$q$ with $|\shifts(P_0,q)|\le d$ the two specific rotations~$\pi(q)$ and $\rho(q)$ as defined in \cref{def:two-specific-rots}.

As already discussed, $\pi(q)$ and $\rho(q)$ may not exist.
If they exist, then by \cref{lem:the-two-unique} they are unique.
Moreover, by \cref{lem:characterization-two-rotations-stable}\eqref{lem:rho-nex+pi-nex}, we may assume that at least one of $\pi(q)$ and $\rho(q)$ exist
as otherwise $\sm(P)\cap \sm(P[\shifts(P,q)])=\emptyset$,
implying that $P$ does not admit a $d$-robust matching.
We distinguish between three cases, in each case describing how to add some constraints to the LP defined above.
\begin{flalign}
  \textbf{Case (1): Both \boldmath$\pi(q)$ and $\rho(q)$ exist.}
 \text{ Add the constraint }
     \tag{LP3.1} x_{\pi(q)} - x_{\rho(q)} \le 0.&& \label{eq:both-ex}
\end{flalign}
\noindent  
  By \cref{lem:characterization-two-rotations-stable}, statements \eqref{lem:rho-exists-in->stable}, \eqref{lem:pi-ex-nin->stable}, and \eqref{lem:pi-ex-in+rho-nex-or-nin->nstable},
  the stable matching defined according to a closed subset is stable in $P[\shifts(P_0,q)]$ if and only if $x_{\rho(q)}=1$ or $x_{\pi(q)}=0$.
  \begin{flalign}
    \textbf{Case (2): \boldmath$\pi(q)$ exists but $\rho(q)$ does not.}
    \text{ Add the constraint }
  \tag{LP3.2}
    x_{\pi(q)} =0 &&\label{eq:pi-ex-rho-not}
  \end{flalign}
  \noindent
  The above constraint is justified by  \cref{lem:characterization-two-rotations-stable}\eqref{lem:pi-ex-in+rho-nex-or-nin->nstable}.
  \begin{flalign}
    \textbf{Case (3): \boldmath$\pi(q)$ does not exist but $\rho(q)$ exists.}
   \text{ Add the constraint }
   \tag{LP3.3}
    x_{\rho(q)} =1 && \label{eq:pi-ex-pi-not}
  \end{flalign}
  This constraint is justified by \cref{lem:characterization-two-rotations-stable}, statements~\eqref{lem:pi-nex+rho-ex-nin->nstable} and \eqref{lem:rho-exists-in->stable}.

Note that in each of the three cases, we add to the constraint matrix a row which has at most one $+1$, at most one $-1$ and the remaining values are all $0$s.
Thus, we can infer by
\ifshort \citet{Camion1965}
\else \cref{prop:TUM}
\fi that the resulting constraint matrix is still totally unimodular and all primal solutions to our problem are integral.
Since the matrix has $O(n^4)$ rows and $O(n^2)$ columns,
solving the thus constructed LP can be done in polynomial time.

Since all stable matchings match the same set of agents
(\cref{prop:SMI-matched-agents-the-same}), it is apparent from the
above LP that finding a $d$-robust and perfect matching if it exists
can be done in polynomial time. Finding a $d$-robust matching, if one
exists, with minimum egalitarian cost can also be done in polynomial
time by the following: For each rotation $\rho$ we can compute how
adding $\rho$ to a stable matching changes its egalitarian score.
Then, it is sufficient to add an appropriate optimization objective to
the LP constructed above.
\end{proof}

\section{Robustness and Preferences with Ties: NP-hardness}\label{sec:Robust+Ties}
\appendixsection{sec:Robust+Ties}
When the input preferences may contain ties, we consider a swap to be a pair of two agents that belong to two neighboring tied classes.
\toappendix{
  For the case with ties, the preference list~$\succeq$ of each agent may be expressed as a transitive and complete binary relation on the set of the agents who she finds acceptable.
  The expression~``$x \succeq_i y$'' means that $i$ weakly prefers $x$ over $y$ (\emph{i.e.}\ $x$ is better or as good as $y$). We use $\succ_i$ to denote the asymmetric part (\emph{i.e.}\ $x\succeq_i y$ and $\neg (y\succeq_i x)$)
and $\sim_i$ to denote the symmetric part of $\succeq_i$ (\emph{i.e.}\ $x\succeq_i y$ and $y \succeq_i x$).

Formally, we define the swap distance~$\tau(\succeq_i, \succeq_{i'})$ between two preference lists with ties as follows.
\begin{align*}
  \delta(\succeq_i, \succeq_{i'}) \coloneqq
  \begin{cases}
    \infty, & \text{ if } \succeq_i \text{ and } \succeq_{i'} \text{ have }\\
    & \text{\emph{different} acceptable sets,}\\
    |\{ (x,y)\in \succ_i  \mid (y,x) \in \succeq_{i'}\}| +
    |\{(x,y) \in \sim_{i} \mid (x,y) \notin \sim_{i}\}| , & \text{otherwise.}\\
  \end{cases}
\end{align*}

In particular, if an agent has a preference list $(a, b,c)$, meaning that
all three agents are tied on the first position,
then moving to the list $c \succ (a, b)$ requires two swaps (swapping $a$ with $c$ and $b$ with $c$).
}
Finding a stable matching can be done in $O(n^2)$ time even when ties are present~\cite{Irving1994}. In contrast, presence of ties makes \pRMl{} NP-hard:

\newcommand{\robustties}{
\pRMl{} with ties is NP-hard even when the number~$d$ of swaps allowed is one.
}
\begin{theorem}%
  \label{thm:robust-ties-np-hard-d-unbounded}
  \robustties
  \end{theorem}

\appendixproofwithstatement{thm:robust-ties-np-hard-d-unbounded}{\robustties}{
\begin{proof}%
  We reduce from \textsc{Independent Set}.
  Let $I=(G, k)$ be an instance of \textsc{Independent Set}.
  Further, let $V(G)=\{v_1,\ldots, v_n\}$ and $E(G)=\{e_1,\ldots, e_m\}$ denote the set of vertices and the set of edges in $G$ respectively.
  Without loss of generality, we assume that the number~$n$ of vertices in $V$ is at least three and the solution size~$k$ is at least two.
  We construct an instance of \pRMl{} with two sets of agents,~$U$ and $W$, and with the number of allowed swaps equal to $d=1$. %

  \myparagraph{Agent set\boldmath~$U$.} This set consists of the following $n+2m+2$~agents.
  \begin{align*}
    U & = V \cup E \cup F \cup A, \text{ where}\\
    V & = \{v_1,\ldots, v_n\}, \\
    E & = \{e_1,\ldots, e_m\}, \\
    F & = \{f_1,\ldots, f_m\}, \\
    A & = \{a_0,a_{1}\}.
  \end{align*}

  \myparagraph{Agent set\boldmath~$W$.} This set consists of the following $n+2m+2$~agents.
  \begin{align*}
    W & = T \cup S \cup E_V \cup B, \text{ where}\\
    T & = \{t_1,\ldots,t_{n-k}\}, \\
    S & = \{s_1,\ldots, s_k\}, \\
    E_V & = \{e^{v_i}_{\ell}, e^{v_j}_{\ell} \mid e_{\ell} \in E(G) \text { with } e_{\ell}=\{v_i,v_j\}\}, \\
    B & = \{b_0, b_{1}\}.
  \end{align*}
  Note that we use $v_i$ (resp.\ $e_\ell$) for both a vertex and its corresponding vertex agent (resp.\ an edge and its corresponding edge agent).
  It will, however, be clear from the context what we are referring to.
  
  The preference lists of these agents are defined as follows, where $[\star]$ means that the elements in~$\star$ are ranked in an arbitrary but fixed order, while~$(\star)$ means that the elements in~$\star$ are tied.
  The symbol~$\ldots$ at the end of each preference list denotes an arbitrary but fixed order of the remaining not mentioned agents.

  \myparagraph{Preference lists of the agents from \boldmath$U$.}
  \begin{alignat*}{3}
    \forall i \in [n] & \quad  &v_i &\colon (T) \succ b_0 \succ [\{e_{\ell}^{v_i} \mid e_\ell \in E(G) \text{ such that } v_i \in e_{\ell} \}] \succ (S)  \succ \ldots, \\
    \forall \ell \in [m] \text{ with } e_\ell=\{v_i,v_j\} & &e_{\ell} &\colon (e^{v_i}_{\ell}, e^{v_j}_{\ell}, b_0) \succ \ldots , \\
    \forall \ell \in [m] \text{ with } e_\ell=\{v_i,v_j\} \text{ and } i < j& &f_{\ell} &\colon e^{v_j}_{\ell}  \succ e^{v_i}_{\ell} \succ b_0 \succ \ldots,  \\
    & &a_0 &\colon b_0 \succ b_{1}  \succ \ldots,\\
    & &a_1 &\colon b_1 \succ b_{0}  \succ \ldots. 
  \end{alignat*}

  \myparagraph{Preference lists of the agents in \boldmath$W$.}  
  \begin{alignat*}{3}
    \forall i \in [n-k] & \quad  & t_i  &\colon (V) \succ a_0 \succ a_1 \succ \ldots, \\
    \forall i \in [k] & & s_i & \colon (V) \succ a_0 \succ a_1 \succ \ldots, \\
    \forall \ell \in [m] \text{ with } e_\ell=\{v_i,v_j\} & &e^{v_i}_{\ell} &\colon e_{\ell} \succ a_0 \succ f_\ell \succ v_{i} \succ \ldots ,\\
    & &e^{v_j}_{\ell} &\colon e_{\ell} \succ a_0 \succ f_\ell \succ v_{j} \succ \ldots , \\
      & &b_0 &\colon a_0 \succ a_1 \succ \ldots, \\
      & &b_1 &\colon a_1 \succ a_{0} \succ \ldots.
   \end{alignat*}

   We use $P$ to denote the above preference profile.
   Before we prove the correctness of our construction, we observe some properties which every stable matching must satisfy.

   \begin{claim}\label[claim]{claim:prop-robust-ties}
     Every stable matching~$M$ for $U$ and $W$ with respect to the initial preferences must satisfy the following.
     \begin{compactenum}
     \item Each agent~$t_i \in T$ must be matched with an agent from $V$.
     \item Each agent~$a_j \in A$ must be matched with $b_j$. 
     \item For each edge~$e_\ell \in E(G)$ with $e_\ell =\{v_i, v_j\}$, the agents~$e_\ell$ and $f_\ell$ must have $\{M(e_\ell), M(f_\ell)\} = \{e_\ell^{v_i}, e_{\ell}^{v_j}\}$.
     \item Each agent~$s_i \in S$ must be matched with an agent from $V$.
   \end{compactenum}
   \end{claim}
   \begin{proof}
     \renewcommand{\qedsymbol}{(of
       \cref{claim:prop-robust-ties})~$\diamond$}

     The first statement is straight-forward because every vertex agent~$v_i \in V$ and every selector agent~$t_j\in T$ rank each other at the first position.

     Analogously, we obtain the second statement.

     Now, consider an arbitrary edge~$e_{\ell}\in E(G)$ and let $v_i$ and $v_j$ denote the endpoints of edge~$e_{\ell}$.
     Since agent~$a_0$ is already matched with $b_0$
     we can neglect them from the preference lists of $e_{\ell}, f_{\ell}, e_{\ell}^{v_i}$, and $e_{\ell}^{v_j}$.
     Consequently, one can verify that the partners of $e_{\ell}$ and $f_{\ell}$ must be from $\{e_{\ell}^{v_i}, e_{\ell}^{v_j}\}$.

    By the first three statement, there are $k$ agents left from $V$ who each must be matched with some agent from $S$ because every agent from $S$ ranks every agent from $V$ at the first position.
   \end{proof}

   We show that $G$ admits an independent set of size $k$ if and only if the profile~$P$ has a stable matching~$M$ that remains stable in each profile~$P'$ that differs from $P$ by at most one swap.

   For the ``only if'' direction, let $V'=\{v_{i_1},\ldots, v_{i_k}\}$ be an independent set of $k$ vertices with $i_1 < i_2 < \ldots < i_{k}$.
   For the sake of convenience, let $V\setminus V' = \{v_{j_1}, \ldots, v_{j_{n-k}}\}$ with
   $j_1 < j_2 < \ldots < j_{n-k}$.
   
   We claim that the following perfect matching~$M$ is stable in every profile which differs from the original one by at most one swap.
   \begin{align*}
     M = {}& \{\{a_0,b_0\}, \{a_1,b_1\}\} \cup \{\{\{s_{r}, v_{i_r}\} \mid v_{i_r} \in V'\}\} \cup{}\\
         & \{\{t_{r}, v_{j_r}\} \mid v_{j_r} \in V\setminus V'\} \cup{}\\
     & \{\{e_\ell, e^{v_i}_\ell\}, \{f_\ell, e^{v_j}\} \mid  e_\ell =\{v_i,v_j\} \text{ for some } e_\ell \in E(G) \text{ and } v_i \in V'\}\cup{}\\
     & \{\{e_\ell, e^{v_i}_\ell\}, \{f_\ell, e^{v_j}\} \mid  e_\ell =\{v_i,v_j\} \text{ for some } e_\ell \in E(G) \text{ and } \{v_i, v_j\} \cap V' = \emptyset \text{ and } i < j\}.
   \end{align*}
   Note that the partners of $S$ and $T$ can be of arbitrary order, and that the partners of the agents~$e_\ell$ and $f_\ell$ for which none of the endpoints of the corresponding edge~$e_\ell$ are in the independent set~$S$ can also be swapped.
   We fix this order for the sake of the simplicity of the reasoning.

   One can verify that $M$ is stable in the original profile~$P$.
   To see why it remains stable in every profile, denoted as $P'$, that differs from the original one by at most one swap,
   we observe the following.
   \begin{compactenum}
     \item No agent~$a_j$ from $A$ is involved in a blocking pair because for each agent~$y\neq b_j = M(a_j)$  other than $a_j$'s partner~$b_j$ it holds that
     if $a_j$ shall prefer $y$ to $b_j$ in $P'$,
     then this agent~$y$ must be $b_{1-j}$.
     However, agent~$b_{1-j}$ will still prefer her partner~$a_{1-j}$ to $a_j$ because $P$ and $P'$ differ by only one swap.

     \item Analogously, no agent~$b_j$ from $B$ is involved in a blocking pair.

     \item No agent~$z$ from $S\cup T$ is involved in a blocking pair because of the following.
     The partner of $z$ is an agent from $V$. 
     For each agent~$y \neq M(z)$ other than $z$'s partner~$M(z)$
     if agents~$z$ and $y$ would form a blocking pair in $P'$ then $z$ would have preferred $y$ to an agent from $M(z)$.
     To achieve this, however, 
     we need at least $|V|-1$ swaps, which is more than two.

     \item Analogously, no agent~$e_{\ell}$ from $E$ is involved in a blocking pair because of the following.
     The partner~$M(e_{\ell})$ of $e_{\ell}$ is an agent from $\{e_\ell^{v_i}, e_{\ell}^{v_{j}}\}$ with $e_{\ell}=\{v_i,v_j\}$. 
     For each agent~$y \neq M(e_{\ell})$ other than $e_{\ell}$'s partner~$M(e_{\ell})$
     if $e_{\ell}$ would form with $y$ a blocking pair in $P'$ then she must have preferred $y$ to
     her partner~$M(e_{\ell})$.
     To achieve this, we would need at least $2$ swaps.

     \item No agent~$f_{\ell}$ from $F$ is involved in a blocking pair because of the following.
     The partner of $f_{\ell}$ is an agent from $\{e_{\ell}^{v_i}, e_{\ell}^{v_j}\}$ with $e=\{v_i, v_j\}$.
     For each agent~$y\neq M(f_{\ell})$ other than $f_\ell$'s partner~$M(f_{\ell})$
     if agent~$f_{\ell}$ would form with $y$ a blocking pair of $M$,
     then she must have preferred $y$ to $M(f_{\ell})$.
     By the preference list of $f_{\ell}$ and by the definition of $M$ this agent~$y$ must be an agent~$e_\ell^{v}$ with $v\in e$ and $M(e_{\ell}^v)=e_{\ell}$.
     However, agent~$y$ will still prefer her partner~$e_\ell$ to $f_{\ell}$
     as the swap distance between $e_{\ell}$ and $f_{\ell}$ in the initial preference list of $e_{\ell}^{v}$ is two.
     
     \item Agents~$v_i$ and $e_{\ell}^{v_j}$ with $v_i \neq v_j$ \emph{cannot} form a blocking pair of $M$ in $P'$ because
     the swap distance between $e_{\ell}^{v_j}$ and $M(v_i)\in S\cup T$ in the initial preference list of $v_i$ is at least $|S|$, which is more than one.
   \end{compactenum}    
 
   Thus, the only possible blocking pairs would involve $v_i$ and $e_{\ell}^{v_i}$ with
   $v_i \in e_\ell$ for some $e_{\ell}\in E(G)$.
   Suppose, for the sake of contradiction, that $\{v_i, e_{\ell}^{v_i}\}$ is blocking $M$ in $P'$.
   This implies that $M(v_i)\in S$ and $M(e^{v_i}_\ell) = f_{\ell}$.
   However, by construction, $M(v_i)\in S$ implies that $v_i \in V'$ and $M(e^{v_i}_\ell) = f_{\ell}$ implies that $v_i \notin V'$---a contradiction.

   For the ``if'' direction, let $M$ be a stable matching that is stable in every profile that differs from the original one by one swap.
   We claim that $V'=\{v_i \mid M(v_i)\in S\}$ is an independent set of size $k$.
   Obviously $V'$ has $k$ vertices by our observation above that every agent~$v_i$ is matched to an agent that is either from $T$ or from $S$,
   and $|S|=k$.
   To show that $V'$ is an independent set, suppose for the sake of contradiction
   that $V'$ contains two adjacent vertices~$v_i$ and $v_j$ and let $e_{\ell}=\{v_i, v_j\}$ be the incident edge.
   This means that $M(v_i),M(v_j)\in S$.
   Then, consider the profile~$P'$ that differs from $P$ by one swap in the preference list of $e_\ell^{v_i}$, depicted as follows:
   \begin{align*}
     e_{\ell}^{v_i}\colon e_{\ell} \succ a_0 \succ v_i \succ f_\ell.
   \end{align*}
   Since $v_i$ prefers agent~$e_{\ell}^{v_i}$ to its partner which is from $S$,
   the stability of $M$ implies that $M(e^{v_i}_\ell)= e_\ell$; recall that $M(e^{v_i}_\ell)$ cannot be matched to $a_0$ as reasoned before.
   Consequently, $M(e^{v_j}_{\ell})= f_\ell$.
   However, consider the profile~$P''$ that differs from $P$ by one swap in the preference list of $e_\ell^{v_j}$, depicted as follows:
   \begin{align*}
     e_{\ell}^{v_j}\colon e_{\ell} \succ a_0 \succ v_j \succ f_\ell.
   \end{align*}
   Since $v_i$ prefers agent~$e_{\ell}^{v_j}$ to its partner which is from $S$,
   it follows that $M$ is not stable in $P''$ as we have just reasoned that $M(e_{\ell}^{v_j})=f_\ell$ but $v_j$ prefers $v_j$ to $f_\ell$ in $P''$--a contradiction.
 \end{proof}
 }

\section{\NStable Matchings}
We now present our results on the complexity of finding \nstable matchings which are perfect or within a given \egalcostn{} bound.
We start in \cref{sec:nstable-approx} by observing that all four problems variants of \nstability{} are NP-hard.
Indeed, we provide a stronger result, which says that under the standard complexity theory
assumption P${}\neq{}$NP the minimization variants of all considered problems do not admit a polynomial-time polynomial-factor approximation algorithm. In \cref{sec:nstable-param} we study the influence of the number of allowed swaps on the complexity of the problem variants.

\subsection{Classical and Approximation Hardness}\label{sec:nstable-approx}
\looseness=-1
\appendixsection{sec:nstable-approx}
To show hardness, we will focus on the so-called gap variants of our problems, and prove that these gap variants are NP-hard.
These gap problems can be solved by the corresponding approximation algorithms %
so that their NP-hardness will rule out polynomial-time approximation algorithms for our problem.
Loosely speaking, an $\alpha$-gap variant of some minimization problem~$\mathcal{Q}$
has, as input, a specific instance~$I$ of $\mathcal{Q}$ and a cost upper bound~$q \in \mathbb{N}$ %
and asks whether \begin{inparaenum}
  \item $I$ admits a solution of cost at most~$q$, or
  \item each solution for $I$ has cost at least $\alpha\! \cdot\! q$
  \end{inparaenum}%
  (without requirement on the answer when the optimum solution is in the ``gap'' interval $(q, \alpha\! \cdot\! q)$).
Note that, to decide between these two options we can use a factor-$\alpha$ approximation algorithm (if it exists), an algorithm that is guaranteed to find a solution of cost at most~$\alpha \!\cdot\! \opt$ where $\opt$ is the minimum cost. Hence, if the $\alpha$-gap problem is NP-hard, a polynomial-time factor-$\alpha$ approximation algorithm implies P${}={}$NP. 

\toappendix{
To make the presentation easier, we use the following decision-focused definition of approximation algorithms, which only solve the gap variant of an optimization problem. By the reasoning \ifshort in the main part, \else above, \fi
ruling out the existence of such approximation algorithms also rules out the existence the standard form approximation algorithms which produce solutions.
\begin{definition}\label[definition]{def:p-approx}
  Let $\poly\colon \mathds{N} \to \mathds{N}$ be a polynomial whose domains and co-domains are on the positive integers.
  An algorithm $\mathcal{A}$ is a \myemph{polynomial-time and $\poly$-approximation algorithm for \pGNSPMl~(\pGNSPM)}
  if
  for each preference profile~$P$ and each positive integer~$\globald \in \mathds{N}$, the algorithm~$\mathcal{A}$ runs in time~$|P|^{O(1)}$ and satisfies the following:
  \begin{inparaenum}[(1)]
    \item  if $P$ admits a \gns{$\globald$} and perfect matching,
    then $\mathcal{A}$ returns ``yes'',
    and
    \item if $P$ admits no \gns{$\poly(\globald) \cdot \globald$} and perfect matching,
    then $\mathcal{A}$ returns~``no''.
  \end{inparaenum}
  
  If such an algorithm exists,
  then we also say that \pGNSPM{} admits a \myemph{polynomial-time and polynomial-factor approximation algorithm}.
\end{definition}

An approximation algorithm for \pLNSPMl~(\pLNSPM) is defined analogously. For the variant where an additional objective is to achieve a given \egalcostn, we use an even weaker notion that allows for bi-criteria approximation.

\begin{definition}\label[definition]{def:p1-p2-approx}
  Let $\polyone, \polytwo\colon \mathds{N} \to \mathds{N}$ be two polynomials whose domains and co-domains are on the positive integers.
  An algorithm $\mathcal{A}$ is a \myemph{polynomial-time and $(\polyone,\polytwo)$-approximation algorithm for \pGNSEMl~(\pGNSEM)}
  if for each preference profile~$P$ and each positive integer~$\globald \in \mathds{N}$, the algorithm~$\mathcal{A}$ runs in time~$|P|^{O(1)}$ and satisfies the following:
  \begin{inparaenum}[(1)]
    \item  if $P$ admits a \gns{$\globald$} matching with \egalcostn{} at most $\egalcost$,
    then $\mathcal{A}$ returns ``yes'',
    and
    \item if $P$ admits no \gns{$\polyone(\globald) \cdot \globald$} matching with \egalcostn{} at most $\polytwo(\egalcost)\cdot \egalcost$,
    then $\mathcal{A}$ returns ``no''.
    \end{inparaenum}

    If such an algorithm exists,
    then we also say that \pGNSEM{} admits a \myemph{polynomial-time and polynomial-factor approximation algorithm}.
\end{definition}

Here, again an approximation algorithm for \pLNSEMl~(\pLNSEM) is defined analogously.
}

\newcommand{\nstableinapprox}{%
  For each $\Pi \in \{$\pGNSEM, \pGNSPM, \pLNSPM, \pLNSEM$\}$,
  $\Pi$ is NP-hard, and does not admit a polynomial-time polynomial-factor approximation algorithms, unless P${}={}$NP.
  For \pLNSPM{} and \pLNSEM{}, the statement holds even if $\locald = 1$.
}
\begin{theorem}\label[theorem]{thm:nearly-stable-inapproximable}
  \nstableinapprox
\end{theorem}
\begin{proof}
  The NP-hardness will follow from the inapproximability results by setting the corresponding approximation factors to 1.
  Thus, we only need to show the inapproximability results,
  which are based on the same basic construction.
   We first give the details of the construction. Then, we prove that, on the instances resulting from the construction, approximability of \pGNSPM, \pGNSEM, \pLNSPM, or \pLNSEM\ implies polynomial-time solvability of all NP-complete problems.

  Let $\polyone,\polytwo\colon \mathds{N} \to \mathds{N}$ be two arbitrary polynomials.
  We will show non-existence of any polynomial-time and $\poly$\nobreakdash-approximation algorithm, using a reduction which introduces a gap in the near stability between an optimally nearly stable solution and any other nearly stable~solution.

  We reduce from \textsc{Independent Set}, which has, as input,
  an undirected graph~$G$ with vertex set~$V(G)$ and edge set~$E(G)$ and a positive integer~$k \in \mathds{N}$,
  and asks whether $G$ contains an \emph{independent set} of size~$k$, i.e.\ a $k$-vertex subset of~$V'\subseteq V(G)$ of pairwise non-adjacent vertices.
  Let $I=(G, k)$ be an instance of \textsc{Independent Set}.
  Let $V(G)=\{v_1,\ldots, v_n\}$ and $E(G)=\{e_1,\ldots, e_m\}$ denote the set of vertices and the set of edges in $G$, respectively.
  We interpret the edges as two-element subsets of $V(G)$.
  For each vertex $v_i \in V$, by $E(v_i)$ we denote the set of edges incident with vertex~$v_i$ in $G$.

  From $G$ we will construct a preference profile~$P$. The lower thresholds for the gap problems we are constructing are as follows.
  We define threshold for the number of swaps for a \gnsnopa{} matching as $\globald\coloneqq m+n$,
  the threshold for the number of swaps per agent of a \lnsnopa{} matching as $\locald\coloneqq 1$, 
  and the threshold for the egalitarian cost of a \gns{$\globald$} matching as $\egalcost \coloneqq k+ (\globaldapproxbound + \localdapproxbound+ 2)\cdot (3m+(2n+k)\cdot k + (2n-k)\cdot(n-k))$.
  For ease of notation, let 
 $\globaldapprox\coloneqq\globaldapproxbound+\localdapproxbound+1$, and $\egalcostapprox\coloneqq\egalcostapproxform$.

 \looseness=-1

  \proofparagraph{Construction.} We construct a profile $P$ as follows. We introduce the following disjoint sets of agents: $V, T, E, F$
  (men); $W, S, R, E_V$ (women); and two disjoint sets~$A\cup B$ and $C \cup D$ of auxiliary agents.
  Sets $V$ and $W$ will represent the vertices
  of~$G$, sets $R$, $S$, and~$T$ will force a selection of $k$
  vertices, and sets $E$, $E_V$, and~$F$ will ensure that the selected
  vertices are pairwise nonadjacent.  
  The auxiliary agents from $A \cup B$ enforce that only swaps of some specific agents are relevant
  while the auxiliary agents from $C \cup D$ require that each matching within some appropriate \egalcostn\ must be perfect.

  \myparagraph{The non-auxiliary agents.}
  Specifically, the non-auxiliary sets contain the following agents: 
  \ifshort
  $V $~$\coloneqq$ $\{v_i \mid v_i \in V(G)\}$,
  $T$  $\coloneqq$ $\{t_i \mid v_i \in V(G)\}$,
  $E$  $\coloneqq$  $\{e_{\ell} \mid e_\ell \in E(G)\}$, 
  $F$  $\coloneqq$  $\{f_{\ell} \mid e_\ell \in E(G)\}$,  
  $W$ $\coloneqq$  $\{w_i \mid v_i \in V(G)\}$,
  $S$  $\coloneqq$  $\{s_i \mid i \in [k]\}$, 
  $R$  $ \coloneqq$  $ \{r_i \mid i \in [n - k]\}$, and
  $E_V$  $\coloneqq$  $\{e_\ell^{v_i}, e_{\ell}^{v_j} \mid e_\ell = \{v_i,v_j\} \text{ for some edge~}e_\ell \in E(G)\}$.
 \else
 
  {\centering
    \begin{tabular}{l@{\,}l@{\,}ll@{\,}l@{\,}l}
      $V $ & $\coloneqq$ & $\{v_i \mid v_i \in V(G)\}$, \quad & $W$ & $\coloneqq$ & $\{w_i \mid v_i \in V(G)\}$, \\
      $T$ & $\coloneqq$ & $\{t_i \mid v_i \in V(G)\}$, & $S$ & $\coloneqq$ & $\{s_i \mid i \in [k]\}$, \\
           &&       & $R$ & $ \coloneqq$ & $ \{r_i \mid i \in [n - k]\}$, \\
      $E$ & $\coloneqq$ & $\{e_{\ell} \mid e_\ell \in E(G)\}$, &  $E_V$ & $\coloneqq$ & $\{e_\ell^{v_i}, e_{\ell}^{v_j} \mid e_\ell = \{v_i,v_j\} \text{ for some edge~}e_\ell \in E(G)\}$, and \\
      $F$ & $\coloneqq$ & $\{f_{\ell} \mid e_\ell \in E(G)\}$.
    \end{tabular}
   \par }

 \fi 
 \noindent Note that we use $v_i$ (resp.\ $e_\ell$) for both a vertex and its corresponding vertex agent (resp.\ an edge and its corresponding edge agent).
  It will, however, be clear from the context what we are referring to.
  The preference lists of the above agents are defined as follows (men are placed on the left and women on the right). For the sake of readability the non-auxiliary agents are omitted in each list, and we will describe them in detail later on.

  { \begin{tabular}{l@{\quad}l@{\quad}l}
      $\forall v_i \in V(G) \colon$ &
             $v_i \colon w_i \succ [\{e_\ell^{v_i} \mid e_\ell \in E(v_i)\}] \succ s_1  \succ \ldots \succ s_k$,
              & $w_i \colon v_i \succ t_i$, \\
       & $t_i \colon w_i \succ r_1 \succ \ldots \succ r_{n - k}$, & \\
      $\forall j \in [n - k] \colon $ & & $r_j \colon t_1 \succ \ldots \succ t_n$, \\
      $\forall j \in [k] \colon $ & & $s_j \colon v_1 \succ \ldots \succ v_n$, \\
                  \end{tabular}
                  
                  \begin{tabular}{l@{\quad}ll}
                    {$\forall e_\ell = \{v_i, v_j\} \in E(G) $ with $i < j\colon$}  %
       & $e_{\ell} \colon e_{\ell}^{v_i} \succ e_{\ell}^{v_j}$, \qquad & $e_{\ell}^{v_i} \colon e_{\ell} \succ v_i \succ f_{\ell}$,\\
       & $f_{\ell} \colon e_{\ell}^{v_i} \succ e_{\ell}^{v_j}$, & $e_{\ell}^{v_j} \colon e_{\ell} \succ v_j \succ f_{\ell}$.
    \end{tabular}
  }

  \myparagraph{The type-one auxiliary agents\boldmath~$A\cup B$.}  
  These agents ensure that only the swaps from the agents of $W\cup E$ are relevant.  We
  say that the auxiliary agents in $A \cup B$ are of \emph{type one}.
  For each agent~$x$ from $V\cup T \cup F \cup S \cup R \cup E_V$ and for each two consecutive agents~$y_1$ and $y_2$ in $x$'s preference list (as described above),
  we introduce $\globaldapprox$~auxiliary agents~$a_x^1(y_1,y_2),\ldots, a_x^{\globaldapprox}(y_1,y_2)$ to $A$
  and $\globaldapprox$~auxiliary agents~$b_x^1(y_1,y_2),\ldots, b_x^{\globaldapprox}(y_1,y_2)$ to $B$ with the following preference lists:
  \begin{inparaenum}[(i)]
    \item If $x\in V\cup T \cup F$, then for all $i \in [\globaldapprox]$ let the preference lists of $a^i_x(y_1,y_2)$ and $b^i_x(y_1,y_2)$ be
    $a^{i}_x(y_1,y_2)$ and    
    $a_x^i(y_1,y_2) \succ x$, respectively,
      and add all $\globaldapprox$~auxiliary agents~$b_x^{i}(y_1,y_2)$ between agents~$y_1$ and $y_2$ in the preference list of $x$.
      \item Otherwise, meaning that $x\in R\cup S \cup E_V$, then for all $i \in [\globaldapprox]$  let the preference lists of $a^i_x(y_1,y_2)$ and $b^i_x(y_1,y_2)$ be
      $ a^{i}_x(y_1,y_2) \succ x$ and    
      $ a_x^i(y_1,y_2)$, respectively,
      and add all $\globaldapprox$~auxiliary agents~$a_x^{i}(y_1,y_2)$ between agents~$y_1$ and $y_2$ in the preference list of~$x$.
  \end{inparaenum}
  In total, we have

  {\centering
    $|A|=|B|=\globaldapprox\cdot \big(\sum_{i=1}^{n}(|E(v_i)|+k)+n\cdot (n-k) + m + (n-k)\cdot (n-1)+k\cdot (n-1)+2\cdot 2m\big)= \globaldapprox\cdot (2n^2-n+7m).$\par
  }

  \myparagraph{Type-two auxiliary agents\boldmath~$C \cup D$.}
  To enforce that every matching within \egalcostn~$\egalcostapproxbound$ %
  must be perfect, we introduce type-two auxiliary agents and append them to the preference list of each
  non-auxiliary agent and each type-one auxiliary agent.
  Formally, for each agent~$x \in V\cup T \cup E \cup F \cup A \cup W \cup S \cup R \cup E_V \cup B$,
  we introduce $\egalcostapprox$~auxiliary agents~$C_x\coloneqq \{c_{x}^1$, $\ldots, c_x^{\egalcostapprox}\}$ and add them to $C$, and $\egalcostapprox$~auxiliary agents~$D_x\{d_{x}^1$, $\ldots, d_x^{\egalcostapprox}\}$ and add them to $D$; recall that $\egalcostapprox = \egalcostapproxform$.
  The preference lists of these agents are as follows:
   \begin{inparaenum}[(i)] %
    \item If $x\in V \cup T \cup E \cup F \cup A$, then for all $i \in [\egalcostapprox]$ let the preference lists of $c^i_x$ and $d^i_x$ be
    $ d^{i}_x \succ [D_x\setminus \{d^i_x\}]$ and    
    $ c^{i}_x \succ [C_x\setminus \{c^i_x\}] \succ x$, respectively,
      and append all $\egalcostapprox$~auxiliary agents~$d_x^{i}$ to the end of the preference list of $x$.
      \item   Otherwise, that is, $x\in W \cup S \cup R \cup E_V \cup B$, then for all $i \in [\egalcostapprox]$ let the preference lists of $c^i_x$ and $d^i_x$ be
      $ d^{i}_x \succ [D_x\setminus \{d^i_x\}]\succ x$ and    
      $ c^{i}_x \succ [C_x\setminus \{c^i_x\}]$, respectively,
      and append all $\egalcostapprox$~auxiliary agents~$c_x^{i}$ to the end of the preference list of $x$.
  \end{inparaenum}
  In total, we have $|C|=|D|=\egalcostapprox\cdot (|V|+|T|+|E|+|F|+|A|+|W|+|R|+|S|+|E_V|+|B|)$.
  Observe that every matching with \egalcostn{} at most $\egalcostapproxbound$ must assign to every type-two auxiliary agent a partner that is also of type two, as otherwise the \egalcostn{} induced by such two agents would be at least $\egalcostapprox-1=\egalcostapproxbound+1 > \egalcostapproxbound$.

  This completes the construction of the profile~$P$. Clearly, it can
  be constructed in polynomial~time. 
  \appendixcorrectnessproofwithstatement%
  {The correctness proof can be found in \cref{proof:thm:nearly-stable-inapproximable}}%
  {thm:nearly-stable-inapproximable}{\nstableinapprox}{
  \proofparagraph{Correctness of the construction.}
  In the following, we show that the existence of any polynomial-time $\polyone$-factor approximation algorithm for \pGNSPM\ or \pLNSPM\
  or  any polynomial-time $(\polyone,\polytwo)$-factor approximation algorithm for \pGNSEM\ or \pLNSEM\  implies P${}={}$NP.
  More precisely, we claim the following.
  \begin{claim}\label[claim]{claim:IS<->GNSPM}
    \begin{compactenum}[(1)]
      \item\label{IS->GLNSPEM} If $G$ admits a $k$-vertex independent set, then $P$ admits a \gns{$\globald$} and perfect matching, which is also \lns{$\locald$} and has \egalcostn{} at most $\egalcost$.
      \item\label{GNSPM->IS} If $P$ admits a \gns{$\globaldapproxbound$} and perfect matching,
      then $G$ admits a $k$-vertex independent set.
      \item\label{LNSPM->IS}  If $P$ admits a \lns{$\localdapproxbound$} and perfect matching,
      then $G$ admits a $k$-vertex independent set.
    \end{compactenum}
  \end{claim}

  \begin{proof}\renewcommand{\qedsymbol}{(of
      \cref{claim:IS<->GNSPM})~$\diamond$}
    To show the first statement, assume that $V^\star \subseteq V$ is a
  $k$-vertex independent set of $G$. Construct a perfect matching~$M$ for $P$ as
  follows. Let $i_1 < i_2 < \ldots < i_k$ be the indices of the
  vertices in~$V^\star$, that is, for each $z \in [k]$ we have
  $v_{i_z} \in V^\star$. Similarly, let
  $j_1 < j_2 < \ldots < j_{n - k}$ be the indices of the vertices
  in~$V \setminus V^\star$. Matching~$M$ contains the following pairs.
  \begin{compactenum}[(i)]
  \item For each $z \in [k]$,
  match $\{v_{i_z}, s_z\} \in M$ and $\{t_{i_z}, w_{i_z}\} \in M$.

  By the construction of the preference lists,  the \egalcostn{} of the pair~$\{v_{i_z},s_z\}$ is at most $(\globaldapprox+1)\cdot (n+k+k)$; recall that, for each agent~$x$ who is not an auxiliary
    agent and who is not in $W \cup E$, we have placed exactly $(\globaldapprox + 1)$
    type-one auxiliary agents between each pair of non-auxiliary agents in $x$'s preference list.
    The \egalcostn{} of the pair $\{t_{i_z}, w_{i_z}\}$ is one.
  In total, these $2\cdot k$~pairs contribute at most $k+(\globaldapprox+1)\cdot (n+2\cdot k)\cdot k$~units to the \egalcostn{}.
  \item For each $z \in [n - k]$, match $\{v_{j_z}, w_{j_z}\} \in M$
    and $\{t_{j_z}, r_{z}\} \in M$.

    The \egalcostn{} of the pair~$\{v_{j_z}, w_{j_z}\}$ is zero,
    while the \egalcostn{} of the pair~$\{t_{j_z}, r_z\}$ is at most $(\globaldapprox+1)\cdot (n-k+n)$.
 In total, these $2\cdot (n-k)$~pairs contribute at most $(\globaldapprox+1)\cdot (n-k+n)\cdot (n-k)$~units to the \egalcostn{}.
    \item Further, for each edge~$e_\ell \in E(G)$, choose an endpoint from $e_\ell\cap V^\star$ or an arbitrary endpoint of $e_{\ell}$ if
    $e_\ell \cap V^\star = \emptyset$. Say we have picked an endpoint
    with index~$i$. Then, match
    $\{e_{\ell}, e_{\ell}^{v_i}\}, \{f_{\ell}, e_{\ell}^{v_j}\} \in M$, where
    $j$ is the index of the other endpoint of~$e_\ell$, different
    from~$i$.

    The \egalcostn{} of these two pairs is at most $(\globaldapprox+1)\cdot 3$.
    \item Finally, for each type-one auxiliary agent~$a^z_x(y_1,y_2)\in A$,
    match her with its counter-part from~$B$, that is, match $\{a_x^z(y_1,y_2), b_x^z(y_1,y_2)\} \in M$. For each type-two auxiliary agent~$c^z_x \in C$, match it with its counter-part from $D$, that is, match her $\{c_x^z, d_x^z\} \in M$. 

  There is no \egalcostn{} for these pairs.
  \end{compactenum}
  This concludes the definition of~$M$, which is clearly a perfect matching.
  One can verify that the \egalcostn{} of $M$ is at most $k+(\globaldapprox+1)\cdot \big((n+2\cdot k)\cdot k + (2n-k)\cdot (n-k) + 3m\big)=\egalcost$

  It remains to show that $M$ is \gns{$\globald$} and \lns{$\locald$}. We claim
  that after performing the following swaps, indeed, $M$ is stable:
  For each edge agent~$e_\ell \in E$ let $e_{\ell}^{v_i}$ and $e_{\ell}^{v_j}$ be the two vertex agents in its preference lists with $i < j$, 
  if $M(e_\ell)=e_\ell^{v_j}$, then swap the order of these two agents~$e_\ell^{v_i}$ and $e_{\ell}^{v_j}$ in $e_{\ell}$'s preference list; recall that there are no other agents between these two agents.
  For each agent~$w_i$ which is not matched to $v_i$, swap $v_i$ with $t_i$ in $w_i$'s preference list; note that there are exactly $k$ such agents.
  Clearly, each agent has performed at most $1=\locald$ swap and the total number of performed swaps is at most $m + k \le \globald$.
  Denote by~$P'$ the profile that results from these swaps.

  Next, we show that $M$ is stable for~$P'$.
  Clearly, each auxiliary agent receives its most preferred agent, and, hence, no auxiliary agent can be involved in a blocking pair.
  Similarly, each agent~$w_i\in W$ receives its most preferred agent (after the $k$~swaps performed on $W$), and, hence, no agent from $W$ can be involved in a blocking pair.
  Since each edge agent~$e_\ell$ is matched to its most preferred agent, no blocking pair can involve any edge agent~$e_\ell$.
  Since the partner~$M(e_\ell)$ of each edge agent~$e_\ell$ already obtains its most preferred agent, no blocking pair can involve any agent~$M(e_\ell)$.
  Further, since these agents, $M(e_\ell)$, are the only agents
  which may be preferred by any agent~$f_\ell$ to $M(f_\ell)$, no blocking
  pair can involve any~$f_\ell$; recall that we have just reasoned that no auxiliary agent is involved in a blocking pair.
  Furthermore, each vertex agent~$v_i$ whose corresponding vertex does not belong to the independent set, i.e.\ $v_i \in V \setminus V^\star$ is matched to its most preferred agent.
  Similarly, each agent~$t_{z}$ with $M(t_{z})= w_{z}$ cannot be involved in a blocking pair because she already obtains its most preferred agent.

  A potential blocking pair must hence involve an agent from 
  $\{t_{j_z}\mid v_{j_z}\in V\setminus V^{\star}\} \cup V^{\star}$.

  Consider an agent~$t_{j_z}$ with $v_{j_z}\in V \setminus V^{\star}$.
  By the construction of matching~$M$, it follows that $M(t_{j_z})=r_z$.
  Since neither $w_{j_z}$ nor any auxiliary agent can be involved in a blocking pair,
  by the preference list of $t_{j_z}$ it follows that $t_{j_z}$ could only form a blocking pair with an agent~$r_{z'}$ such that $z< z'$.
  However, this agent~$r_{z'}$ prefers its partner~$M(r_{z'})=t_{j_{z'}}$ to $t_{j_z}$.
  Hence, no agent~$t_{j_z}$ can be involved in a blocking pair.
  
  Consider an agent~$v_{i_z}$ which corresponds to a vertex from the independent set~$V^{\star}$.
  By the construction of matching~$M$, it follows that $M(v_{i_z})=s_z$.
  By our reasoning above, $w_{i_z}$ will not form a blocking pair with $v_{i_z}$ as it already obtains its most preferred partner.
  Agent~$v_{i_z}$ prefers agent~$s_{z'}$ to its partner~$s_{z}$ only if $z'< z$.
  However, for each $z' < z$, agent~$s_{z'}$ prefers its partner~$M(s_{z'})=v_{i_{z'}}$ to agent~$v_{i_z}$ (recall that $i_{z'}<i_{z}$).
  Thus, no agent from $S$ will form with $v_{i_z}$ a blocking pair.
  Any blocking pair must thus be of the form
  $\{v_i, e^{v_i}_\ell\}$ where $v_i \in V^{\star}$ 
  and $e^{v_i}_\ell \in E(v_i)$. However, since $v_i \in V^\star$, for each
  of its incident edges, say $e_\ell$, we have matched $e_\ell^{v_i}$ to
  its most preferred agent~$e_\ell$. Thus, indeed, there is no blocking pair, showing that $M$ is \gns{$\globald$} and \lns{$\locald$}.

  For the second statement of \cref{claim:IS<->GNSPM}, assume that $M$ is a \gns{$\globaldapproxbound$} and perfect matching for $P$
  and let $P'$ be a profile that results from~$P$ by making at most $\globaldapproxbound$~swaps such that~$M$ is stable in~$P'$.
  
  Recall that for each agent~$y$ which is either non-auxiliary or an auxiliary agent of type one we have introduced $2\cdot \egalcostapprox$ type-two auxiliary agents, contained in $C_y$ and in $D_y$.
  Observe that either $C_y$ or $D_y$ finds only the agents from the other set acceptable.
  Hence, by the perfectness of $M$, the partners of all agents from $C_y$ are exactly the agents from $D_y$.
  Consequently, we can ignore all type-two auxiliary agents in the preference lists of the remaining agents.
  In particular,  for each pair of  type-one auxiliary agents~$a_x^{z}(y_1,y_2)$ and $b^z_x(y_1,y_2)$, one of them finds only the other agent acceptable (ignoring the type-two auxiliary agents).
  Again, by the perfectness of $M$, we infer that each~$a_x^{z}(y_1,y_2)$ is matched to its counter-part~$b_x^z(y_1,y_2)$.
  Hence, from now on, when discussing the partners of a non-auxiliary agent, we only need to consider the non-auxiliary agents in its preference~list.
  
  Recall that there are at least~$\globaldapprox$ type-one auxiliary agents between each non-auxiliary agents in the preference list of each agent from $V\cup R\cup S\cup T\cup E_v\cup F$.
  By the fact that $\globaldapprox=\globaldapproxbound+\localdapproxbound+1$, it is impossible to perform $\globaldapproxbound$~swaps so as to switch the positions of two non-auxiliary agents in the preference list of an agent from~$V \cup R \cup S \cup T \cup E_V \cup F$.
  Thus, the only swaps performed to obtain $P'$ are without loss of generality in the preference lists of agents in~$E \cup W$ and are only within the non-auxiliary agents.

  Let
  $V' = \{v_i \in V(G) \mid M(v_i)  \in S\}$.
  We claim that $V'$ is a $k$-vertex independent set in~$G$. First,~$V'$ has cardinality~$k$ because $M$ is perfect and
  the only remaining acceptable partners to every agent in~$S$ are those in~$V$.
  Suppose, for the sake of contradiction, that $V'$ contains two adjacent vertices~$v_i$ and $v_j$ and let $e_{\ell}=\{v_i,v_j\}$ be their incident edge.
  By the definition of $V'$, it follows that both $v_i$ and $v_j$ are assigned partners from $S$.
  Since $v_i$ prefers $e_{\ell}^{v_i}$ to every agent from $S$
  and since $e_\ell^{v_i}$ prefers only $e_\ell$ to $v_i$,
  by the stability of $M$ in $P'$, it follows that $M(e_{\ell}^{v_i})=e_{\ell}$; recall that no swaps are performed in between any two non-auxiliary agents in the preference lists of the agents from $V\cup E_V$.
  Analogously, it must hold that $M(e_{\ell}^{v_j})=e_{\ell}$---a contradiction to $M$ being a matching.

  The reasoning for the third statement is analogous to what we have done for the second statement.
  Instead of arguing about the total number of swaps, we only need to argue that the number of swaps changed per agent is $\localdapproxbound$ which is
  strictly smaller than by~$\globaldapprox$. 
  Thus, it is still impossible to change the positions of two non-auxiliary agents in any the preference list of any non-auxiliary agent from $V\cup R\cup S\cup T\cup E_v\cup F$.
  \end{proof}

  The next claim establishes a close connection between the \egalcostn{} and the perfectness of a matching.

  \begin{claim}\label[claim]{claim:NSPM<->NSEM}
      If $M$ is a matching with \egalcostn{} of at most $\egalcostapproxbound$, then this matching must be perfect.
  \end{claim}

  \begin{proof}\renewcommand{\qedsymbol}{(of
      \cref{claim:NSPM<->NSEM})~$\diamond$}
    Assume that $M$ has \egalcostn{} at most $\egalcostapproxbound$. %
    It is straight-forward to see that this $M$ must be perfect as otherwise the cost of leaving one agent unmatched is equal to the length of this agent's preference list,
    exceeding the budget~$\egalcostapproxbound$ because the length of each agent's preference list is at least $\egalcostapprox>\egalcostapproxbound+1$.
    \end{proof}

  Now, we continue with our correctness proof.
  \myparagraph{Inapproximability of \pGNSPM.} Suppose, for the sake of contradiction, that there exists a $\polyone$-approximation algorithm~$\mathcal{A}$ for \pGNSPM, running in polynomial-time.
  Then, we can use algorithm~$\mathcal{A}$ to decide \textsc{Independent Set} in polynomial time, showing P${}={}$NP.
  Given an arbitrary instance~$I=(G,k)$ of \textsc{Independent Set},
  we construct profile~$P$ and define $\globald$ as described above and let $\mathcal{A}$ run on $(P,\globald)$. 
  If $I$ is a yes-instance, then by the first implication of \cref{claim:IS<->GNSPM}, it follows that $P$ admits a \gns{$\globald$} and perfect matching, and
  by \cref{def:p-approx}, $\mathcal{A}$ return ``yes''.
  If $I$ is a no-instance, then by the contrapositive of the second implication of \cref{claim:IS<->GNSPM},
  it follows that $P$ does not admit a \gns{$\globaldapproxbound$} and perfect matching.
  By \cref{def:p-approx},~$\mathcal{A}$ returns ``no''.

  \myparagraph{Inapproximability of {\pGNSEM}.} Again, suppose, for the sake of contradiction, that there exists a $(\polyone,\polytwo)$-approximation algorithm~$\mathcal{A}$ for \pGNSEM, running in polynomial-time.
  Then, we can use algorithm~$\mathcal{A}$ to decide \textsc{Independent Set} in polynomial time, showing P${}={}$NP, as follows.
  Given an arbitrary instance~$I=(G,k)$ of \textsc{Independent Set},
  we construct profile~$P$ and define $\globald$, $\locald$, and $\egalcost$ as described above and let $\mathcal{A}$ run on $(P,\globald,\egalcost)$. 
  If $I$ is a yes-instance, then by the first implication of \cref{claim:IS<->GNSPM} it follows that $P$ admits a \gns{$\globald$} matching with \egalcostn{} at most $\egalcost$, and
  by \cref{def:p1-p2-approx}, $\mathcal{A}$ return ``yes''.
  If $I$ is a no-instance, then by the contrapositive of the second implication of \cref{claim:IS<->GNSPM} 
  it follows that $P$ does not admit a \gns{$\globaldapproxbound$} perfect matching.
  This implies that $P$ does not admit a \gns{$\globaldapproxbound$} matching with \egalcostn{} at most $\egalcostapproxbound$
  as otherwise, by \cref{claim:NSPM<->NSEM}, we will have a \gns{$\globaldapproxbound$} and perfect matching for $P$---a contradiction. 
  By \cref{def:p1-p2-approx},~$\mathcal{A}$ returns ``no''.
  
  \myparagraph{Inapproximability of \pLNSPM.} Suppose, for the sake of contradiction, that there exists a $\polyone$-approximation algorithm~$\mathcal{A}$ for \pLNSPM, running in polynomial-time.
  Then, we can use algorithm~$\mathcal{A}$ to decide \textsc{Independent Set} in polynomial time, showing P${}={}$NP, as follows.
  Given an arbitrary instance~$I=(G,k)$ of \textsc{Independent Set},
  we construct profile~$P$ and define $\globald$, $\locald$, and $\egalcost$ as described above and let $\mathcal{A}$ run on $(P,\locald)$. 
  If $I$ is a yes-instance, then by the first implication of \cref{claim:IS<->GNSPM},~$P$ also admits a \lns{$\locald$} and perfect matching.
  Thus, by \cref{def:p-approx}, $\mathcal{A}$ return ``yes''.
  If $I$ is a no-instance, then by the contrapositive of the third implication from \cref{claim:IS<->GNSPM},~$P$ does not admit a \lns{$\localdapproxbound$} perfect matching.
  By \cref{def:p-approx},~$\mathcal{A}$ returns~``no''.

  \myparagraph{Inapproximability of \pLNSEM.} Suppose, for the sake of contradiction, that there exists a $(\polyone,\polytwo)$-approximation algorithm~$\mathcal{A}$ for \pLNSEM, running in polynomial-time.
  Then, we can use algorithm~$\mathcal{A}$ to decide \textsc{Independent Set} in polynomial time, showing P${}={}$NP as follows.
  Given an arbitrary instance~$I=(G,k)$ of \textsc{Independent Set},
  we construct profile~$P$ and define $\globald$, $\locald$, and $\egalcost$ as described above and let $\mathcal{A}$ run on $(P, \locald, \egalcost)$. 
  If $I$ is a yes-instance,
  then by the first implication of \cref{claim:IS<->GNSPM} it follows that $P$ admits a \lns{$\locald$} matching with \egalcostn{} at most $\egalcost$, and
  by \cref{def:p1-p2-approx}, $\mathcal{A}$ return ``yes''.
  If $I$ is a no-instance, then by the contrapositive of the third implication from \cref{claim:IS<->GNSPM},~$P$ does not admit a \lns{$\localdapproxbound$} perfect matching.
  By the contrapositive of \cref{claim:NSPM<->NSEM},~$P$ does not admit a \lns{$\localdapproxbound$} matching with \egalcostn{} at most $\egalcost$.
  By \cref{def:p1-p2-approx}, $\mathcal{A}$ returns~``no''.}
\end{proof}

\subsection{Parameterized Complexity}\label{sec:nstable-param}
\appendixsection{sec:nstable-param}
We now investigate the influence of three natural parameters on the complexity
of obtaining nearly stable matchings: ``total number~$\globald$ of swaps'',  ``number~$\unmatched$ of initially unmatched agents'', and ``number~$\matched$ of initially matched agents''; the latter two will be defined below. Note that
\cref{thm:nearly-stable-inapproximable} implies that even only one
swap leaves the problems pertaining to \lnsnopa{}
matchings NP-hard. This is different from the globally nearly stable
variants, for which simple polynomial-time algorithms for a
constant number of swaps exist. However, we show that removing the dependence on
the number of swaps in the exponent in the running time is impossible
unless FPT${}={}$W[1].

\newcommand{\xpswaps}{%
  \pGNSPM\ and \pGNSEM\ are solvable in $n^{O(\globald)}$ time. 
}
\begin{proposition}%
  \label[proposition]{prop:xp-swaps}
  \xpswaps
\end{proposition}

\appendixproofwithstatement{prop:xp-swaps}{\xpswaps}{
\begin{proof}[Proof sketch]
  Iterate over all $\binom{n^2}{\globald}$ possibilities for making
  $\globald$ swaps and check for each of the resulting profiles using
  the well-known polynomial-time
  algorithms~\cite{GaleShapley1962,IrLeGu1987} whether it admits a
  stable matching which is perfect, or satisfies the required bound on
  the \egalcostn.
\end{proof}
}

A substantial improvement on the above rather trivial
$n^{O(\globald)}$-time algorithm would imply a major breakthrough, as the
following theorem shows.

\newcommand{\globalwonehardswaps}{%
  \pGNSPM{} and \pGNSEM{} are W[1]-hard with respect to the number~$\globald$ of swaps.
  Moreover, they both do not admit any $n^{o(\globald)}$-time algorithm unless the Exponential Time Hypothesis fails.
}
\begin{theorem}%
  \label{thm:global-w1hard-swaps}
  \globalwonehardswaps
\end{theorem}

  \appendixproofwithstatement{thm:global-w1hard-swaps}{\globalwonehardswaps}{
  \begin{proof}
  We first show that \pGNSPM{} is W[1]-hard for $\globald$ and refutes $n^{o(\globald)}$-running time algorithms.
  Then, we show how to adapt the proof to show an analogous result for \pGNSEM.
  To show the results for \pGNSPM{}, we provide a polynomial-time
  reduction from the W[1]-complete \textsc{Independent Set} problem,
  parameterized by the solution size~$k$~\cite{CyFoKoLoMaPiPiSa2015},
  and set the parameter to $\globald=2k$.
  Let
  $(G, k)$ be an instance of \textsc{Independent Set} where we seek
  for an independent set of size~$k$ in the $n$-vertex, $m$-edge graph
  $G$, with $V(G) = \{v_1, \ldots, v_n\}$ and
  $E(G) = \{e_1, \ldots, e_m\}$.
  We construct a preference profile~$P$ with two disjoint sets of agents,
  $A$ and $B$, each consisting of five groups and a dummy agent:
  $A \coloneqq T \cup V \cup W \cup E \cup E_Y \cup \{h_1\}$
  and 
  $B \coloneqq S \cup X \cup Y \cup F \cup F_V \cup \{h_2\}$.
  Note that we use the symbols~$v_i$ (resp.\ $e_\ell$) for both vertices and agents (resp.\ for both edges and agents).
  It will, however, be clear from the context what we are referring to when we use them.  
  The two dummy agents~$h_1$ and $h_2$ are used to make performing some swaps non-beneficial.
  
  \myparagraph{Agent sets\boldmath~$T$ and $S$.}
  For each $z \in [k]$, introduce two \emph{selection agents~$t_z$ and $s_z$} add them to
  $T$ and $S$, respectively.
  These agents will be unmatched in every stable matching of~$P$ and matching them will
  force a selection of $k$~vertices from~$G$ into an independent set.
  Their acceptable agents are a subset of the vertex agents that we
  introduce as follows.

  \myparagraph{Agent sets\boldmath~$V$, $W$, $X$, and $Y$.}
  For each vertex $v_i \in V(G)$, introduce four agents~$v_i$, $w_i$, $x_i$, and $y_i$,
  and add them to the sets~$V$, $W$, $X$, and $Y$, respectively.
  For each $i \in [n]$, these agents will form a path
  $v_i, x_i, w_i, y_i$ in the acceptability graph.
  The basic idea is that, in the initial profile, \emph{every} stable matching must match agent~$v_i$ to agent~$x_i$ and agent~$w_i$ to $y_i$.
  As we will see, such matchings are imperfect since the agents from~$S$ are unmatched in every stable matching. 
  To obtain a perfect matching, we must match some selection agent~$s_z \in S$ to some vertex agent~$v_i$, selecting the corresponding vertex~$v_i$ into a solution for the input graph.
  This will incur two swaps to make the resulting matching stable.
  Below, we introduce edge agents and add them to the preference lists of agents~$v_i$ and $y_i$, in order to ensure that the selected vertices induce an
  independent set.

  \myparagraph{Agent sets\boldmath~$E$, $E_{Y}$, $F$, and $F_V$.}
  For each edge~$e_\ell \in E(G)$, denote the endpoints of
  $e_\ell$ by $v_{i}$ and $v_{j}$ such that $i < j$.
  Introduce four \emph{edge agents}~$e_\ell$, $e_{\ell}^{y_j}$, $f_{\ell}$, and $f_{\ell}^{v_i}$,
  and add them to $E$, $E_{Y}$, $F$, and $F_{V}$, respectively.

  The preference lists of the agents are constructed as follows.
  Here, for some set~$Z$, the notation~$[Z]$ means an arbitrary but fixed linear order of~$Z$.
  
  \myparagraph{Preference lists of the agents.}

  \noindent \begin{tabular}{@{}r@{}l@{}r@{}l@{}ll}
    \quad $h_1\colon$& $h_2 \succ [Y]$, & $h_2 \colon$ & $h_1 \succ [V]$,\\[1ex]
    \multicolumn{2}{l}{$\forall z \in [k]$,}\\
     \quad $t_z \colon$ & $y_1 \succ \!\ldots\! \succ y_n$,  &  $s_z \colon$& $v_1  \succ\!\ldots\! \succ v_n$.\\[1ex]
    \multicolumn{2}{l}{$\forall i \in [n]$,}\\
    \quad $v_i \colon$& \multicolumn{3}{@{}l}{${x_i \succ h_2 \succ [\{f^{v_i}_\ell\! \mid \!e_\ell \!=\! \{v_i, v_j\}\!\in\! E(G) \text{ with } i \!<\! j\}] \succ s_1 \succ \!\ldots\! \succ s_k}$,}\\
     \quad $w_i \colon$& $y_i \succ x_i$,   &  $x_i \colon$ & $v_i \succ w_i$,\\[1ex]
    & & $y_i  \colon$& $\mathmakebox[0pt][l]{w_i \succ h_1 \succ [\{e^{y_i}_\ell \!\mid\! e_\ell \!=\! \{v_r, v_i\}\!\in \!E(G) \text{ with } r \!<\! i\}] \succ t_1 \succ \!\ldots\! \succ t_k,}$ \\
     \multicolumn{2}{l}{$\forall \ell \in [m] \text{ with }e_{\ell} = \{v_i, v_j\}\text{ and } i < j$,} \\
     \quad  $e_\ell   \colon$& $f_\ell^{v_i} \succ f_{\ell}$,& $f_{\ell} \colon$ & $e_{\ell} \succ e_\ell^{y_j}$,\\ 
    \quad  $e_\ell^{y_j}   \colon$ &$f_\ell \succ y_j \succ f^{v_i}_{\ell}$, & $f^{v_i}_{\ell} \colon$ & $e^{y_j}_{\ell} \succ v_i \succ e_\ell$.\\
            \end{tabular}

  \myparagraph{Agent sets\boldmath~$T$ and $S$.}
  These agents will be unmatched in every stable matching of~$P$ and matching them will
  force a selection of $k$~vertices from~$G$ into an independent set.
  Their acceptable agents are a subset of the vertex agents that we
  introduce as follows.
  
    \myparagraph{Agent sets\boldmath~$V$, $W$, $X$, and $Y$.}
    For each vertex $v_i \in V(G)$, introduce four agents~$v_i$, $w_i$, $x_i$, and $y_i$,
  and add them to the sets~$V$, $W$, $X$, and $Y$, respectively.
  For each $i \in [n]$, these agents will form a path
  $v_i, x_i, w_i, y_i$ in the acceptability graph.
  The basic idea is that, in the initial profile, \emph{every} stable matching must match agent~$v_i$ to agent~$x_i$ and agent~$w_i$ to $y_i$.
  As we will see, such matchings are imperfect since the agents from~$S$ are unmatched in every stable matching. 
  To obtain a perfect matching, we must match some selection agent~$s_z \in S$ to some vertex agent~$v_i$, selecting the corresponding vertex~$v_i$ into a solution for the input graph.
  This will incur two swaps to make the resulting matching stable.
  Below, we introduce edge agents and add them to the preference lists of agents~$v_i$ and $y_i$, in order to ensure that the selected vertices induce an
  independent set.
  \smallskip
  
\noindent  \cref{fig:global-w1hard-swaps} depicts the crucial part of the induced acceptability graph for an edge~$e_{\ell}=\{v_i, v_j\}\in E(G)$ with $i<j$.
  The weights at both sides of the edges denote the ranks of the respective endpoint towards the other endpoint.

  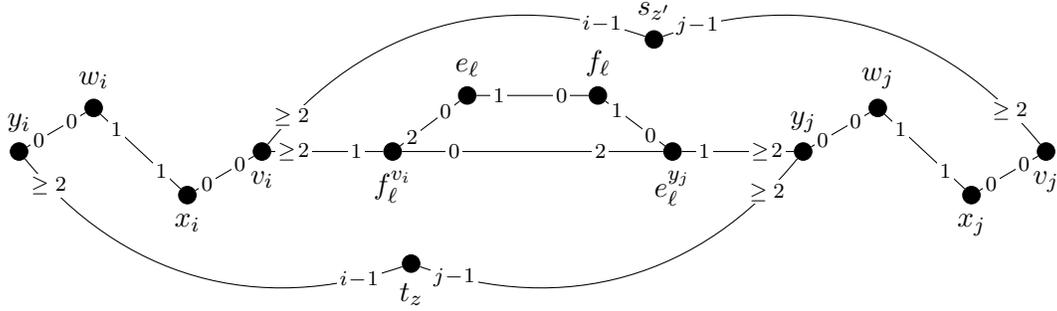
\begin{figure}[t]
    \centering
    \begin{tikzpicture}
      \def \xss {9ex}
      \def \xs {4.5ex}
      \def \ss {3.5ex}
      \node[agent] (yi) {};
      \node[agent, right = 2*\xss  of yi] (vi) {};
      \node[agent, right = \xs  of yi, yshift=\ss] (wi) {};
      \node[agent, left = \xs  of vi, yshift=-\ss] (xi) {};

      \node[agent, right=\xss of vi] (fvi) {};
      \node[agent, right=\xs of fvi, yshift=\xs] (e) {};
      \node[agent, right=\xss of e] (f) {};
      \node[agent, right=\xs of f,yshift=-\xs] (eyj) {};

      \node[agent, right = \xss  of eyj] (yj) {};
      \node[agent, right = 2*\xss  of yj] (vj) {};
      \node[agent, right = \xs  of yj, yshift=\ss] (wj) {};
      \node[agent, left = \xs of vj, yshift=-\ss] (xj) {};

      \foreach \x  in {i, j} {
        \foreach \y / \pos in {y/above,w/above,x/below,v/below} {
          \node[\pos = 0pt of \y\x] {$\y_{\x}$};
        }
      }
      \foreach \x / \y / \pos in {fvi/{f_{\ell}^{v_i}}/below,eyj/{e_{\ell}^{y_j}}/below,e/{e_{\ell}}/above, f/{f_{\ell}}/above} {
          \node[\pos = 0pt of \x] {$\y$};        
        }

        \foreach \i in {i, j} {
          \foreach \s / \t  / \x / \y in {y/w/0/0,w/x/1/1,x/v/0/0}{ 
            \draw (\s\i) edge node[pos=0.2, fill=white, inner sep=1pt] {\scriptsize $\x$}node[pos=0.76,fill=white,inner sep=1pt] {\scriptsize $\y$} (\t\i);
          }
        }

          \foreach \s / \t  / \x / \y in {vi/fvi/{\ge\!2}/1,fvi/e/2/0,fvi/eyj/0/2,e/f/1/0,f/eyj/1/0,eyj/yj/1/{\ge\!2}}{ 
            \draw (\s) edge node[pos=0.2, fill=white, inner sep=1pt] {\scriptsize $\x$}node[pos=0.76,fill=white,inner sep=1pt] {\scriptsize $\y$} (\t);
          }

          \gettikzxy{(yi)}{\yix}{\yiy};
          \gettikzxy{(yj)}{\yjx}{\yjy};
          \node[agent] at (\yix*0.5 + \yjx*0.5, \yiy-\xss) (T) {};
          \node[below = 0pt of T] {$t_z$};
          \draw (yi) edge[bend right = 35] node[pos=0.08, fill=white, inner sep=1pt]  {\scriptsize $\ge 2$} node[pos=0.9, fill=white, inner sep=1pt]  {\scriptsize $i\!-\!1$}(T);
       \draw (T) edge[bend right = 35] node[pos=0.9, fill=white, inner sep=1pt]  {\scriptsize $\ge 2$} node[pos=0.08, fill=white, inner sep=1pt]  {\scriptsize $j\!-\!1$}(yj);

          \gettikzxy{(vi)}{\vix}{\viy};
          \gettikzxy{(vj)}{\vjx}{\vjy};
          \node[agent] at (\vix*0.5 + \vjx*0.5, \viy+\xss) (S) {};
          \node[above = 0pt of S] {$s_{z'}$};
          \draw (vi) edge[bend left = 35] node[pos=0.08, fill=white, inner sep=1pt]  {\scriptsize $\ge 2$} node[pos=0.9, fill=white, inner sep=1pt]  {\scriptsize $i\!-\!1$}(S);
       \draw (S) edge[bend left = 35] node[pos=0.9, fill=white, inner sep=1pt]  {\scriptsize $\ge 2$} node[pos=0.08, fill=white, inner sep=1pt]  {\scriptsize $j\!-\!1$}(vj);

     \end{tikzpicture}
     \caption{Part of the acceptability graph of the profile constructed in the proof of \cref{thm:global-w1hard-swaps}.}
     \label[figure]{fig:global-w1hard-swaps}
  \end{figure}

  To complete the construction, define the total number of swap as $\globald = 2k$.
  Clearly, the construction can be done in polynomial time.
  Now observe that every stable matching from the constructed profile does not match exactly $2k$ agents, namely those from $T\cup S$.
  
  To show that our construction is indeed a parameterized reduction it remains to show that 
  $G$ admits a $k$-vertex independent set if and only if there exists a preference profile~$P'$ with $\tau(P, P')\le d =2k$
  which admits a perfect stable matching~$M$.

  For the ``only if'' part, assume that there exists a $k$-vertex independent
  set~$V' \subseteq V$ in~$G$.
  We define the preference profile~$P'$ by performing the following $\globald=2k$~swaps which involve the agents from $W \cup X$ that correspond to the vertices from the independent set.
  For each $z\in [k]$, swap the two agents~$y_i$ and $x_i$ in the preference list of agent~$w_i$, and
    swap the agents~$v_i$ and $w_i$ in the preference list of agent~$x_i$.

  Now, we construct the following perfect matching~$M$ for $P'$.
  Let $V' = \{v_{i_1}, v_{i_2}, \ldots, v_{i_k}\}$, where
  $i_1 < i_2 < \ldots < i_k$.
  \begin{enumerate}
    \item Put $\{h_1,h_2\}\in M$.
    \item For each~$z \in [k]$, put
  $\{v_{i_z},s_z\}, \{w_{i_z}, x_{i_z}\}, \{t_z, y_{i_z}\} \in M$.
  \item   For each~$i \in [n] \setminus \{i_1, i_2, \ldots, i_k\}$, put
  $\{v_i, x_i\}, \{w_i, y_i\} \in M$.
  \item For each edge~$e_{\ell}\in E(G)$, let $v_i$ and $v_j$ be the two endpoints of edge~$e_\ell$ with $i < j$, and do the following.
  Recall that we have created two agents with the names~$f_\ell^{v_i}$ and $e_{\ell}^{v_j}$.
  If $v_i\in V'$ belongs to the independent set, implying that $v_j \in V\setminus V'$,
  then put $\{e_\ell, f_\ell\}, \{e_{\ell}^{y_i}, f_{\ell}^{v_j}\} \in M$.
  Otherwise, that is, if $v_i \notin V'$, then put $\{e_{\ell}, f_{\ell}^{v_i}\}, \{e_\ell^{y_j}, f_\ell\}\in M$.
\end{enumerate}
  Clearly, $M$ is perfect. We claim that $M$ is also stable for $P'$.

  Suppose, for the sake of contradiction, that $p$ is a blocking pair of $M$ for the profile~$P'$.
  First, observe that $p$ cannot involve $h_1$, $h_2$, $w_{i}$, or $x_{i}$ for any $i \in [n]$ as these agents already obtain their most preferred agents in $P'$.
  For the same reason, $p$ cannot involve any agent~$v_{i}$ for some~$i \in [n] \setminus \{i_1, i_2, \ldots,  i_k\}$. 
  Further,~$p$ cannot involve an agent~$s_z$ or an agent~$t_z$ for any $z \in [k]$ because of the following.
  For each agent~$c$ such that agent~$s_z$ prefers $c$ to~$M(s_z) = v_{i_z}$ we have that $c$ equals either some agent~$v_i$ with $i \in [n] \setminus \{i_1, i_2, \ldots, i_k\}$, 
  which already obtains her most preferred partner~$M(v_i)=x_i$, or $c$ equals some agent~$v_{i_r}$ with $r < z$, which prefers
  her partner $M(v_{i_r})=s_{i_r}$ to agent~$s_{i_z}$.
  Using a similar reasoning, we thus obtain that $p$ can neither involve an agent~$t_z$, $z\in [k]$.
  Moreover, $p$ cannot involve two edge agents which correspond to two different edges or one vertex agent and one edge agent such that the corresponding vertex and edge are not incident to each other.
  Combining all of the above observations, we infer that $p$ either (a) involves two edge agents which correspond to the same edge or (b) a vertex agent and an edge agent such that the corresponding vertex and edge are not incident to each other.

  For Case (a), observe that in each pair of mutually acceptable edge
  agents that correspond to the same edge, one receives its
  most-preferred partner in $M$ with respect to~$P'$. Hence, Case (b) must hold.

  Let $e_{\ell}$ be the edge corresponding to the edge agent involved
  in $p$ and let $v_i$ and $v_j$ be the two endpoints of
  edge~$e_{\ell}$ with $i < j$. Thus, $p=\{e_{\ell}^{y_j}, y_j\}$ or
  $p=\{v_i, f_{\ell}^{v_i}\}$.
  
  If $p=\{e_{\ell}^{y_j}, y_j\}$, by the definition of blocking pairs,
  it follows that $M(e_{\ell}^{y_j})= f^{v_i}_{\ell}$ and
  $M(y_j) = t_z$ for some $z \in [k]$. However, by the definition of
  $M$, from $M(y_j) = t_z$ we infer that $v_j \in V'$ and from $M(e_{\ell}^{y_j})= f^{v_i}_{\ell}$ we infer $v_j \in V \setminus V'$, a contradiction. 

  Analogously, if $p=\{v_i, f_{\ell}^{v_i}\}$, by the definition of
  blocking pairs, it follows that $M(v_i)=s_{z}$ for some~$z\in [k]$
  and that $M(f_{\ell}^{v_i}) = e_\ell$. However, by our definition of
  $M$, from $M(v_i)=s_{z}$ we infer $v_i \in V'$ and from
  $M(f_{\ell}^{v_i}) = e_\ell$ we infer $v_i \in V \setminus V'$, a
  contradiction.

  Hence, indeed, $M$ is stable in $P'$.

  For the ``if'' part, assume that there exists a perfect matching~$M$ for
  profile $P$ and there exists a preference profile~$P'$ with $\tau(P,P')\le d=2k$ such that $M$
  is stable for $P'$.
  By the perfectness of $M$ there are $i_1,\ldots, i_k \in [n]$ such that for each $z\in [k]$
  we have that $M(v_{i_z}) \in S$. 
  We show that the vertex subset~$V' \coloneqq \{v_{i_1}, \ldots, v_{i_k}\}$ is a
  $k$-vertex independent set in~$G$.

  First of all, we claim the following.
  \begin{claim}\label{claim:S-T-same-vertex-subset}
    For each agent~$v_i\in V$ it holds that $M(v_{i}) \in S$ if and only if $M(y_i) \in T$.
    Moreover, no agent in $V\cup Y \cup E_Y \cup F_V$ changes her preference list in $P'$.
  \end{claim}

  \begin{proof}\renewcommand{\qedsymbol}{(of
      \cref{claim:S-T-same-vertex-subset})~$\diamond$}
    We define two subsets~$I_1 \coloneqq\{{i} \in [n] \mid M(v_{i}) \in S\}$
    and $I_2\coloneqq\{{i'}\in[n] \mid M(y_{i'}) \in T\}$.
    To show the first statement, it suffices to show that $I_1=I_2$.
    Clearly, $|I_1|=|I_2|=k$ because $|S|=|T|=k$ and $M$ is a perfect matching.

    If we can show that for each $i\in I_1\cup I_2$ there are exactly two distinct swaps  in
    $P'$ in comparison to~$P$, exactly one swap in $x_i$'s preference list and exactly one swap in $w_i$'s preference list,
    then by the swap budget $\globald=2k$ and by the cardinalities of $I_1$ and $I_2$,
    it follows that $I_1 = I_2$.

    Now, consider an index~$i \in I_1$; the case when $i \in I_2$ is symmetric and omitted.
    By our definition of $I_1$, we have that $M(v_i)\in S$.
    Since $M$ is perfect, it follows that $\{x_i, w_i\} \in M$.
    Since $x_i$ and $v_i$ are each other's most preferred agent in profile~$P$, at least one swap occurs in the preference lists of $x_i$ and $v_i$ to make $M$ stable for $P'$; otherwise they would form a  blocking pair.
    Similarly, since $w_i$ and $y_i$ are each other's most preferred agent in $P'$,
    at least one swap occurs in the preference lists of $w_i$ or $y_i$ to make $M$ stable for $P'$.
    Since $|I_1|=k$, meaning that there are at least $k$ distinct agents in $V$ matched to agents in $S$, the preference lists of $x_i$ and $v_i$ are affected indeed by exactly one swap (otherwise there would be more than $2k$ swaps in total).
    Analogously, exactly one swap occurs in the preference lists of $w_i$ and $y_i$.
    By our swap budget, it follows that only agents in $V \cup X \cup W \cup Y$ may have a different preference list in $P'$ compared to $P$.
    This, in particular, implies that $M(h_1)=M(h_2)$ as otherwise we need at least one more swap to make $M$ stable in $P'$.
    
    Observe that at least one agent, namely~$h_2$, is between $v_i$'s partner~$M(v_i)\in S$ and $x_i$.
    Thus, it takes more than one swap to make $M$ stable in $P'$ if we change the preference list of~$v_i$ and not the preference list of $x_i$.
    Analogously, at least one agent, namely~$h_1$, is between $y_i$'s partner~$M(y_i)$ and $w_i$; recall that we have just reasoned that $M(h_1)=h_2$.
   Thus, it also takes more than one swap to make $M$ stable in $P'$ if we change the preference list of $y_i$ and not that of $w_i$
   Summarizing, to make $M$ stable for $P'$, there is exactly one swap in the preference list of $x_i$
   and there is exactly one swap in the preference list of~$w_i$.

   The second statement follows directly from our swap budget and from the above reasoning that for each $i \in I_1$, there is exactly one swap in the preference list of $x_i$
   and there is exactly one swap in the preference list of~$w_i$ that are performed for $P$ to obtain $P'$.
  \end{proof}

  To show that $V'$ is indeed an independent set, suppose, towards a contradiction,
  that $V'$ contain two adjacent vertices~$v_{i}$ and $v_j$ with $i < j$.
  Let $e_{\ell} = \{v_{i}, v_{j}\}$ be the incident edge. %

  By the first statement in \cref{claim:S-T-same-vertex-subset},
  it follows that $M(y_j)\in T$.
  
  By the second statement in \cref{claim:S-T-same-vertex-subset},
  agent~$v_i$ does not change her preference list in $P'$, meaning that agent~$v_i$ prefers $f^{v_{i}}_\ell$ to her partner~$M(v_i)\in S$ in $P'$.
  By the stability of $M$,
  it follows that $f^{v_i}_{\ell}$ prefers her partner~$M(f^{v_i}_\ell)$ to~$v_i$.
  Again, by the second statement in \cref{claim:S-T-same-vertex-subset},
  agent $f^{v_i}_\ell$ does not change her preference list in $P'$,
  meaning that $f^{v_{i}}_\ell$ prefers only $e^{y_{j}}_\ell$ to $v_{i}$.
  By the stability of $M$, we have
  $M(f^{v_{i}}_\ell)= e^{y_{j}}_\ell$.
  This implies that $\{e^{y_{j}}_\ell, y_{j}\}$ is a blocking pair, because
  $M(y_{j})\in T$ by the above and by the second statement of \cref{claim:S-T-same-vertex-subset}.
  This is a contradiction. Thus, indeed $I'$ is a $k$-vertex independent set.
  The correctness follows.

  The fact that an
  $n^{o(\globald)}$-time algorithm for \pGNSPM\ would contradict the
  Exponential Time Hypothesis follows from this reduction in
  conjunction with the fact that an $n^{o(k)}$-time algorithm for
  \textsc{Independent Set} would contradict the Exponential Time
  Hypothesis~\cite{CyFoKoLoMaPiPiSa2015}.

  \myparagraph{The egalitarian
    case.} %
  To show the desired statements for the egalitarian case, we use the
  same idea of constructing type-two auxiliary agents in the proof
  of \cref{thm:nearly-stable-inapproximable}: We append a sufficiently
  large number~$\Delta$ (to be specified later) of auxiliary agents to
  the end of each preference list which we constructed for the
  perfectness case such that the following conditions are satisfied.
  First, it is possible to match all auxiliary agents in pairs without
  inducing any egalitarian cost. Second, the matching $M$ constructed
  from an independent set as above has egalitarian cost at
  most~$\Delta - 1$. Third, every \gns{$\globald$} matching within
  egalitarian cost at most $\Delta - 1$ must be perfect. As mentioned,
  this can be done using the same construction as for the type-two
  auxiliary agents in \cref{thm:nearly-stable-inapproximable},
  adjusting the number of agents such that we add
  $\Delta \coloneqq (2k+1)\cdot (|T|+|V| + |W|+|E| + |E_Y| + 1) + 2k=
  (2k+1)\cdot (3n+2m+1) + 2k$ agents to the end of each preference
  list. The correctness follows in an analogous way as we prove
  \cref{claim:IS<->GNSPM}\eqref{IS->GLNSPEM}--\eqref{GNSPM->IS} and 
  \cref{claim:NSPM<->NSEM}.
\end{proof}
}

By \cref{prop:SMI-matched-agents-the-same},
the set of unmatched agents is the same across all stable matchings of a given preference profile~$P$. 
We call an agent \emph{initially unmatched} if she is not
contained in any stable matching of the initial profile; otherwise she is \emph{initially matched}.
From
\cref{thm:few-matched-for-swap} it follows that, in order to assign partners to  
the~$\unmatched$ initially unmatched agents, we need to allow for at
least $\globald \geq \unmatched/2$ swaps in \pGNSPM.
The number~$\unmatched$ is thus a smaller parameter than $\globald$, meaning that it could be harder to obtain parameterized tractability result with respect to~$\unmatched$ than to~$\globald$.  Indeed, we obtain intractability for $\unmatched$.

\newcommand{\wonehardunmatched}{%
  \pGNSPM\ and \pGNSEM\ are W[1]-hard with respect to the number of initially unmatched~$\unmatched$ agents. This also holds for \pLNSPM\ and
  \pLNSEM, even if $\locald = 1$.
}
\begin{corollary}%
  \label[corollary]{cor:w1h-unmatched}
  \wonehardunmatched
\end{corollary}
\appendixproofwithstatement{cor:w1h-unmatched}{\wonehardunmatched}{
\begin{proof}
  The reduction in the proof of \cref{thm:nearly-stable-inapproximable} indeed is
  a parameterized reduction with respect to the number of unmatched
  agents which shows W[1]-hardness: First, observe that the reduction
  runs in polynomial time. Second, \textsc{Independent Set} is
  W[1]-hard with respect to the number~$k$ of vertices in the sought
  independent set. Third, the number of unmatched agents in any stable matching of the
  initial profile~$P$ constructed by the reduction is at most $2k$. To see
  this, observe the following. Let $M$ be a stable matching for~$P$.
  For each vertex agent~$v_i \in V$, it must hold that $\{v_i, w_i\} \in M$, since this pair would
  otherwise block~$M$. Thus, for each edge
  $e_\ell = \{v_i, v_j\} \in E(G)$, the edge agent~$e_\ell$ is matched by~$M$ to either
  $e^{v_i}_\ell$ or~$e^{v_j}_\ell$, and $f_\ell$ is matched to
  $\{e^{v_i}_\ell, e^{v_j}_\ell\} \setminus \{M(e_\ell)\}$. Furthermore,
  for each $z\in [n - k]$, agent~$t_z\in T$ is matched to agent~$r_z\in R$,
  saturating~$R$. Hence, the only unmatched agents are the $k$ agents
  in~$\{t_{n-k+1},\ldots,t_{n}\}$ and the $k$ agents in~$S$.
\end{proof}
}

On a side note, it is not hard to obtain a fixed-parameter algorithm
for \pGNSPM\ with respect to the number~$\matched$ of \emph{initially
  matched agents}, that is, the number of agents that occur in every
stable matching of the initial profile.

\newcommand{\globalmatchedfpt}{%
  \pGNSPM\ can be solved in $O(\matched^{\matched^2\cdot \log{(\matched)}}\cdot n^2)$~time and admits a problem kernel with $2\matched$ agents that can be computed in
  linear time. %
}
\begin{proposition}%
  \label[proposition]{prop:GNSPM-matched-agents-fpt}
  \globalmatchedfpt
\end{proposition}

\appendixproofwithstatement{prop:GNSPM-matched-agents-fpt}{\globalmatchedfpt}{
  \begin{proof}[Proof Sketch]

  Let $P$ be a preference profile with agents sets~$U$ and $W$ of size $n$ each. 
  Further let $n_1$ (resp.\ $n_2$) denote the number of agents from $U$ that are initially matched (resp.\ unmatched) under any stable matching of $P$.
  Analogously, we define $n_3$ and $n_4$ for the set $W$, i.e.\ $n_3$ (resp.\ $n_4$) denotes the number of from $W$ that are initially matched (resp.\ unmatched) under any stable matching of $P$.
  By the definition of matching, it follows that
  \begin{align*}
    n_1 = n_3  \text{ and }\   n_2 = n_4.
  \end{align*}
  Hence, together with the definition of $\matched$ and $\unmatched$, it follows that
  \begin{align}\label{eq:matched-unmatched-relation}
    n_1+n_3 = 2n_1 = 2 n_3 = \matched\ \text{ and }\ n_2 + n_4 =2n_2 =2n_4 =\unmatched.
  \end{align}
  Observe that no two initially unmatched agents are acceptable to
  each other since, if they were, then they would form a blocking
  pair. Thus, in each matching that represents a solution to \pGNSPM,
  each initially unmatched agent from $U$ (resp.\ $W$) is partnered with an initially
  matched agent from $W$ (resp.\ $U$).
  Thus, if $n_2 > n_3$ or $n_4 > n_1$, we can immediately
  return no (or a trivial no-instance).
  By \eqref{eq:matched-unmatched-relation},
  we obtain that
  \begin{align}\label{eq:matched-unmatched-U-W}
   |U|= n_1 + n_2 \le n_1+n_3 =  \matched\ \text{ and }\ 
   |W|= n_3 + n_4 \le n_3+n_1 =  \matched.
  \end{align}
  In other words, we obtain a problem kernel with at most $2\matched$ agents. This takes $O(n^2)$ time.

  To solve \pGNSPM{} in $O(2^{\matched^2\cdot \log({\matched})} \cdot n^2)$~time,   
  observe that by \eqref{eq:matched-unmatched-U-W} each agent from $U\cup W$ has at most $\matched$ agents in her preference list, 
  and there are $2\matched$~agents from $U \cup W$.
  We hence iterate through all $(\matched!)^{2{\matched}}$ target profile~$P'$ which differs from $P$ by at most $\globald$ swaps.
  For each of these preference profiles, we check in $O(n^2)$ time whether it admits a stable and perfect matching.
\end{proof}
}
We conclude this section by remarking that the kernelization approach for \cref{prop:GNSPM-matched-agents-fpt} cannot be directly adapted to work for the egalitarian case because not every initially unmatched agent needs to be matched in an optimal egalitarian stable matching. %

\section{Conclusion and Open Questions}

In this paper we have introduced and studied a framework describing the strength of the stability of matchings under preferences.
Our framework unifies and extends some of the few approaches that already exist in the literature, such as additive $\alpha$-stability~\cite{pini_stability_2013,anshelevich_anarchy_2013},
$r$-maximal stability~\cite{DruBou13}, and robustness to the errors in inputs~\cite{MaiVaz2018,MaiVaz2018-birkhoff-arxiv}.
We have elaborated that all these approaches can be expressed by the same model,
where the central idea is to investigate the preference profiles which have bounded distance to the input profile.
Thus, we open up a general framework to study questions that have already received attention in the literature. 

From a computational point of view, we have shown a somehow counter-intuitive relation between strength of stability, and other criteria of social optimality. On the one hand, there exists a polynomial-time algorithm for finding robust matchings if they exist (recall that robustness is a stronger concept than classic stability) even if we additionally aim to reach social optimality. On the other hand, if we ask about \nstability{} instead of robustness, the problem becomes computationally hard in many aspects: it is hard to approximate and hard from the point of view of parameterized complexity. Our computational results are summarized in \Cref{tab:summary}.    

We conclude with some challenges for future research.
First of all, for the case where no $d$-robust matchings exist, we may look for a matching that admits the fewest number of blocking pairs~\cite{BiMaMcD2012,CHSYicalp-par-stable2018} in every profile that has swap distance~$d$ to the input profile. 

Second, continuing our research in \cref{sec:Robust+Ties} where we showed that \pRMl{} becomes NP-hard when ties are allowed, our \nstability{} concept can be generalized to the case with ties. Moreover, both robustness and \nstability{}, though introduced for the bipartite variant (\textsc{Stable Marriage}), can be generalized to the non-bipartite variant (\textsc{Stable Roommates}).
It would be interesting to see whether our algorithmic results transfer to these cases.

Regarding preference restrictions~\cite{BreCheFinNie2017}, it would be interesting to know whether assuming a special preference structure can help in finding tractable cases for nearly stable matchings.

\subsubsection*{Acknowledgments} Piotr Skowron was supported by the Foundation for Polish Science within the Homing programme (Project title: "Normative Comparison of Multiwinner Election Rules"). Jiehua Chen and Manuel Sorge were supported by the European Research Council (ERC) under the European Union’s Horizon 2020 research and innovation programme under grant agreement numbers~677651 (JC) and~714704 (MS).\\ \includegraphics[width=50px]{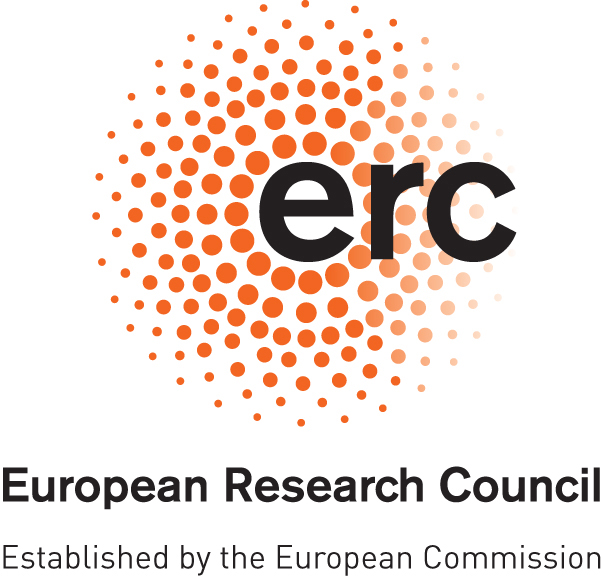}\hspace{.5cm} \includegraphics[width=50px]{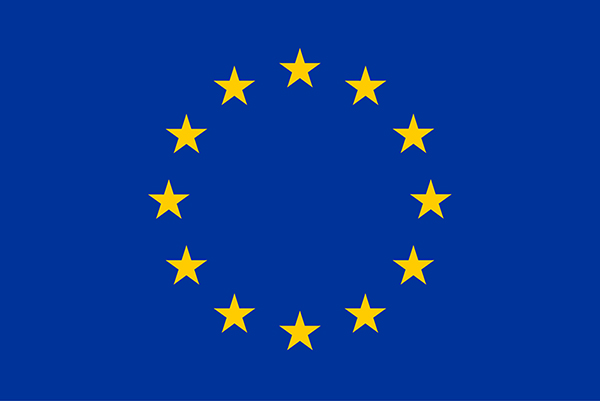}%

\bibliographystyle{abbrvnat}

\bibliography{bib-arxiv}

\end{document}